\documentclass[aos]{imsart}

\RequirePackage[OT1]{fontenc}

\usepackage[table]{xcolor}
\usepackage{multirow}

\usepackage{a4wide}
\usepackage[ruled,linesnumbered]{algorithm2e}
\usepackage{hyperref}
\usepackage{subcaption}{}
\usepackage{enumerate}
\usepackage[normalem]{ulem}
\usepackage{amsmath,amsthm,amssymb,mathrsfs}
\usepackage{stmaryrd}
\usepackage{nicefrac}
\usepackage{mathtools} 
\usepackage[utf8]{inputenc}
\usepackage{amsfonts,dsfont}

\usepackage{makecell}
\usepackage{soul,graphicx}
\usepackage{natbib}
  \usepackage{pgf, tikz}
  \usetikzlibrary{fit,patterns,decorations.pathreplacing}

\startlocaldefs

\usepackage{amsthm}
\newtheorem{theorem}{Theorem}[section]

\newtheorem{corollary}[theorem]{Corollary}
\newtheorem{lemma}[theorem]{Lemma}

\newtheorem{assumption}{Assumption}

\theoremstyle{remark}
\newtheorem{remark}[theorem]{Remark}
\newtheorem{example}[theorem]{Example}
\theoremstyle{definition}
\newtheorem{definition}[theorem]{Definition}

\usepackage[ruled]{algorithm2e}

\newcommand{\R}{\mathbb{R}}

\newcommand{\cG}{\mathcal{G}}

\newcommand{\indep}{\rotatebox[origin=c]{90}{$\models$}}

\renewcommand{\P}{\mathbb{P}}
\newcommand{\E}{\mathbb{E}}

\newcommand{\TDP}{\mathrm{TDP}}
\newcommand{\TDR}{\mathrm{TDR}}
\newcommand{\FDR}{\mathrm{FDR}}

\newcommand{\BH}{\mathrm{BH}}
\newcommand{\FDP}{\mathrm{FDP}}

\newcommand{\mtc}{\mathcal}

\newcommand{\wh}[1]{{\widehat{#1}}}
\newcommand{\ol}[1]{\overline{#1}}
\newcommand{\ind}[1]{{\mathds{1}\left\{#1\right\}}}

\newcommand{\cH}{{\mtc{H}}}

\renewcommand{\l}{\ell}

\newcommand{\rev}[1]{\textcolor{black}{#1}}

\definecolor{mulberry}{rgb}{0.77, 0.29, 0.55}

\newcommand{\stkout}[1]{\ifmmode\text{\sout{\ensuremath{#1}}}\else\sout{#1}\fi}

\endlocaldefs

\newcommand{\lrt}{r}

\DeclareMathOperator*{\argmax}{arg\,max}
\DeclareMathOperator*{\argmin}{arg\,min}

\definecolor{mygreen}{rgb}{0.82, 1.0, 0.82}
\definecolor{myred}{rgb}{ 1.0, 0.84, 0.84}

\begin{document}

\begin{frontmatter}

\title{
Adaptive Novelty Detection with false discovery rate guarantee\\
}
\runtitle{
Adaptive Novelty Detection with FDR guarantee
}


\centerline{Ariane Marandon, Lihua Lei, David Mary and Etienne Roquain} 

\begin{abstract}
  This paper studies the semi-supervised novelty detection problem where a set of ``typical'' measurements is available to the researcher. Motivated by recent advances in multiple testing and conformal inference, we propose AdaDetect, a flexible method that is able to wrap around any probabilistic classification algorithm and control the false discovery rate (FDR) on detected novelties in finite samples without any distributional assumption other than exchangeability. In contrast to classical FDR-controlling procedures that are often committed to a pre-specified $p$-value function, AdaDetect learns the transformation in a data-adaptive manner to focus the power on the directions that distinguish between inliers and outliers. Inspired by the multiple testing literature, we further propose variants of AdaDetect that are adaptive to the proportion of nulls while maintaining the finite-sample FDR control. The methods are illustrated on synthetic datasets and real-world datasets, including an application in astrophysics.
\end{abstract}



\end{frontmatter}


\section{Introduction}

\subsection{Novelty detection}\label{sec:setting}

In this paper, we consider a novelty detection problem (see, e.g., \cite{BLS2010} and references therein) where we observe:
\begin{itemize}
\item a null training sample (NTS hereafter) $Y = (Y_1,\dots,Y_n)$ of ``typical" measurements where $Y_i$s share a common marginal distribution $P_0$ which we refer to as the null distribution;
\item and a test sample $X = (X_1,\dots,X_m)$ of ``unlabeled" measurements for which the marginal distribution of $X_i$ is denoted by $P_i$, which might be different from $P_0$.
\end{itemize}
These measurements are assumed to take values in a general space $\mathcal{Z}$ endowed with a prescribed $\sigma$-field. For example, the space can be the set of real matrices ($\mathcal{Z}=\R^{d\times d'}$) or real vectors ($\mathcal{Z}=\R^d$), whose dimension is potentially large.

Putting two samples together, we observe
$
Z = (Z_1,\dots,Z_{n+m})=(Y_1,\dots,Y_{n},X_1,\dots,X_m).
$
The aim is to detect novelties, namely $X_i$s with $P_i\neq P_0$.
This task is illustrated in Figure~\ref{fig:intro} on a classical image dataset, where we want to detect hand-written digit `$9$'s in the test sample based on an NTS of digits `$4$'. 
The procedure, that declares as novelties the images with red boxes, can make false discoveries (digit `$4$') and true discoveries (digit `$9$'). 

To avoid false positives that might be costly in practice, we seek to control the false discovery rate (FDR), defined as the average proportion of errors among the discoveries, while attempting to maximize the true discovery rate (TDR), defined as the average portion of detected novelties. 
FDR has been a very popular criterion in multiple testing and exploratory analysis since its introduction by \cite{BH1995}; see \cite{benjamini2010discovering} for a detailed discussion and  \cite{BC2015,Slope2015,javanmard2019false,Barber2020robust,ma2021global} for recent developments, among others.

\begin{figure}[t!]
\begin{center}
\begin{tabular}{ccc}
 \hspace{-1.5cm}
Null sample&\hspace{-3cm}
Test sample&\hspace{-2cm} 
Procedure
 \vspace{-.5cm}
\\
\hspace{-1.5cm}
\includegraphics[scale=0.2]{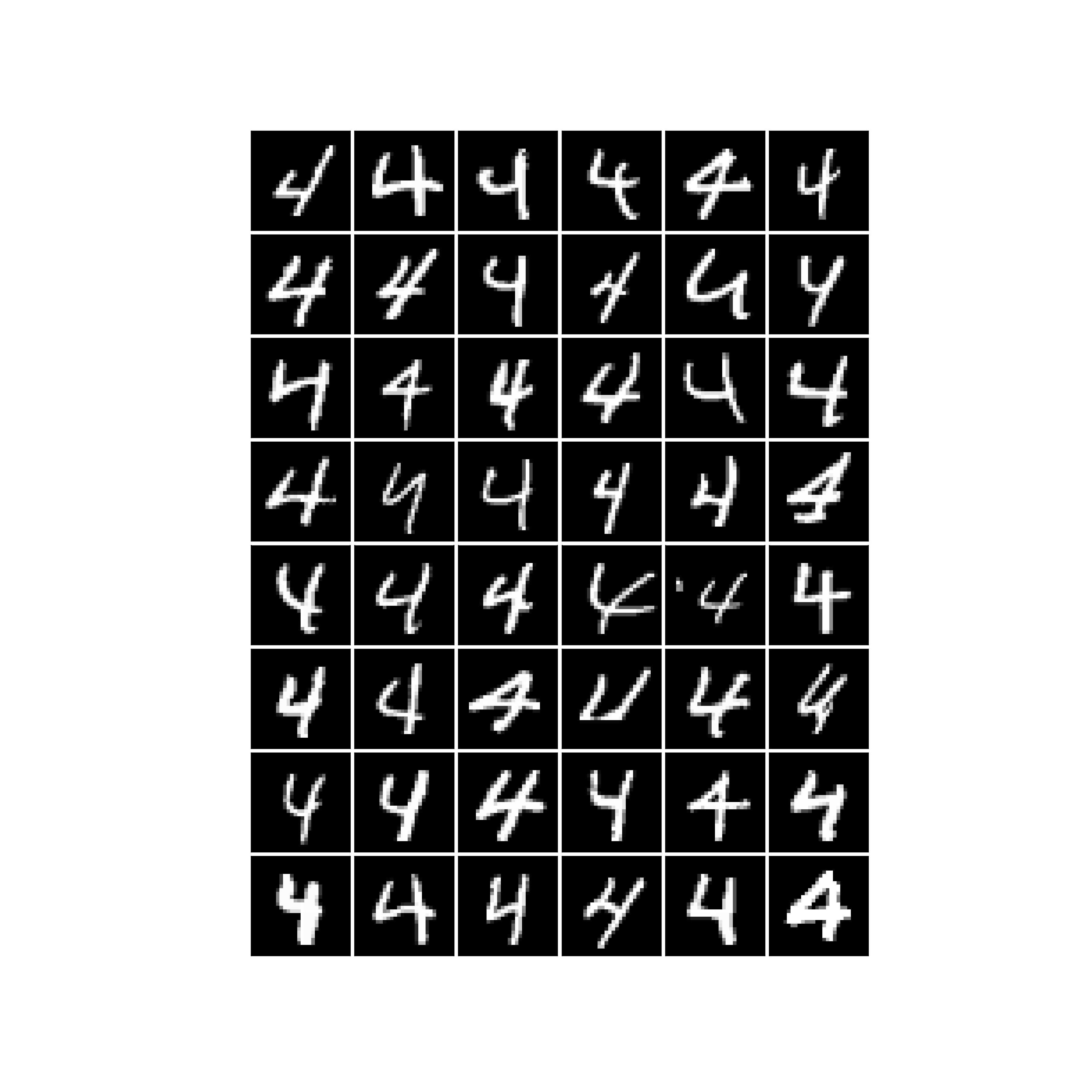}&\hspace{-3cm}
\includegraphics[scale=0.2]{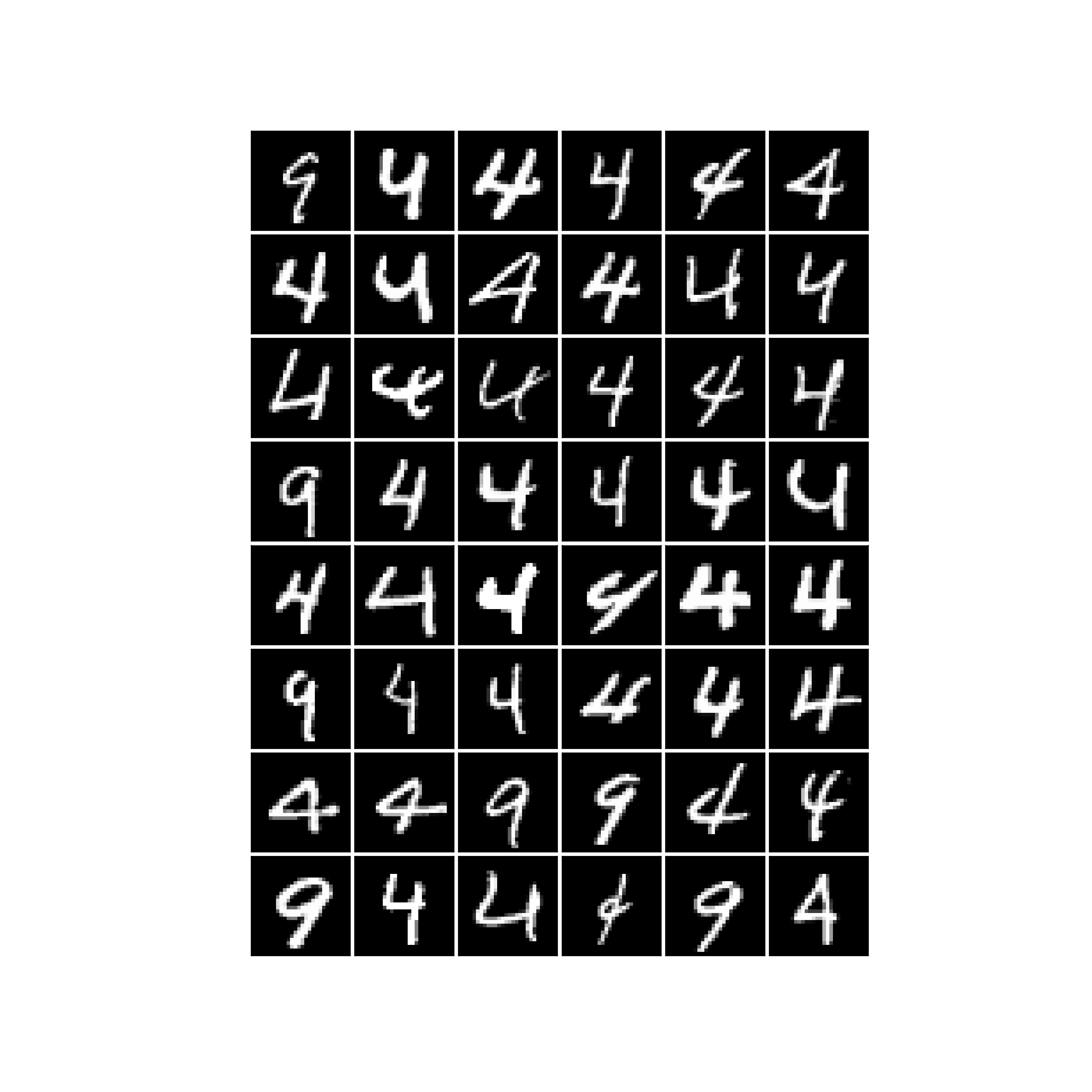}&\hspace{-2cm} 
\includegraphics[scale=0.2]{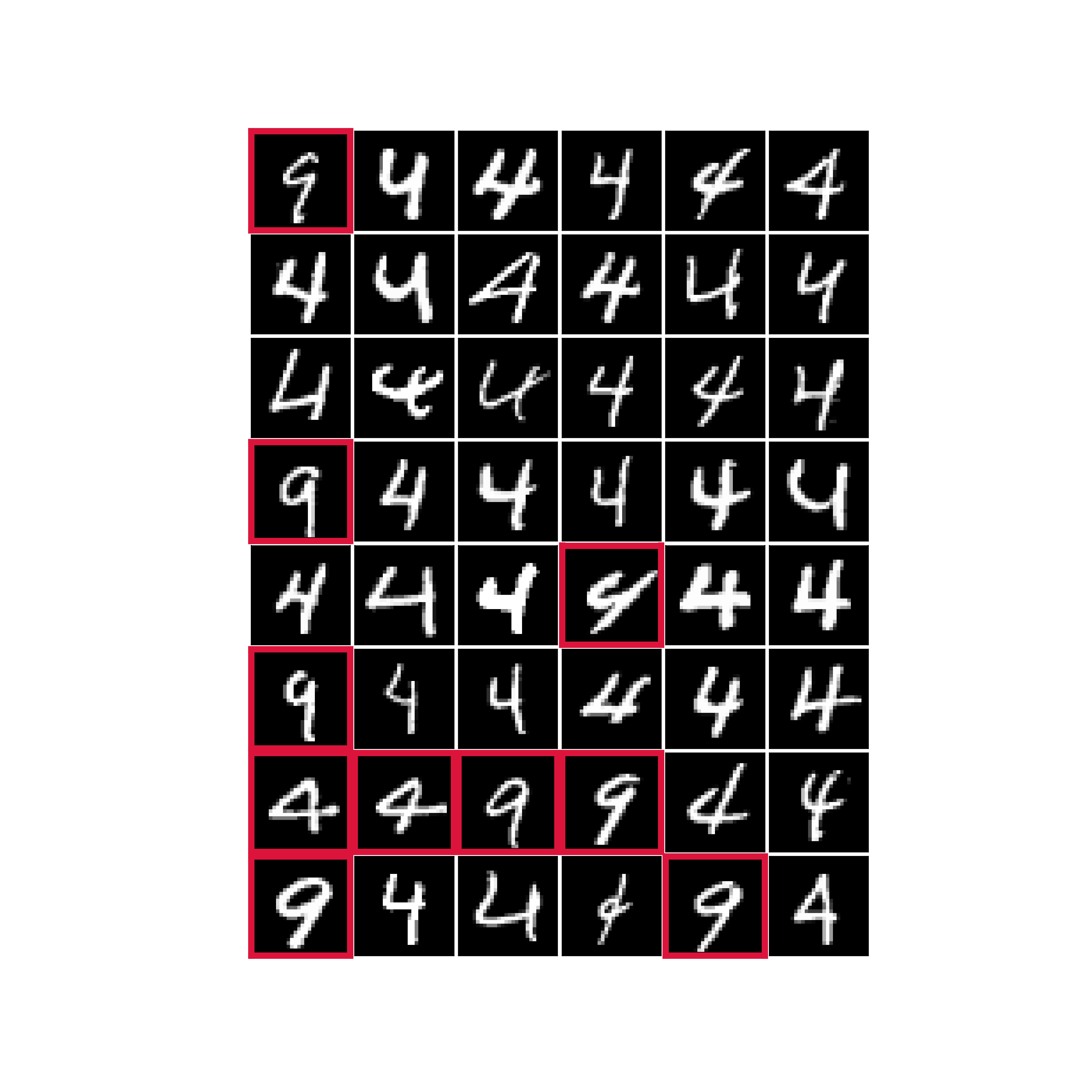} 
\end{tabular}
\end{center}
 \vspace{-.5cm}
\caption{Illustration of the novelty detection task on the MNIST dataset \citep{MNIST}; see Section~\ref{sec:realdata} for more details on the setting.} 
\label{fig:intro}
\end{figure}

\subsection{Existing strategies}

 For a standard multiple testing problem where the null distribution $P_0$ is known, the celebrated Benjamini-Hochberg (BH) procedure \citep{BH1995} controls the FDR in finite samples uniformly over all alternative distributions, when the test statistics are independent or satisfy the positive regression dependency on each one from a subset (PRDS) property; see \cite{BH1995,BY2001}.   
Variants of the BH procedure have been proposed to relax the conservatism when the fraction of true nulls is not close to $1$, such as the Storey-BH \citep{STS2004} or Quantile-BH procedure  \citep{BKY2006,Sar2008,BR2009}, and to robustify the FDR control under more general dependence structures (see \citealp{fithian2020conditional} and references therein).

Despite this generality,  BH-like methods have two major limitations in novelty detection problems with multivariate measurements:
\begin{itemize}
 \item[(i)] it is  based on $p$-values or, more generally, univariate scores with a {\it known distribution under the null}, which is typically out of reach for such problems;
\item[(ii)] the score function that transforms the multivariate measurements into univariate test statistics (e.g., the $p$-value transformation) is {\it pre-specified}, while it should be {\it learned from data}  for the sake of power.
\end{itemize}
We now discuss several existing solutions that partially circumvent these limitations. Table~\ref{Table1} provides a summary of the properties of each method, along with the corresponding applicable settings. 

A popular solution in the multiple testing literature is the empirical Bayes approach, which operates on the local FDR instead of the $p$-values. Assuming a two-group mixture model \citep{ETST2001}, the local FDR is defined as the probability of being null conditional on the observed measurement values. The latter can be estimated by estimating the null and alternative densities together with the proportion of nulls; see \cite{Efron2004, Efron2007b,  Efron2008,Efron2009b}. Combining local FDRs appropriately controls FDR asymptotically, under the assumptions that allow the model to be consistently estimated, and achieves optimal power, as shown in a series of paper by \cite{SC2007,CS2009,SC2009,CARS2019}. We refer to this procedure as the SC procedure hereafter.
Despite the appealing optimality guarantees, the model assumptions tend to be fragile when the dimension $d$ of the test statistics is moderately high. In such cases, accurate model estimation is hard to come by and the FDR of the SC procedure can thus be inflated; see our numerical experiments in Section~\ref{sec:numexp} for an illustration. 

Another line of research stems from conformal inference. While this technique is designed for prediction inference (see \cite{angelopoulos2021gentle} for a recent review), it can also be employed in the novelty detection problem. In particular, it can generate {\it conformal $p$-values} that are super-uniform under the null without any model assumption beyond that the data are exchangeable \citep[e.g.,][]{vovk2005algorithmic, balasubramanian2014conformal, bates2021testing}.
This approach starts by transforming $Z_j$ into a univariate score $S_j$, called the {\it non-conformity score}, that measures the conformity to the data and then computes an {\it empirical $p$-value}, also known as the conformal $p$-value, to evaluate the statistical evidence of being a novelty:
\begin{equation}\label{equemppvalues}
p_j=(n+1)^{-1}\Big(1+\sum_{i=1}^n \ind{S_i\geq S_{n+j}}\Big).
\end{equation} 
Each $p$-value is marginally super-uniform under the null due to exchangeability and hence yields a valid test. Nonetheless, since the conformal $p$-values all use the same null sample, the above operation induces dependence between the $p$-values, making it unclear whether common multiple testing procedures are guaranteed to control FDR. \cite{bates2021testing} carefully study the dependence structure and show that the split (or inductive) conformal $p$-values are PRDS. As a consequence, BH procedure applied on these conformal $p$-values controls the FDR. 
However, the approach limits the construction of the scores to be based solely on null examples and hence cannot learn the patterns of novelties in the mixed samples, unless extra labelled novelties are available \citep{liang2022integrative}, which are not always possible. Even when labelled novelties are present, they may behave differently than the ones in the mixed sample that we aim to detect. For this reason, 
\cite{bates2021testing} apply the one class classification techniques \citep[e.g.][]{scholkopf2001estimating} that are not adaptive to the novelties. In sum, while the method successfully solves the issue (i), it falls short of adequately addressing issue (ii). On the other hand, while other versions of conformal $p$-values, like full conformal $p$-values \citep{vovk2005algorithmic} and cross conformal $p$-values \citep{vovk2015cross, barber2021predictive}, can use test samples and yield marginally valid $p$-values, they generally fail to satisfy the PRDS property, making it unclear whether the BH procedure would control FDR. 

A subsequent work by \cite{yang2021bonus} proposes the Bag Of Null Statistics (BONuS) procedure for multiple testing problems with high dimensional test statistics, which largely motivates our method. The BONuS procedure learns a score function of the form $S_i=g(Z_i,(Z_1,\dots,Z_n))$ and the method is valid as long as $g(Z_i,\cdot)$ is permutation invariant thereby allowing the transformation to be adapted to novelties. While the framework is flexible, they focus on the parametric setting where the null distribution is known, like Gaussian or multinomial, and the measurements are independent. In these cases, they propose using the estimated local FDR as the score function for which the alternative distribution and null proportion are learned by an empirical Bayes approach. The BONuS procedure controls the FDR in finite samples regardless of the quality of the estimates, even if the working model is completely wrong. However, for novelty detection problems, the local FDR involves unknown null and alternative densities, which are difficult to fit in high dimensions. 
Hence, point (ii) mentioned earlier remains partially addressed.

Lastly, we briefly review other related work that study different settings. The ``counting knockoffs" procedure introduced by \cite{weinstein2017power} is designed for multiple testing for high-dimensional linear models with random design matrices. \cite{mary2021semisupervised} show that it is equivalent to applying the BH procedure to the scores $S_1, \ldots, S_{n+m}$ and closely related to the BONuS procedure. More recently, \cite{rava2021burden} develop a method that is equivalent to applying the BH procedure on the conformal $p$-values to obtain a finite sample control of the false selection rate (FSR) for the task of (supervised) classification.



\begin{table}
\begin{tabular}{|c|c|c|c|c|}
\hline
 & Finite sample & Adaptative  & Learning   & Unknown  \\
$^{\displaystyle{\textrm{Method}}}$ &  FDR control &  score &  alternative  &  null \\
\hline\hline
\cite{BH1995} &\cellcolor{mygreen} yes & \cellcolor{myred}no & \cellcolor{myred}no & \cellcolor{myred}no\\
\hline
\cite{SC2007} &  \cellcolor{myred}no &\cellcolor{mygreen}yes & \cellcolor{mygreen}yes & \cellcolor{mygreen}yes \\
\hline
\cite{weinstein2017power} & \cellcolor{mygreen} & \cellcolor{myred} & \cellcolor{myred}&  \cellcolor{mygreen} \\
\cite{mary2021semisupervised} & \cellcolor{mygreen} $^{\displaystyle{\textrm{yes}}}$ & \cellcolor{myred}$^{\displaystyle{\textrm{no}}}$ & \cellcolor{myred} $^{\displaystyle{\textrm{no}}}$&  \cellcolor{mygreen}$^{\displaystyle{\textrm{yes}}}$ \\
\hline
\cite{bates2021testing} &  \cellcolor{mygreen}yes & \cellcolor{mygreen}yes &  \cellcolor{myred}no &\cellcolor{mygreen}yes \\
\hline
\cite{yang2021bonus} &\cellcolor{mygreen}yes & \cellcolor{mygreen}yes &\cellcolor{mygreen}yes & \cellcolor{myred}no    \\
\hline
AdaDetect (our approach) & \cellcolor{mygreen}yes &\cellcolor{mygreen}yes &\cellcolor{mygreen}yes &\cellcolor{mygreen}yes\\
\hline
\hline
\end{tabular}\\
\caption{Properties of different methods and the specific settings in which they can be applied for novelty detection.}
\label{Table1}
\end{table}

\subsection{Contributions}

In this work we introduce AdaDetect, an extension\footnote{More precisely, we extend the version of BONuS where the score function is fit only in the initial stage; see the discussion in Section~\ref{sec:discussion} for more details.} of the BONuS procedure for novelty detection problems. In particular, we show how to leverage flexible off-the-shelf classification algorithms in machine learning to address both issues (i) and (ii) without compromising the FDR-controlling guarantees. 
\rev{In a nutshell, AdaDetect operates by initially splitting the null sample in two parts, $(Y_1, \ldots, Y_k)$ and $(Y_{k+1}, \ldots, Y_n)$, generating a membership label $A_{j} = -1$ if $Z_{j}\in \{Y_1, \ldots, Y_k\}$ and $A_{j} = 1$ otherwise, and subsequently calculating a score function using a binary classifier trained on $(Z_i, A_i)_{i=1}^{n+m}$ and applying the BH procedure on the empirical $p$-values. For the example illustrated in Figure~\ref{fig:intro}, Adadetect would split the null samples (digits `$4$') into two subsets and train a probabilistic classifier using both the null and test samples to distinguish the first subset of the null sample and the mix of the second subset of the null sample and the test sample (digits `$4$' and `$9$'). The predicted probability to be in the mixed sample is taken as the score. When the classification algorithm performs well, the scores tend to be larger for novelties than for nulls, because novelties are only present in the mixed sample. A comprehensive description of the procedure can be found in Section~\ref{sec:adaptiveteststat}.}

We summarize our main results below. 

\begin{itemize}

\item In Section \ref{sec:control}, we revisit the theoretical guarantees in \cite{weinstein2017power,mary2021semisupervised,bates2021testing} and provide new FDR bounds based on an extension of the leave-one-out technique in the multiple testing literature. The bounds show that AdaDetect, as well as its $\pi_0$-adaptive variants Storey-AdaDetect and Quantile-AdaDetect, controls the FDR in finite samples with {\it arbitrary} classification algorithms even if the algorithm performs poorly. This is in sharp contrast to the SC procedure which heavily relies on correct model specification and consistent estimation. 
\item In addition, we extend the result in \cite{bates2021testing} to show that the empirical $p$-values are PRDS under a  more general exchangeability assumption, even if the score function depends on both null and test samples. 
For instance, our condition covers the Gaussian distributions with equi-correlation (Example~\ref{ex:equi}).
This PRDS property suggests that the resulting $p$-values can be applied in other contexts beyond the FDR control (e.g., \citealp{GS2011}).


\item In Section~\ref{sec:calibration}, 
{we show that {\it any} score function that is monotone in the ratio between the average density of novelties and the null density yields the optimal power. In particular, the optimal classifier to distinguish between the null and mixed samples is efficient despite that the null training is split and that the mixed sample is contaminated by nulls}. The optimal score function can be obtained by minimizing certain loss function such as the cross-entropy loss that is commonly used in neural networks (NN hereafter). 

\item We provide non-asymptotic power analyses for AdaDetect in Section~\ref{sec:power}. First, we investigate AdaDetect with the score function given by a constrained empirical risk minimizer (ERM) of the 0-1 loss and show it approaches the optimal likelihood ratio test in an appropriate sense. Next, we provide an oracle inequality for general score functions and conditions under which the procedure mimics its oracle version. We apply the results to analyze power for AdaDetect procedures based on NN and on non-parametric kernel density estimation. 
\item  We demonstrate the efficiency, flexibility, and robustness of AdaDetect\footnote{The code is publicly available at \href{https://github.com/arianemarandon/adadetect}{https://github.com/arianemarandon/adadetect}} in Sections~\ref{sec:numexp}~and~\ref{sec:appli} on synthetic, semi-synthetic, and real datasets, including the MNIST image dataset and an astronomy dataset from the 'Sloan Digital Sky Survey'.
\end{itemize}


\section{Preliminaries}

\subsection{Notation}

As in Section \ref{sec:setting}, we let $Y = (Y_1, \ldots, Y_n)$ denote the null training sample (NTS) with a common marginal distribution $P_0$, $X = (X_1, \ldots, X_m)$ the test sample with $X_i\sim P_i$ ($1\leq i\leq m$), $Z = (Z_1, \ldots, Z_{n+m}) =(Y_1, \ldots, Y_n, X_1, \ldots, X_m)$  the full sample, $\cH_0=\{1\leq i\leq m\::\: P_i = P_0\}$ the set of nulls in the test sample  with $m_0 = |\cH_0|, \pi_0 = m_0 / m$, and $\cH_1=\{1,\dots,m\}\setminus \cH_0$ the set of novelties  with $m_1 = |\cH_1|, \pi_1 = m_1 / m$. 
For notational convenience, we write $n+\cH_0$ for the set $\{n+i,i\in \cH_0\}$.
Furthermore, we denote by $P$ the joint distribution of $Z$, which belongs to a family of distributions $\mathcal{P}$ (model). 

Throughout the paper, we consider the semi-supervised setting \citep{mary2021semisupervised} where the null distribution $P_0$ is unknown and one can access it only through the measurements in the NTS. In practice, the NTS can be obtained from external data, past experiments or black-box samplers. 


\subsection{Criteria}

A novelty detection procedure is a measurable function $R(\cdot)$ that takes $Z$ as input and returns a subset of $\{1,\dots,m\}$ corresponding to the indices of detected novelties within $\{X_1, \ldots, X_m\}$. Throughout the paper, we will slightly abuse the notation by using $R$ to refer to both the procedure and the rejection set given by the procedure. Ideally, we want $R(Z)$ to capture novelties (i.e., alternative hypotheses in $\cH_1$) and avoid inliers (i.e., null hypotheses in $\cH_0$). 
Given a procedure $R$, the false discovery rate (FDR) is defined as the expectation of the false discovery proportion (FDP) with respect to the distribution $P\in \mathcal{P}$:
\begin{align}
\FDR(P,R)&= \E_{Z\sim P}[\FDP(P,R)],\:\:\: \FDP(P,R)=\frac{\sum_{i\in \cH_0} \ind{i\in R}}{1\vee |R|}.\label{equFDRFDP}
\end{align}
Similarly, the true discovery rate (TDR) is defined as the expectation of the true discovery proportion (TDP):
\begin{align}
\TDR(P,R)&= \E_{Z\sim P}[\TDP(P,R)],\:\:\: 
\TDP(P,R)=\frac{\sum_{i\in \cH_1} \ind{i\in R}}{1\vee m_1(P)}.\label{equTDRTDP}
\end{align}
Note that $m_1(P)=0$ implies $\TDP(P,R)=0$. Our goal is to build a procedure $R$ that controls the FDR and maximizes the TDR to the fullest extent.

\subsection{BH algorithm and its $\pi_0$-adaptive variants}

Suppose a set of $p$-values $(p_i,1\leq i\leq m)$ is available, the BH algorithm \citep{BH1995} returns
$
R = \{i\in \{1,\dots,m\}\::\:p_i\leq \alpha \hat{k}/m\},
$ where $\alpha$ is the target FDR level and 
\begin{equation}\label{equkchapeau}
\hat{k} = \max\left\{k\in \{0,\dots,m\}\::\: \sum_{i=1}^m \ind{p_i\leq \alpha k/m}\geq k\right\}.
\end{equation}
When the null $p$-values $(p_i,i\in\cH_0)$ are independent, super-uniform, and independent of alternative $p$-values $(p_i, i\in \cH_1)$, the BH procedure is proved to control the FDR at level $\alpha \pi_0$ in finite samples \citep{BH1995}. The independence assumption can be further relaxed to the PRDS condition \citep{BY2001}.

When $\pi_0$ is not close to $1$, the BH procedure is conservative because $\alpha \pi_0 < \pi_0$. When $\pi_0$ is known, it can be applied at level $\alpha / \pi_0$ to close the gap. In practice, $\pi_0$ is usually unknown though. Nonetheless, there exists estimators $\hat{\pi}_0$ of $\pi_0$ such that the BH procedure with level $\alpha / \hat{\pi}_0$ continues to control the FDR under independence. Two celebrated estimators are introduced by \cite{STS2004} and \cite{BKY2006}: 
\begin{align}
\wh{\pi}^{Storey}_0 &=\frac{1+\sum_{i=1}^{m} \ind{p_i\geq \lambda}}{m(1-\lambda)},\quad \lambda>0;\label{estiSto}\\
\text{or }\quad \wh{\pi}^{Quant}_0 &=\frac{m-k_0+1 }{m(1-p_{(k_0)})},\quad k_0\in \{1,\dots,m\}\label{estiQuant},
\end{align}
where $p_{(k_0)}$ is the $k_0$-th smallest\footnote{In this paper, a convention is to order the $p$-values from the smallest to the largest, while the test statistics are ordered from the largest to the smallest.} $p$-value.
These procedures are often called the $\pi_0$-adaptive versions of the BH algorithm. 

\subsection{Our method}\label{sec:adaptiveteststat}

\begin{figure}[t!]
\begin{center}
\includegraphics[scale=0.25]{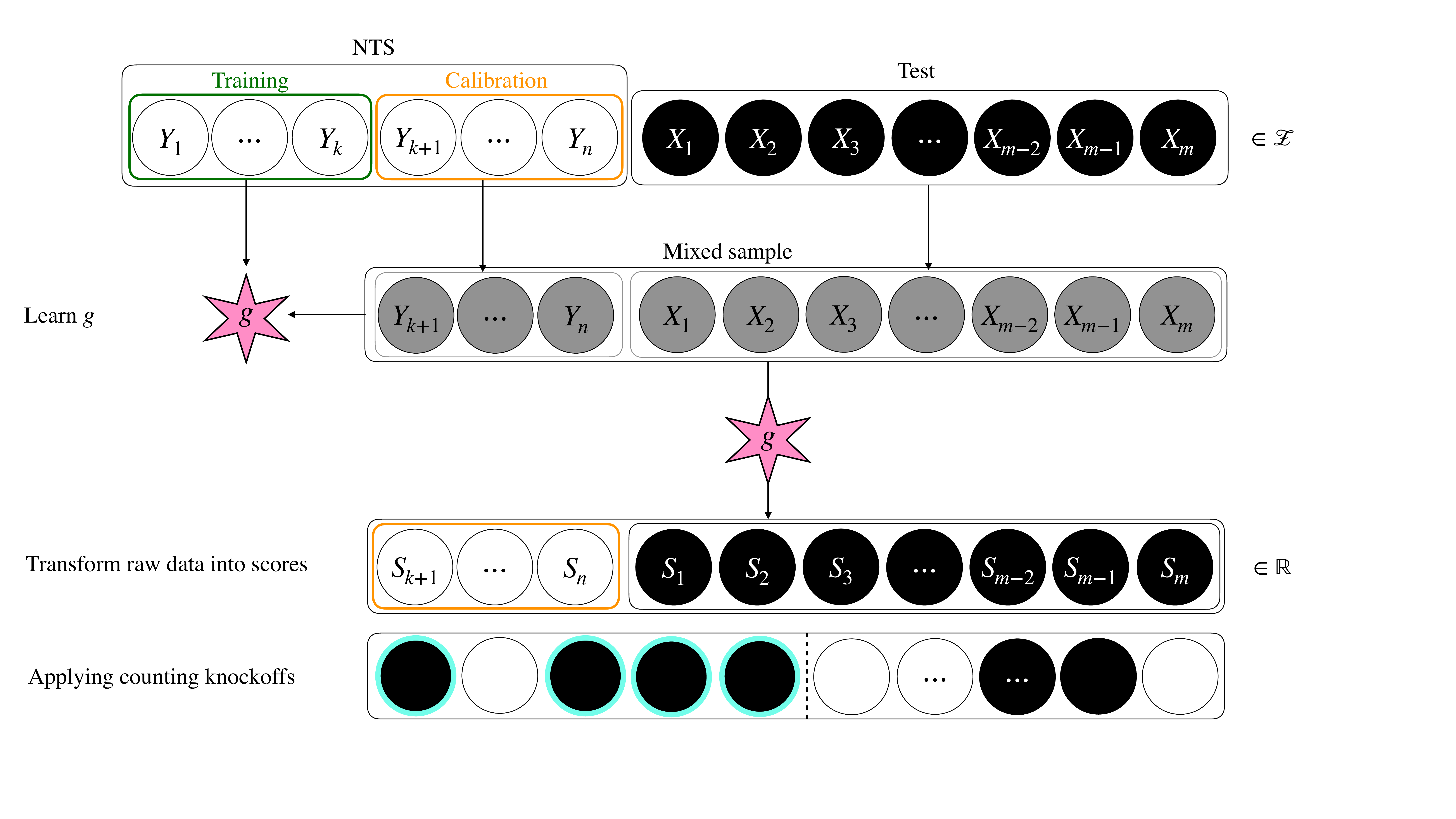}
\vspace{-1cm}
\end{center}
\caption{A schematic illustration of AdaDetect: $\bullet$/$\circ$ stands for a test/null observation, respectively. The vertical dashed line corresponds to the largest threshold $t$ for which $\wh{\FDP}(t)\leq \alpha$ and the $\bullet$  circled in blue correspond to the discoveries of AdaDetect procedure.} 
\label{figBONuS}
\end{figure}

In this paper, we propose a method called AdaDetect. It is an adaptive novelty detection procedure that extends the  existing strategies described in \cite{weinstein2017power}, \cite{yang2021bonus}, \cite{mary2021semisupervised}, and \cite{bates2021testing}. It starts by splitting the null sample $(Y_1,\dots,Y_n)$  in two samples $(Y_1,\dots,Y_k)$ and $(Y_{k+1},\dots,Y_{n})$ with $k\geq 0$. To avoid cluttering our notation, we define $\ell$ as the size of the second null sample, i.e.,
$\ell=n - k.$
It proceeds with the following steps.
\begin{enumerate}
\item Compute a data-driven score function of form
\begin{equation}\label{scorefunction}
g(z)=g(z,(Z_{1},\dots,Z_{k}), (Z_{k+1},\dots,Z_{n +m})),\:\:\: z\in \mathcal{Z},
\end{equation}
which satisfies the following invariance property: for any permutation $\pi$ of $\{k+1,\dots,n+m\}$ and $z,z_1,\dots,z_{n+m}\in \mathcal{Z}$, we have
\begin{equation}\label{constrainedg}
g(z,(z_{1},\dots,z_{k}), (z_{\pi(k+1)},\dots,z_{\pi(n +m)}))=g(z,(z_{1},\dots,z_{k}), (z_{k+1},\dots,z_{n +m})).
\end{equation}
\item Transform the raw data into univariate scores
\begin{equation}\label{equscores}
S_i = g(Z_i; (Z_1,\dots,Z_k) , (Z_{k+1},\dots,Z_{n+m})), \quad i \in \{ k+1, \ldots,n+m\}.
\end{equation}  
Here, we assume that novelties typically have large scores.
\item For each test point $X_j$, generate the empirical $p$-value by comparing $S_i$ with the scores in the NTS: 
\begin{equation}\label{emppvalues}
p_j = \frac{1}{\ell+1}\left(1+\sum_{i=k+1}^{n}\ind{S_i>S_{n+j}}\right), \:\:j\in\{1,\dots,m\}.
\end{equation}
\item Apply the BH algorithm to $(p_1,\dots,p_m)$ at the target level $\alpha$.
\end{enumerate}

We will call this procedure $\mbox{AdaDetect}_\alpha$ in the sequel to emphasize the target level.  
By simple algebra, the last two steps together are equivalent to the  ``counting knockoff'' algorithm proposed by \cite{weinstein2017power} applied to the scores $S_{k+1}, \ldots, S_{n+m}$. Specifically, the method declares $i$ as a novelty if $S_i \ge \hat{t}$ where $\hat{t}$ is the threshold defined by $ \min\big\{t\in \{S_i: k+1\leq i\leq n+m\}\::\: \wh{\FDP}(t) \leq \alpha \big\}$ for $\wh{\FDP}(t)=\frac{m}{\ell+1}\left(1+\sum_{i=k+1}^{n} \ind{S_i\geq t}\right) /\sum_{i=n+1}^{n+m} \ind{S_{i}\geq t}$.
Therefore, the counting knockoff procedure can be seen as a shortcut that avoids computing the empirical $p$-values explicitly. The pipeline for AdaDetect is illustrated in Figure~\ref{figBONuS}.

AdaDetect offers greater flexibility than existing methods in the types of score functions that can be employed. 
\begin{itemize}
\item Prespecified $p$-value transformations are score functions that do not depend on $(Z_{1},\dots,Z_{k})$ and 
$(Z_{k+1},\dots,Z_{n +m})$. For example, when $\mathcal{Z}=\R^d$, the $\chi^2$ test chooses the non-adaptive score $g(z)=\sum_{j=1}^d z_j^2$, $z\in \R^d$.
\item The one-class classification approach considered by \cite{bates2021testing} corresponds to score functions that only depend on $(Z_{1},\dots,Z_{k})$, but not $(Z_{k+1}, \dots, Z_{n+m})$.
\item The BONuS procedure \citep{yang2021bonus} considers empirical Bayes-based score functions that depend on the pooled sample $\{Z_1, \ldots, Z_{n+m}\}$ without distinguishing between the null and mixed samples. 
\item Our proposed method constructs the score function $g(\cdot,(Z_{1},\dots,Z_{k}), (Z_{k+1},\dots,Z_{n +m}))$ as the estimated probability by any probabilistic classifier that distinguishes between $(Z_{1},\dots,Z_{k})$ and $(Z_{k+1},\dots,Z_{n +m})$; see Section~\ref{sec:calibration} for details.
\end{itemize}

Lastly, we propose the Storey-AdaDetect and Quantile-AdaDetect as the $\pi_0$-adaptive versions of AdaDetect applied at level $\alpha / \wh{\pi}^{Storey}_0$ and $\alpha / \wh{\pi}^{Quant}_0$, respectively, in which the $p$-values have been replaced by the empirical ones. 
 
 \begin{remark}\label{rem:invariance} An appealing property of Adadetect and its adaptive versions is that the rejection is invariant to strictly increasing transformations of score function. This feature proves useful in the power analysis of AdaDetect, see Section~\ref{sec:calibration}.
 \end{remark}
 
 \begin{remark}\label{rem:samplesize}
   \rev{By construction, the empirical $p$-values are multiples of $1/(\ell + 1)$. As \cite{mary2021semisupervised} point out, the number of null samples $\ell$ needs to be larger than $m/(\alpha (1\vee M))$ in order to guarantee sufficient resolution of the $p$-values for the BH procedure, where $M\geq 0$ is some high-probability lower bound on the number of rejections. Typically, if $M$ is of the order of $m$, a constant $\ell$ would suffice, while if $M=0$ (i.e., without any prior knowledge on the number of rejections), $\ell$ should be larger than $m/\alpha$. }
\rev{In general practical situations where $n \gtrsim m$, we recommend setting $\ell = m$ and this choice works reasonably well in our numerical experiments. When $m > n$, it might be more appropriate to impose further assumptions on the distribution (e.g., the knowledge of $M$).}
\end{remark}

\section{FDR control}\label{sec:control}

In this section, we prove that AdaDetect and its $\pi_0$-adaptive variants control the FDR.
In Section \ref{sec:exch}, we state the key assumption of exchangeability and show it translates to the scores as long as $g$ satisfies the condition \eqref{constrainedg}. Based on this observation, we prove in Section \ref{sec:PRDS} that the empirical $p$-values are PRDS, which is a highly non-trivial extension of the results by \cite{bates2021testing}. Though the PRDS property implies the FDR control of AdaDetect as a result of \cite{BY2001}, we present in Section~\ref{sec:FDR} an alternative proof based on a new FDR expression that unify and extend the previous FDR bounds. 
Lastly, in Section \ref{sec:adaptFDR}, we prove the FDR control for Storey-AdaDetect and Quantile-AdaDetect based on an FDR bound for general $\pi_0$-adaptive versions of AdaDetect.

\subsection{Exchangeability}\label{sec:exch}

We make the following assumption on the raw measurements throughout the paper.
 \begin{assumption}\label{as:exchangeable0}
\mbox{$(Y_1,\dots,Y_{n},X_i,i\in \cH_0)$ are exchangeable conditional on $(X_i,i\in \cH_1)$.} \footnote{Note that such an assumption implicitly assumes that such a conditional distribution exists, which is always the case for instance when $\mathcal{Z}=\R^d$ or $\mathcal{Z}$ is discrete.}
\end{assumption}
Clearly, Assumption~\ref{as:exchangeable0} holds when the measurements are independent, as assumed by \cite{yang2021bonus} and \cite{bates2021testing}. In general, Assumption~\ref{as:exchangeable0} allows for dependencies among the measurements.

\begin{example}\label{ex:equi}
Consider the observation where
$
Z_i=\mu_i + \rho^{1/2} \xi + (1-\rho)^{1/2} \varepsilon_i, 
$
$1\leq i\leq n+m$, 
with the variables $\xi, \varepsilon_1,\dots,\varepsilon_{n+m}$ being i.i.d. $\sim \mathcal{N}(0,I_d)$, $\rho$ being a nonnegative correlation coefficient, and $\mu_i=0$ for $i\in \{1,\dots,n\}\cup (n+\cH_0)$ (hence $\mathcal{Z}=\R^d$). 
Then Assumption~\ref{as:exchangeable0} holds. 
The case $d=1$ corresponds to the Gaussian equi-correlated case, which is widely studied in the multiple testing literature \citep[e.g.,][]{Korn2004}. 
\end{example}

For our results, a necessary assumption is exchangeability of the scores under the null:
\begin{assumption}\label{as:newexch}
\mbox{$(S_{k+1}, \ldots, S_{n}, S_{n+i},i\in \cH_0)$ is exchangeable conditionally on $(S_{n+i},i\in \cH_1)$.}
\end{assumption}

It turns out the exchangeability of the raw measurements translates to the scores. 
\begin{lemma}\label{lem:ariane}
Under Assumption~\ref{as:exchangeable0}, the adaptive scores defined by \eqref{equscores} satisfy Assumption~\ref{as:newexch} for any score function that satisfies the condition \eqref{constrainedg}.
\end{lemma}
This result substantially simplifies the FDR analysis presented in the next section. To avoid unnecessary mathematical complications, we make the following mild assumption.
\begin{assumption}\label{as:newnoties}\label{as:noties}
\mbox{$(S_{k+1}, \ldots, S_{n+m})$ have no ties almost surely.}
\end{assumption}

\subsection{The $p$-values are PRDS}\label{sec:PRDS}

Following \cite{BY2001}, we say a family of $p$-values $(p_i,1\leq i\leq m)$ is PRDS on $\cH_0$ if, for any $i\in \cH_0$ and nondecreasing\footnote{A set $D\subset [0,1]^m$ is said to be nondecreasing if for any $x\in D$ and $y\in [0,1]^m$, we have $y\in D$ provided that $y_i\geq x_i$ for all $i$.  } measurable set $D\subset [0,1]^m$, the function $u\in [0,1]\mapsto \P((p_j,1\leq j\leq m)\in D\:|\: p_i= u)$ is nondecreasing.

\begin{theorem}\label{thm:PRDS_multidim}
For any family of scores $(S_{k+1},\dots,S_{n+m})$ satisfying Assumptions~\ref{as:newexch}~and~\ref{as:noties}, the empirical $p$-values defined in \eqref{emppvalues} are PRDS on $\cH_0$ and the null $p$-values are super-uniform.
In particular, under Assumptions~\ref{as:exchangeable0}~and~\ref{as:noties}, this result holds for the $p$-values generated by AdaDetect with a score function satisfying \eqref{constrainedg}.
 \end{theorem}

 We present a proof of Theorem~\ref{thm:PRDS_multidim} in Section~\ref{sec:proofthm:PRDS_multidim}. It extends Theorem~2 in \cite{bates2021testing} to dependent scores. Notably, the AdaDetect scores are dependent in general even if the measurements $Z_i$'s are independent because the data-adaptive score function depends on the entire dataset.

Theorem~\ref{thm:PRDS_multidim} has interesting consequences. First, the celebrated result for the BH procedure \citep{BY2001, RW2005} implies that AdaDetect {strongly controls the FDR at level $\alpha \pi_0$}. 
Second, the PRDS property is also useful for other purposes, such as post hoc inference \citep{GS2011}, FDR control with structural constraints \citep{ramdas2019sequential,loper2019smoothed}, online FDR control \citep{zrnic2021asynchronous, fisher2021saffron}, hierarchical FDR control \citep{foygel2015p} and weighted FDR control with prior knowledge \citep{ramdas2019unified}. Hence, our result paves the way for developing similar AdaDetect-style procedures in these contexts.

\subsection{A new FDR expression}\label{sec:FDR}
While the PRDS property implies the FDR control for AdaDetect, we pursue an alternative way based on a new expression for the FDR of the BH procedure in our setting, which would also yield a lower bound for FDR that is not implied by the PRDS property.

\begin{theorem}\label{thm:FDRBONuS}
Consider any family of scores $(S_{k+1},\dots,S_{n+m})$ satisfying Assumptions~\ref{as:newexch}~and~\ref{as:noties}. Let $R_\alpha$ denote the rejection set of BH procedure applied to $p$-values defined in \eqref{emppvalues} at level $\alpha$. Then, for any distribution $P\in \mathcal{P}$,
\begin{align}\label{equnew}
\FDR(P,R_\alpha)=  \sum_{i\in \cH_0}  \E\left(\frac{ \lfloor \alpha (\l+1)K_i/m\rfloor }{(\ell +1) K_i}\right),
\end{align}
where $K_i$ is a random variable that takes values in $\{1,\dots,m\}$ for any $i \in \cH_0$. In particular, under Assumptions~\ref{as:exchangeable0}~and~\ref{as:noties}, \eqref{equnew} holds with $R_\alpha=\mbox{AdaDetect}_{\alpha}$, the AdaDetect procedure at level $\alpha$.
\end{theorem}

The proof of Theorem~\ref{thm:FDRBONuS} is presented in Section~\ref{proof:thm:FDRBONuS}. It is similar to the classical leave-one-out technique to prove the FDR control for step-up procedure \citep[e.g.][]{FZ2006,RV2011,ramdas2019unified,Gir2022}, though it is non-trivial to handle empirical $p$-values. 
Since for any $x>0$ and integer $k$, we have  $ \lfloor x\rfloor k \leq \lfloor x k\rfloor  \leq xk$, 
expression \eqref{equnew} immediately implies the following bounds.

\begin{corollary}\label{FDRbounds}
Under Assumptions~\ref{as:exchangeable0}~and~\ref{as:noties}, the following holds, for any values of $k,\ell,m\geq 1$ and any parameter $P\in \mathcal{P}$: 
\begin{equation}\label{equboundsFDR}
m_0 \lfloor \alpha (\ell+1)/m\rfloor/(\ell +1) \leq \FDR(P,\mbox{AdaDetect}_\alpha)\leq \alpha m_0/m.
\end{equation}
In particular, $\FDR(P,\mbox{AdaDetect}_\alpha)= \alpha \pi_0$ when $\alpha (\ell+1)/m$ is an integer.
\end{corollary}

Corollary~\ref{FDRbounds} recovers Theorem~3.1 in \cite{mary2021semisupervised} which imposes a slightly more restrictive condition than Assumption~\ref{as:newexch}. Their proofs are based on martingale techniques and the proof for the lower bound is particularly involved. Here, we rely instead on the exact expression \eqref{equnew}, which is arguably simpler and more comprehensible. 

\subsection{New FDR bounds for $\pi_0$-adaptive procedures}\label{sec:adaptFDR}

For each $i\in \cH_0$, let $\mathcal{D}_{i}$ be the distribution of $(p'_j,1\leq j\leq m)$, where
\rev{\begin{equation}\label{def:Di}
\left\{\mbox{\begin{tabular}{l}
$p'_j=0$, $j\in \cH_1$, $p'_i=1/(\ell +1)$;\\
 $p'_j$, $j\in \cH_0\backslash\{i\}$ are i.i.d. conditionally on $U$ with a common c.d.f. $F^U$;\\
 $U=(U_1,\dots,U_{\ell +1})$ has i.i.d. $U(0,1)$ components,
\end{tabular}}\right.
\end{equation}
}
where $F^U$ denotes the discrete c.d.f. $F^U(x)=(1-U_{(\lfloor x(\ell+1)\rfloor+1)} )\ind{1/(\ell+1)\leq x<1 }+\ind{x\geq 1}$, $x\in \R$, and $U_{(1)}>\dots> U_{(\ell+1)}$ denote the order statistics of the vector $U$.
Note that the distribution $\mathcal{D}_{i}$ only depends on $i$, $m$, $\ell$ and $\cH_0$. The following general result holds.

\begin{theorem}\label{thm:AdaptBONuS}
In the setting of Theorem~\ref{thm:FDRBONuS}, denote $p=(p_i,1\leq i\leq m)$ the family of empirical $p$-values defined in \eqref{emppvalues} and consider any function $G : [0,1]^m \rightarrow (0,\infty)$ that is coordinate-wise nondecreasing. Then the procedure, denoted by $R_{\alpha m/G(p)}$, combining the BH algorithm  at level $\alpha m/G(p)$ with these empirical $p$-values is such that, for any parameter $P\in \mathcal{P}$,
\begin{equation}\label{equboundsadaptive}
\FDR(P,R_{\alpha m/G(p) }) \leq  \alpha \sum_{i\in \cH_0}  \E_{p'\sim \mathcal{D}_i}\left(\frac{1}{G(p')} \right),
\end{equation}
where $\mathcal{D}_{i}$ is defined by \eqref{def:Di}. In particular, this FDR expression holds for $R_{\alpha m/G(p) }=\mbox{AdaDetect}_{{\alpha m/G(p) }}$ 
under Assumptions~\ref{as:exchangeable0}~and~\ref{as:noties}.
\end{theorem}

Theorem~\ref{thm:AdaptBONuS} is proved in Section~\ref{proof:thm:AdaptBONuS}. 
\rev{In a nutshell, the distribution $\mathcal{D}_{i}$ is a least favorable distribution for the FDR of the adaptive BH procedure applied to empirical $p$-values defined in \eqref{emppvalues}. 
It can be seen as an adaptation of the classical leave-one-out technique for adaptive BH procedures; see \cite{BKY2006} and Theorem~11 of \cite{BR2009}.}

This result generalizes Theorem~6 of \cite{bates2021testing} which only works for the Storey-BH procedure. Our proof technique is fundamentally different and works for a broad class of estimators of $\pi_0$. 
Applying Theorem~\ref{thm:AdaptBONuS} to the estimators defined in \eqref{estiSto} and \eqref{estiQuant}, we obtain the following result.

\begin{corollary}\label{StoreyBH_multidim}
Under Assumptions~\ref{as:exchangeable0}~and~\ref{as:noties}, the following holds: 
\begin{itemize}
\item  Storey-AdaDetect controls the FDR at level $\alpha$ for any $\lambda= K/(\ell+1)$ and $K\in \{2,\dots,\ell\}$.
\item Quantile-AdaDetect controls the FDR at level $\alpha$ for any $k_0\in \{1,\dots,m\}$. 
\end{itemize}
\end{corollary}

The proof of Corollary~\ref{StoreyBH_multidim} is presented in Section~\ref{proof:StoreyBH_multidim}. It bounds the RHS of \eqref{equboundsadaptive} via combinatoric arguments. The result for Quantile-AdaDetect is novel. The result for Storey-AdaDetect was proved in \cite{yang2021bonus} for BONuS, with a different proof technique, in the special case where the scores are independent. Hence, we extend it to the exchangeable case. 

\rev{To illustrate the robustness of $\pi_0$-adaptive AdaDetect under dependence, consider Example~\ref{ex:equi} with common alternative means $\mu_i\equiv \mu \in \R^d$ and a fixed score function $S_i=\mu^T Z_i$, $1\leq i\leq n+m$}. One alternative approach to Storey-AdaDetect is to apply the Storey-BH procedure on the marginal $p$-values $p_i=\bar \Phi( S_{n+i}/\|\mu\|)$, $1\leq i\leq m$. Interestingly, Figure~\ref{fig:robustdep} shows that the Storey-BH procedure inflates the FDR substantially in the presence of high correlation while Storey-AdaDetect \rev{with $k = 0$} controls the FDR for any correlation $\rho$ (as implied by Corollary~\ref{StoreyBH_multidim}). Hence, while Storey-AdaDetect is only based on an NTS without the knowledge of the true null distribution, it is more robust to dependence than Storey-BH that requires more information. Furthermore, Storey-AdaDetect is more powerful than Storey-BH because the effect of the common variable $\xi$ is cancelled out in the calculation of empirical $p$-values.

\begin{figure}[t!]
\centering
\begin{minipage}{0.845\linewidth}
\includegraphics[width=\linewidth]{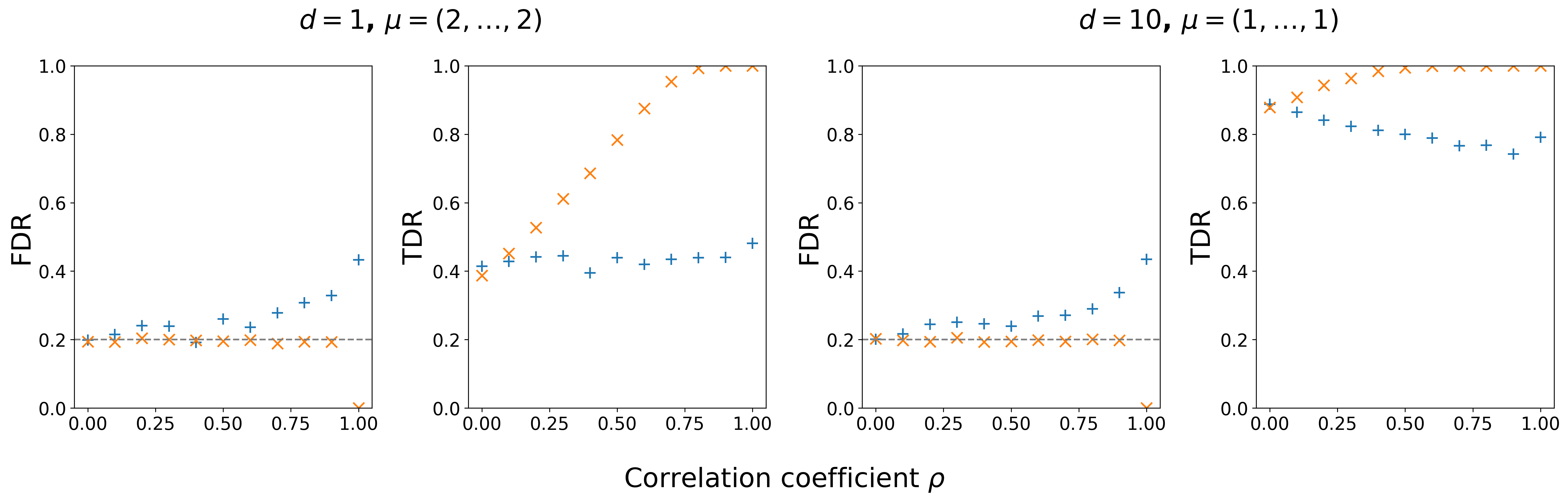}
\end{minipage}%
\begin{minipage}[c]{0.15\linewidth}
\includegraphics[width=\linewidth]{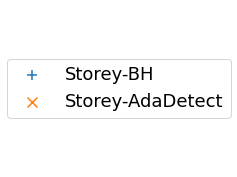}
\end{minipage}
\caption{FDR and TDR for Storey-BH and Storey-AdaDetect (both with oracle test statistics/scores) in Example~\ref{ex:equi} with varying correlation $\rho\in [0,1]$. The dimension $d = 1$ in the two left panels and $d = 10$ in the two right panels. In all settings, $m=100$, $n=\ell=1000$, $\alpha=0.2$, $\pi_0=0.9$, and $\lambda=500/1001$.} 
\label{fig:robustdep}
\end{figure}

\begin{remark}\label{rem:lasso}
Assumptions~\ref{as:newexch}~and~\ref{as:noties} hold true in other contexts. For example,  this is the case for LASSO-based scores in the Gaussian linear model where the design matrix has i.i.d. entries with a known distribution \citep{weinstein2017power}. Hence, the FDR bounds we developed also hold in those cases.
\end{remark}

\section{Constructing score functions}\label{sec:calibration}

While any score function satisfying \eqref{constrainedg} can be used in AdaDetect, we discuss principles and various techniques to construct score functions that yield high power. Section~\ref{sec:mixtureass} introduces the assumptions and notation. In Section~\ref{sec:settingpower} we show that the optimal score function is given by any monotone function of the ratio between the average density of novelties and the average density of all points. We proceed by discussing two methods to approach the optimal score based on direct density estimation (Section~\ref{sec:densityestim}) and classification (Section~\ref{sec:PU}). The latter is more scalable and flexible in the sense that it is able to wrap around any probabilistic classification algorithms. In Section~\ref{sec:cross}, we discuss a cross-validation approach for hyper-parameter tuning and model selection without compromising the finite-sample FDR control.

\subsection{Assumptions and notation}\label{sec:mixtureass}

In this section, we make the following two assumptions:
\begin{assumption}\label{as:indep}
 $Y_1,\dots,Y_{n},X_1,\dots,X_m$ are mutually independent.
 \end{assumption}
 Given the setting of Section~\ref{sec:setting}, we thus have under Assumption~\ref{as:indep} that $(Y_1,\dots,Y_{n},X_i,i\in \cH_0)$ are i.i.d. $\sim P_0$ and independent of $(X_i,i\in \cH_1)$ which are  mutually  independent.

\begin{assumption}\label{equ-marg}
For each $i\in \{0\}\cup \cH_1$, $P_i$ has a positive density $f_i$ w.r.t. a measure $\nu$. 
\end{assumption}
Let 
\begin{align}
f&=\pi_0 f_0+ \pi_1 \bar f_1,\label{equ:f}\\
 \bar f_1&= m_1^{-1}\sum_{i\in \cH_1} f_i.\label{equ:f1bar}
\end{align}
Under Assumptions~\ref{as:indep} and~\ref{equ-marg}, $f_0$ is the average density of $(Z_1,\dots,Z_k)$, $\bar{f}_1$ is the average alternative density, $f$ is the average density of the test sample $(X_1,\dots,X_m)$. Similarly the average density of $(Z_{k+1},\dots,Z_{n+m})$ is $f_\gamma$ where
\begin{align}
\gamma=\frac{m_1}{\ell+m};\:\:\:\:\:\:
f_\gamma=(1-\gamma) f_0 + \gamma \bar{f}_1 =\frac{\ell}{\ell+m} f_0+ \frac{m}{\ell+m} f.\label{equfgamma}
\end{align}
Compared to $f$, the mixture $f_\gamma$ is contaminated by more nulls, that is,  $\pi_0\leq 1-\gamma = \frac{\ell+m_0}{\ell +m}$. Lastly, we define the density ratio
\begin{equation}\label{equLR}
\lrt(x)= \frac{\pi_1 \bar f_1(x)}{ f(x)}, \quad x\in \mathcal{Z}.
\end{equation}
Note that $\lrt(x)\in (0,1)$ for $\nu$-almost every $x\in \mathcal{Z}$ by Assumption~\ref{equ-marg}.

  \subsection{Optimal score function}\label{sec:settingpower}

\rev{Recall that AdaDetect is equivalent to applying the counting knockoff on the scores which relies on an estimator $\widehat{\mathrm{FDP}}$ (Section \ref{sec:adaptiveteststat}). For each given $t$, when $\ell$ and $m$ is large, 
$\widehat{\mathrm{FDP}}(t)\approx m \:\P_{S_i\sim P_0}(S_i \ge t)/\E[|R(t)|] \approx \E[|R(t)\cap \{k+1, \ldots, n\}|]/\E[|R(t)|],$
where $R(t)$ is set of rejections at threshold $t$. The RHS is called the marginal FDR (mFDR), an error metric that is close to FDR when $|R(t)|$ is large and often used for asymptotic power analysis of FDR-controlling procedures \citep[e.g.][]{SC2007, lei2018adapt}. The following theorem derives the optimal score function among all procedures that reject hypotheses with $S_i$ above some thresholds subject to mFDR control (see \citealp{weinstein2021permutation, rosset2022optimal} for results for FDR instead of mFDR). }
\begin{theorem}\label{th:SCextended}
Assume Assumptions~\ref{as:indep} and~\ref{equ-marg} hold. The likelihood ratio function $r(\cdot)$ defined in \eqref{equLR} is an optimal score function in the sense that the rejection set $R=\{i \in \{1,\dots,m\}\::\: \lrt(X_i)\geq c(\alpha)\}$, where $c(\alpha)\in (0,1)$ is chosen such that $\mbox{mFDR}(R)=\alpha$ (assuming it exists), has a higher TPR than any rejection set $R'=\{i \in \{1,\dots,m\}\::\: \lrt'(X_i)\geq c'\}$ where $c'\in \R$ and $r':\mathcal{Z}\mapsto \R$ is measurable with mFDR at most $\alpha$.
\end{theorem}

The proof can be found in Section~\ref{sec:th:SCextended}. 
Theorem~\ref{th:SCextended} suggests the following oracle procedure.

\begin{definition}\label{def:oraclebonus}
The oracle AdaDetect procedure, denoted by AdaDetect$^*$, is defined as the AdaDetect procedure with the score function $\lrt(\cdot)$ defined in \eqref{equLR}.
\end{definition}

Since AdaDetect is invariant under any strictly monotone transformation of the score function (see Remark~\ref{rem:invariance}), AdaDetect$^*$ can be realized as any AdaDetect procedure with a score function of the form
\begin{equation}\label{equoraclescore}
g^*= \Psi \circ \lrt, \mbox{ for some increasing continuous $\Psi:(0,1)\to {\R}$},
\end{equation}
where $\Psi$ could depend on unknown parameters. This is a crucial property of AdaDetect that enables flexible classification methods to construct score functions without concerning about the composition of nulls and novelties that may change the oracle score $\lrt$.

Since $\lrt$ (or $g^*$) is unknown, the oracle procedure AdaDetect$^*$ is not directly accessible in practice. Our goal is to learn a $g^{*}$ in the form of \eqref{scorefunction} that satisfies the constraint \eqref{constrainedg}.

\subsection{Density estimation}\label{sec:densityestim}

A first example of score function is built from density estimation. From \eqref{equ:f} and \eqref{equfgamma}, the following score
\begin{equation}\label{equ:gstardensity}
\rev{g^*(x)=f_\gamma(x)/f_0(x) = 1-\gamma/\pi_1+ (\pi_0\gamma/\pi_1) (1-r(x))^{-1}}
\end{equation}
 is indeed of the form \eqref{equoraclescore}.  
A straightforward approach is to directly estimate the densities as follows.
\begin{itemize}
\item Estimate $f_0$ by a density estimator $\wh{f}_0$ based on the sample $(Z_1,\dots,Z_k)$
\item Estimate $f_\gamma$ by a density estimator $\wh{f}_\gamma$ based on the mixed sample $(Z_{k+1},\dots,Z_{n+m})$ via a mixture estimation approach.
\item Estimate $g^*(x)$ by $\wh{g}(x)=\wh{f}_\gamma(x)/\wh{f}_0(x)$ assuming that $\wh{f}_0(Z_i)>0$.
\end{itemize}
Above, the density estimators can be either parametric or non-parametric. Both versions will be considered in the sequel (see Section~\ref{sec:densityestim} and the numerical experiments in Section~\ref{sec:numexp}).
Note that \cite{yang2021bonus} applies this approach when $f_0$ is known.

\subsection{PU classification}\label{sec:PU}

While density estimation is straightforward, it is not scalable when the dimension $d$ is large; see the numerical experiments in Section~\ref{sec:numexp} for an illustration. In this section, we consider a different strategy that estimate density ratios through probabilistic classification \citep[e.g.][]{friedman2003multivariate, sugiyama2012density, lei2021distribution, wang2022approximate}.

Define $(Z_1,\dots,Z_k)$ as the ``positive sample" and the sample $(Z_{k+1},\dots,Z_{n+m})$ as the ``unlabeled sample",  and let $(A_1,\dots,A_k)=(-1,\dots,-1)$ and $(A_{k+1},\dots,A_{n+m})=(1,\dots,1)$ the corresponding labels.   
In this context, the classification task is typically referred to as the PU (positive unlabeled) classification, which is an active research area; see \cite{du2014analysis,calvo2007learning,Guo_2020_CVPR,ivanov2020dedpul} among others and \cite{bekker2020learning} for a recent review. 
Here, we are considering a slightly different setting where the unlabeled samples are independent but not identically distributed.
 
 Usually, the classifier is learned by empirical risk minimization (ERM) where the objective function is in the form of
$
\wh{J}_\lambda(g) = \sum_{i=1}^{n+m} \lambda_{A_i}\ell(A_i,g(Z_i) ) = \sum_{i=1}^{k} \ell(-1,g(Z_i) ) + \lambda\sum_{i=k+1}^{n+m} \ell(1,g(Z_i) ),
$
where $\ell:\{-1,+1\}\times \R\to \R_+$ is a loss function and $\lambda_{a}=\lambda\ind{a\geq 0} + \ind{a\leq 0}$ with $\lambda> 0$ measuring the relative cost misclassifying a positive sample to misclassifying an unlabeled sample. Here, $g$ is a function that belongs to $\mathcal{G}$, a class of measurable functions from $\mathcal{Z}$ to $\R$ and the classifier corresponds to the sign of $g$. Typical choices of the loss function include the hinge loss $\ell(a,u)=0.5(1-au)_+$ and the cross entropy loss $\ell(a,u)=-\log(1-u)\ind{a=-1} -\log(u)\ind{a=+1}$. The population objective function is given by
\begin{align}
J_\lambda(g)&= \E \wh{J}_\lambda(g) =k \E_{Z\sim f_0} \ell(-1,g(Z) ) + \lambda (\ell+m) \E_{Z\sim f_\gamma} \ell(1,g(Z) ) \label{equJlambda2},
\end{align}
where $f_\gamma$ is defined in \eqref{equfgamma}. The following result shows that the minimizer of \eqref{equJlambda2} over all measurable functions yields an optimal score in the form of \eqref{equoraclescore} when the loss function $\ell$ is appropriately chosen.

\begin{lemma}\label{lem:PUclassiforacle}
Let $g^\sharp$ denote the minimizer of \eqref{equJlambda2} over all measurable functions. 
\begin{itemize}
\item[(i)] When $\ell(\cdot,\cdot)$ is the hinge loss, assuming that the set $\{x\in \mathcal{Z}\::\: f_\gamma(x)=cf_0(x)\}$ is of $\nu$-measure zero for any $c > 0$, where $\nu$ is defined in Assumption \ref{equ-marg}, 
  $$g^\sharp(x)=\mathrm{sign}\left(\frac{\lambda (\ell+m)}{k} \frac{f_\gamma(x)}{f_0(x)}-1\right) = \mathrm{sign}\left(\frac{\lambda\ell}{k} + \frac{\lambda m_0}{k}(1 - r(x))^{-1} - 1\right),$$
  and the minimum is unique $\nu$-almost everywhere.
\item[(ii)] When $\ell(\cdot,\cdot)$ is the cross entropy,
  $$g^\sharp(x)=\frac{\lambda (\ell +m) f_\gamma(x)}{\lambda (\ell +m) f_\gamma(x)+kf_0(x)} = \left(1+\left\{\frac{\lambda\ell}{k} + \frac{\lambda m_0}{k}(1 - r(x))^{-1}\right\}^{-1}\right)^{-1},$$
  and the minimum is unique $\nu$-almost everywhere.
\end{itemize}
\end{lemma}

The proof is presented in Section~\ref{proofPUclassiforacle}. 
Clearly, $g^\sharp$ is an optimal score function in the form of \eqref{equoraclescore} with the cross-entropy loss but not so with the hinge loss because the sign function is not strictly monotone. For cross-entropy loss, when $\lambda = 1$,
\begin{equation}\label{equ:gstarproba}
g^\sharp(x)=\frac{ \frac{\ell +m}{n+m} f_\gamma(x)}{\frac{\ell +m}{n+m} f_\gamma(x)+\frac{k}{n+m} f_0(x)},
\end{equation}
which can be roughly interpreted as the posterior probability to be in class $1$. 

In practice, it is computationally infeasible and statistically inefficient to optimize over all measurable functions. Instead, we often choose a function class $\mathcal{G}$ and estimate the score function by
\begin{align}
\wh{g} &\in \argmin_{g\in \mathcal{G}}\wh{J}_\lambda(g).\label{equgPU}
\end{align}
By construction, $\wh{J}_\lambda(g)$ is invariant to permutations of $(Z_{k+1},\dots,Z_{n+m})$, $\wh{g}$ always satisfies the condition \eqref{constrainedg}. When $\mathcal{G}$ has low complexity, we should expect $\wh{g}\approx g^\sharp_{\mathcal{G}}$ where
\begin{align}
g^\sharp_{\mathcal{G}} &\in \argmin_{g\in \mathcal{G}} J_\lambda(g). \label{equgPUth}
\end{align}
On the other hand, when $\mathcal{G}$ is sufficiently rich, we can expect $g^\sharp_{\mathcal{G}}\approx g^\sharp$. In summary, when the function class $\mathcal{G}$ and the loss function $\ell(\cdot, \cdot)$ are chosen appropriately, $\hat{g}\approx g^\sharp_{\mathcal{G}}\approx g^\sharp$, which is an optimal score function.

We illustrate the roles of function classes and loss functions in a simple setting where the positive class and the unlabeled class are generated from two gaussian distributions with dimension $1$ or $2$. The results are presented in Figure~\ref{fig:illustrationERM}, with each row corresponding to a data-generating process. For all settings, the first panel displays the null and alternative distributions and the second panel displays the distributions of the positive and unlabeled classes. In all settings, we plot $\hat{g}$ and $g^\sharp$ for hinged loss (SVM) and cross-entropy losses with two function classes, an inaccurate one (Logistic Regression) and an accurate one (Neural Networks). 
In the two-dimensional settings, we display the functions by contour plots. {For instance, for the cross entropy loss, we can observe the NN function class outperforms the logistic function class for approximating $g^\sharp$.}


In conclusion, both the loss function $\ell(\cdot, \cdot)$ and the function class $\mathcal{G}$ are pivotal. Among the two loss functions we discuss, the cross entropy loss with a sufficiently rich function class (e.g., fully-connected neural networks) is particularly suitable for AdaDetect in that it is computationally feasible and approximately optimal. In contrast to the classification literature, hinged loss is undesirable for our purpose since the estimator does not converge to an optimal score. 

\begin{figure}
	\hspace{-7mm}
	\includegraphics[width=10cm]{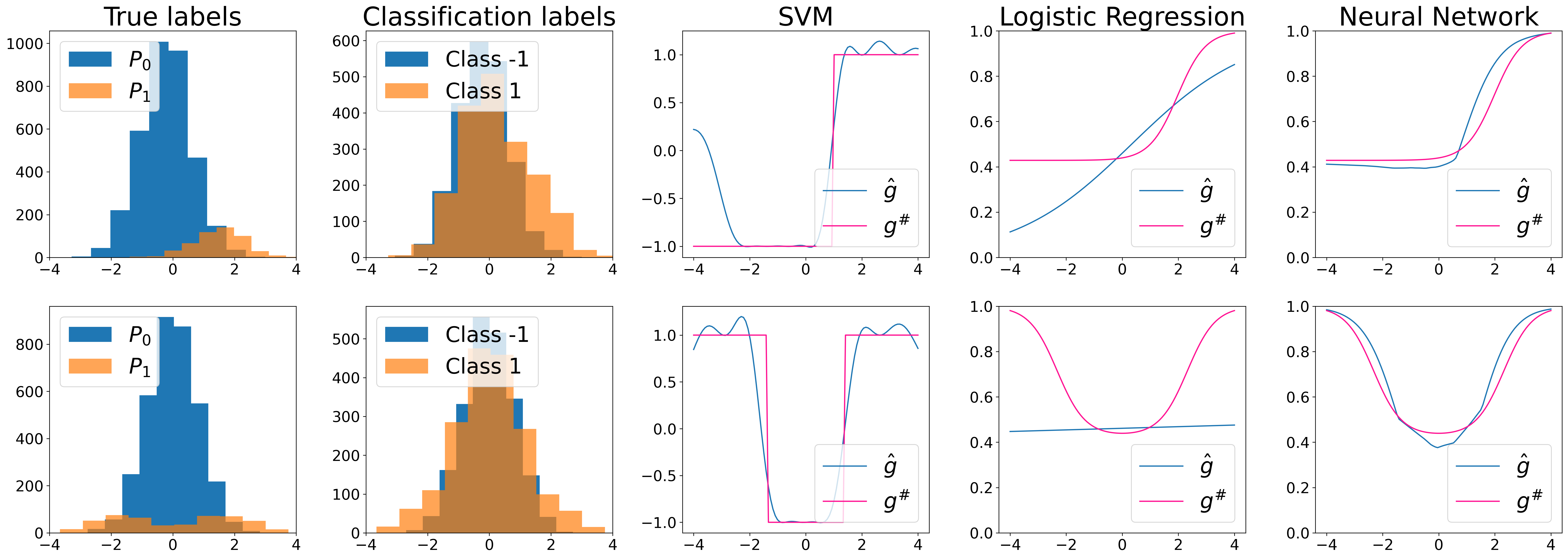}\\
	\includegraphics[width=10cm]{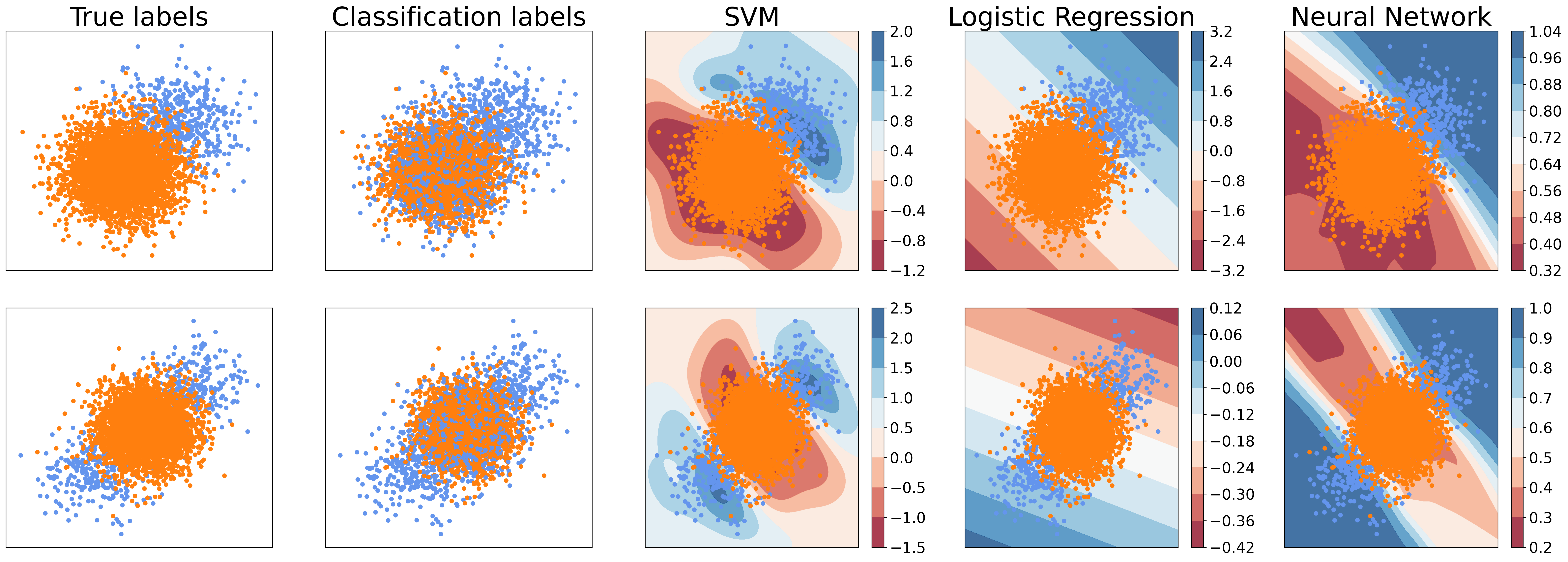}
	\caption{Plot of $g^*$ and $g^\sharp$ in different settings (rows) with different loss functions $\ell(\cdot,\cdot)$ and function classes $\mathcal{G}$ (with default parameters in \texttt{scikit-learn}). In all settings, $m = 1000$, $m_0 = 500$, $m_1 = 500$, $n = 3000$, and $k=2000$. The top two rows correspond to $d = 1$ and the bottom two rows correspond to $d = 2$. In all cases, $P_0 = \mathcal{N}(0, I_d)$. For the first and third rows, $P_1 = \mathcal{N}((2,\dots,2), I_d)$ (one-sided alternatives); for the second and fourth rows, $P_1 = 0.5\mathcal{N}((2,\dots,2), I_d) + 0.5\mathcal{N}((-2,\dots,-2), I_d)$ (two-sided alternatives).
	\label{fig:illustrationERM}}   
\end{figure}

\subsection{AdaDetect with cross-validation}\label{sec:cross}

In previous sections, we focus on a single score function. Nevertheless, most density estimation and classification algorithms involve hyperparameters that require data-driven tuning to maximize the power. Examples include the bandwidth for kernel density estimation, the maximum depth for random forests, the width and number of hidden layers for neural networks, and the numerical algorithm to optimize the loss.

Formally, we assume the researcher has a class of candidate score functions $\{{g}_\upsilon,\upsilon \in \mathcal{U}\}$ indexed by the hyper-parameter $\upsilon$. The goal is to choose $\hat{\upsilon}$ based on data and use $g_{\hat{\upsilon}}$ as the score function without breaking the FDR guarantee. By Theorem \ref{thm:FDRBONuS}, the FDR is controlled so long as $g_{\hat{\upsilon}}$ satisfies the condition \eqref{constrainedg}. Motivated by the ``double BONuS'' procedure proposed in \cite{yang2021bonus}, we propose the following version of AdaDetect with cross-validation, which we abbreviate as the AdaDetect cv procedure. 

\begin{enumerate}
\item Split $(Y_1, \ldots, Y_k)$ further into two parts $(Y_{1},\dots,Y_{s})$ and $(Y_{{s+1}},\dots,Y_{k})$ for some $s < k$.
\item Generate a class of score functions $g_{\upsilon}$ that satisfy a stronger condition than \eqref{constrainedg}:
\begin{equation*}
g_{\upsilon}(z,(z_{1},\dots,z_{s}), (z_{\pi(s+1)},\dots,z_{\pi(n +m)}))=g(z,(z_{1},\dots,z_{s}), (z_{s+1},\dots,z_{n +m})).
\end{equation*}
\item For each $g_{\upsilon}$, apply AdaDetect with $(Y_{k+1},\dots,Y_n, X_1,\dots,X_m)$ being the test sample and $(Y_1,\dots,Y_k)=(Y_1,\dots,Y_s;Y_{s+1},\dots,Y_k)$ being the NTS. Denote by $r_{\upsilon}$ the number of rejections.
\item Choose $\hat{\upsilon}\in \argmax_{\upsilon \in \mathcal{U}} r_{\upsilon}$
\item Apply AdaDetect with score function $g_{\hat{\upsilon}}$ to the original problem (with $(X_1,\dots,X_m)$ being the test sample and $(Y_1, \ldots, Y_n)=(Y_1, \ldots, Y_k;Y_{k+1},\ldots, Y_n)$ being the NTS).
\end{enumerate}

The pipeline to compute $g_{\hat{\upsilon}}$ is illustrated in Figure \ref{figDoubleBONuS}. By definition, each $g_\upsilon$ is invariant to permutation of $(Y_{s+1}, \ldots, Y_n, X_1, \ldots, X_m)$ and hence invariant to permutation of the mixed sample $(Y_{k+1}, \ldots, Y_n, X_1, \ldots, X_m)$. Thus, $r_{\upsilon}$ is also invariant to $(Y_{k+1}, \ldots, Y_n, X_1, \ldots, X_m)$, implying that $\hat{\upsilon}$ is so as well. As a result, $g_{\hat{\upsilon}}$ satisfies the condition \eqref{constrainedg}. Therefore, the results in Section~\ref{sec:control} all carry over to the AdaDetect cv procedure. 

In principle, we can use any other objective function that is invariant to the mixed sample than the number of rejections $r_{\upsilon}$. Nonetheless, $r_{\upsilon}$ tends to be a good proxy for the number of rejections in the last step and hence a better objective to optimize than the indirect ones like classification accuracy.

\begin{remark}\label{rem:samplesizecv}
When fitting the hyper-parameter $\upsilon$, the sample sizes $s,k-s,\ell+m,\ell$ do not maintain the same proportions as the original sizes $k,\ell,m$. 
Our recommendation, following the guidelines in Remark~\ref{rem:samplesize}, is to choose $s$ such that $k-s$ is of the same order as $\ell+m$ and $s$ is of the same order as $m$ (e.g., $\ell=m$, $s=3m$, $k=4m$). 
\end{remark}  

\begin{remark}
  \rev{The cross-validation can rule out overfitted models that performs well in training data but does poorly out of sample. By including nonsophisticated baseline models that likely generalize, the power of AdaDetect becomes less sensitive to overfitting of other complicated models or other failure modes that we have yet discovered. For example, the researcher can always add a non-adaptive score that cannot incur overfitting and might be underpowered. 
  }
\end{remark} 

\begin{figure}[t!]
\begin{center}
    \includegraphics[scale=0.22]{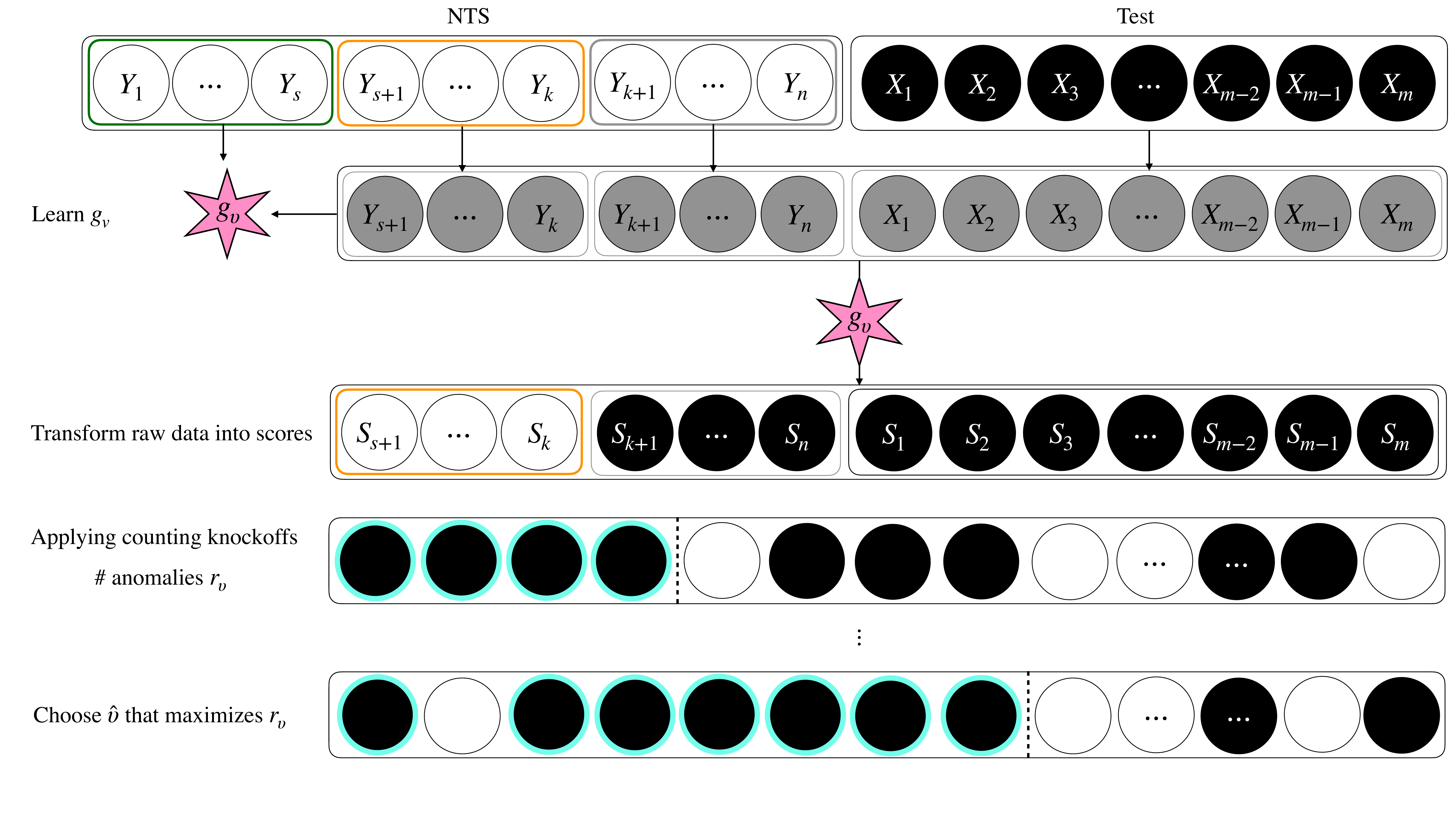}
\end{center}
\caption{The pipeline to compute score function $g_{\hat{\upsilon}}$ for AdaDetect cv. Same pictural conventions as in Figure~\ref{figBONuS}.} 
\label{figDoubleBONuS}
\end{figure}

\section{Power results}\label{sec:power}

In this section, we analyze the power of AdaDetect with appropriately chosen score functions. Throughout this section we assume that the measurements take values in $\mathcal{Z}=\R^d$. We start in Section~\ref{sec:NPPUclassif} with a specific score function given by a constrained empirical risk minimizer (ERM) with the $0$-$1$ loss and show it is as powerful as the classification approach based on the optimal score functions defined in \eqref{equoraclescore} when the function class is sufficiently flexible and up to asymptotically vanishing remainder terms. In Section~\ref{sec:genresultpower}, we turn to a general estimated score function that is close to an oracle (deterministic) score function on all measurements in the mixed sample. When the latter is sufficiently smooth, we show that AdaDetect with the estimated score function is as efficient as AdaDetect with the oracle score function, up to explicit remainder terms that are asymptotically vanishing. 

\subsection{A constrained ERM score function}\label{sec:NPPUclassif}

For the convenience of theoretical analysis, we study a constrained empirical risk minimizer (ERM) score function with $0$-$1$ loss motivated by the Neyman-Pearson (NP) formulation of classification problems given in   \cite{BLS2010}; see also \cite{cannon2002learning} and \cite{scott2005neyman}. Define
\begin{align}
\hat{R}_0(g)&=   k^{-1} \sum_{i=1}^{k} \ind{g(Z_i)\geq 0}, \:\:\:\:R_0(g)=\E \hat{R}_0(g) =\P_{Z\sim f_0} (g(Z)\geq 0)\label{NPexactRgamma},\\ 
\hat{R}_\gamma(g)& =(m+\ell)^{-1} \sum_{i=k+1}^{n+m} \ind{g(Z_i)<  0}, \:\:\:\:R_\gamma(g)=\E \hat{R}_\gamma(g) = (1-\gamma) (1-R_0(g))+ \gamma R_1(g),\nonumber\\
R_1(g)& =\P_{Z\sim \bar{f}_1} (g(Z)< 0),\nonumber
\end{align} 
where $\gamma$, $f_0$, and $\bar{f}_1$ are defined in \eqref{equfgamma}.
We consider a function class $\cG$ with a finite Vapnik-Chervonenkis (VC) dimension $V(\cG)$ (\citealp{vapnik1998statistical}) and the following constrained ERM score function 
\begin{align}
  \hat{g}&\in \argmin_{g \in \cG}\left\{\hat{R}_\gamma(g)\::\: \hat{R}_0(g)\leq \beta + \epsilon_0\right\},\label{defghat}
\end{align}
for some $\epsilon_0 > 0$, as well as its population version
\begin{align}
g^\sharp_\cG&\in \argmin_{g \in \cG}\left\{ R_\gamma(g)\::\:  R_0(g)\leq \beta \right\}.\label{equNP}
\end{align}

\begin{theorem}\label{powerhighdim}
Consider the setting of Theorem \ref{th:SCextended}. Assume $\alpha, \beta\in (0,1)$, $k,m_1\geq 1$, and $g^\sharp_\cG$, defined in \eqref{equNP}, satisfies $R_0(g^\sharp_\cG) = \beta$. Fix any $\delta\in (0,1/2)$. 
Then there exist constants $C,C'>0$ that only depend on $\delta$ such that, if 
\begin{equation}\label{equDelta}
\epsilon_0 = C\sqrt{\frac{V(\cG) + \log(1/\delta)}{k}}, \quad \Delta = C' \gamma^{-1} \sqrt{\frac{V(\cG) + \log(1/\delta)}{k\wedge \ell}},
\end{equation}
where $\gamma$ is defined in \eqref{equfgamma},  the following results hold.
\begin{itemize}
\item[(i)]
With probability at least $1-\delta$, 
$
{R}_0(\hat{g})  \leq \beta + \Delta
$
and
$
{R}_1(\hat{g}) \leq {R}_1(g^\sharp_\cG) + \Delta.
$
\item[(ii)] Let $M = \lceil (1-{R}_1(g^\sharp_\cG)-\Delta) m_1\rceil$. Assume that
 \begin{equation}\label{equlconstraint}
\rev{1-{R}_1(g^\sharp_\cG)  \geq  (1+\alpha^{-1})  \Delta, \quad \ell\geq \frac{2 m}{\alpha M}, \quad \beta \le \frac{0.4\alpha M}{m}. }
\end{equation}
Then, with probability at least $1-\delta$,
\begin{align}
\mbox{AdaDetect}_\alpha &\supset \{i\in \{1,\dots,m\}\::\: \hat{g}(X_i)\geq 0\}\label{rejectsalot},\\
|\mbox{AdaDetect}_\alpha\cap \cH_1|/m_1&\geq 1-{R}_1(g^\sharp_\cG)-\Delta,\label{highpower}
\end{align}
where $\mbox{AdaDetect}_\alpha$ denotes the rejection set of AdaDetect with score function $\hat{g}$. 
\end{itemize}
\end{theorem}

The proof of Theorem~\ref{powerhighdim} is presented in Section~\ref{sec:proofpowerhighdim}. The idea is to show that there are many alternatives with a nonnegative score, while there are only a few true nulls with nonnegative scores. This yields small empirical $p$-values for hypotheses with nonnegative scores, which implies that the procedure $\mbox{AdaDetect}_\alpha$ detects these nonnegative scores, see Lemma~\ref{positivescoreBH}. 

Theorem \ref{powerhighdim} (i) shows that $\hat{g}$ has a similar classification accuracy to $g^\sharp_\cG$ on both the NTS and mixed sample. It is analogous to Theorem~2 in \cite{BLS2010}, though \cite{BLS2010} considers a different setting where the proportion of nulls is random. Theorem \ref{powerhighdim} (ii) entails that, with high probability, all hypotheses with nonnegative scores will be rejected and the power of AdaDetect with $\hat{g}$ is nearly as large as the power of the classification procedure given by $g_\cG^\sharp$. 

Note that the Lagrangian form of the above problem is in the form of the weighted loss defined in \eqref{equJlambda2}. Thus, by Lemma~\ref{lem:PUclassiforacle} (i), there exists $\lambda_\beta > 0$ such that 
$g_{\cG}^\sharp(x)=g^{*}(x) = \frac{\lambda_\beta (\ell+m)}{k} \frac{f_\gamma(x)}{f_0(x)}-1,$
if the constraint is feasible and $\cG$ is sufficiently rich to include the above function. Above, $g^{*}(x)$ satisfies \eqref{equoraclescore} and hence yields the optimal power. If we define $b=R_1(g^\sharp_\cG)-{R}_1(g^*)$ as the bias due to the constraint, Theorem \ref{powerhighdim} (ii) implies that, with probability $1- \delta$, the power of AdaDetect with $\hat{g}$ is at most $\Delta + b$ below the optimal power. Thus, the function class $\cG$ incurs a tradeoff that a richer class yields a smaller $b$ but a larger $\Delta$ and vice versa.

Here, we aim at making $\cG$ as flexible as possible while ensuring $\Delta = o(1)$. When $k, \ell$, and $m$ are of the same order and $\delta$ is a constant,
$
\Delta \asymp  \frac{m}{m_1} \sqrt{\frac{V(\cG) }{m}}.
$
Hence, $\Delta=o(1)$ if 
\begin{equation}\label{condition-vanishing}
\frac{m_1}{m} \gg  \sqrt{\frac{V(\cG) }{m}}.
\end{equation}
For illustration, consider the class $\cG_{N,L,s}$ of ReLU feed forward neural networks with fixed topology, maximum width $N \asymp m^c$ ($c\in (0,1)$), depth $L\asymp \log m$ and sparsity $s\asymp N \log m$. \cite{bartlett2019nearly} show that
$
V(\cG_{N,L,s})\leq 2sL\log(4eN) \lesssim m^c (\log m)^3.
$
Hence, condition \eqref{condition-vanishing} reads in this case
$
m_1/m \gg  m^{\frac{c-1}{2}} (\log m)^{3/2}.
$
This implies that $\Delta = o(1)$ unless the novelties are too sparse. 
On the other hand, given the approximation ability of class of neural networks, we should expect $1-{R}_1(g^\sharp_{\cG_{N,L,s}})\approx 1-{R}_1(g^*)$. Thus, Theorem~\ref{powerhighdim} (ii) implies the resulting score function is nearly optimal.

\begin{remark}[Choice of $\beta$]\label{rem:beta}
Since $\mbox{AdaDetect}_\alpha$ controls FDR at level $\alpha$, it is necessary to impose an upper bound on $\beta$ in Theorem~\ref{powerhighdim} (ii). Roughly speaking, our condition on $\beta$ guarantees that the classifier $g^\sharp_\cG$ controls the FDR at level $\alpha$, up to remainder terms. 
\end{remark}

\begin{remark}
The condition on $\ell$ in \eqref{equlconstraint} is needed to ensure that the minimum value $1/(1+\ell)$ that $p$-values can take is sufficiently small so that the BH procedure can reject. A similar condition was introduced in \cite{mary2021semisupervised}, see also Remark~\ref{rem:samplesize}. 
\end{remark}

\subsection{General score functions}\label{sec:genresultpower}

Now we move to general score functions. Let $g^*$ be any measurable function $\R^d\to \R$ in the form of \eqref{equoraclescore} and 
\begin{align}
\ol{G}_0(s) &= \P_{X\sim P_0}(g^*(X) \geq s),\quad s\in \R\label{equGbar};\\
\zeta_{r}(\eta)&=\max_{u\in [\alpha (r\vee 1)/m,\alpha]} \left\{\frac{\ol{G}_0(\ol{G}_0^{-1}(u)-2\eta) - u}{u }\right\},\quad \eta>0, r\in \{0,\dots,m\}.\label{funcdelta}
\end{align} 
Here, $\zeta_r(\cdot)$ measures the local fluctuation of $\ol{G}_0$. We suppress the dependence on $\alpha$ and $m$ to ease notation. Furthermore, consider any data-driven score function $\hat{g}$ satisfying the condition \eqref{constrainedg} and let
\begin{align}
\hat{\eta}&=\max_{k+1\leq i\leq n+m} |\hat{g}(Z_i; (Z_1,\dots,Z_k) , (Z_{k+1},\dots,Z_{n+m}))- g^*(Z_i) |,\label{diffscore}
\end{align}
which measures the maximal discrepancy of scores in the mixed sample. In the following, $\mbox{AdaDetect}_{\alpha}$ denotes the procedure with score function $\hat{g}$ and $\mbox{AdaDetect}_{\alpha}^{*}$ denotes the procedure with score function $g^{*}$. 

\begin{theorem}\label{th:BHestimated}
Fix any $r\in \{0,\dots,m\}$ and let $\mathcal{R}=\{|\mbox{AdaDetect}^*_\alpha|\geq r\}$. Assume $m\geq 1$, $\l, k\geq 0$, $n=k+\l\geq 1$, and $\ol{G}_0$ \eqref{equGbar} is continuous and strictly decreasing. Under Assumptions~\ref{as:indep} and~\ref{equ-marg}, for any $\delta,\eta\in (0,1)$ such that $ (\ell+1)\delta\alpha(r\vee 1)/m\geq 2$,
\begin{align}\label{powerBONuSstar}
\P\left(\mathcal{R} \cap\{\mbox{AdaDetect}^*_\alpha\subset \mbox{AdaDetect}_{\alpha'}\}^c\right)\leq \P\left(\hat{\eta}> \eta \right) + 2 m e^{-(3/28)  (\ell+1)\delta^2\alpha(r\vee 1)/m},
\end{align}
where $\alpha'=\alpha(1+3\delta)(1+\zeta_r(\eta))$ and $\zeta_r(\cdot)$ and $\hat{\eta}$ are defined in \eqref{funcdelta} and \eqref{diffscore}, respectively.  Furthermore, \eqref{powerBONuSstar} is also true with $\mbox{AdaDetect}^*_\alpha$ replaced by $\BH^*_\alpha$, the BH algorithm applied to the oracle $p$-values $p^*_i=\ol{G}_0(g^*(X_i))$, $1\leq i\leq m$.
\end{theorem}

The proof is presented in Section~\ref{proof:BHestimated}. The condition $ (\ell+1)\delta\alpha(r\vee 1)/m\geq 2$ is analogous to the one studied \citep{mary2021semisupervised} for fixed score functions. 
When we choose $r=0$, we have $\P(\mathcal{R})=1$ and thus \eqref{powerBONuSstar} implies
$$\P\left(\mbox{AdaDetect}^*_\alpha\subset \mbox{AdaDetect}_{\alpha'}\right)\geq 1 - \P\left(\hat{\eta}> \eta \right) - 2 m e^{-(3/28)  (\ell+1)\delta^2\alpha /m},$$
for any $\delta$ with $(\ell + 1)\delta \alpha / m\ge 2$. If $\ell / m >\!\!> \log m$, choosing $\delta = o(1)$ such that $(\ell+1)\delta^2\alpha /m >\!\!>\log m$ and $\eta$ such that $\P(\hat{\eta}\ge \eta) = o(1)$, we have
$\P\left(\mbox{AdaDetect}^*_\alpha\subset \mbox{AdaDetect}_{\alpha'}\right) = 1 - o(1),$
where $\alpha' = \alpha(1 + \zeta_0(\eta))(1 + o(1)).$ Thus, when $\zeta_0(\eta)$ is small, we show that AdaDetect with the estimated score function and slight inflation of the target level is strictly more powerful than its oracle version. 

In general, when $|\mbox{AdaDetect}_{\alpha}^{*}|$ is larger with high probability, we can choose a larger $r$ to relax the condition on $\delta$, reduces $\zeta_{r}(\eta)$ (and hence $\alpha'$), and improve the RHS of \eqref{powerBONuSstar}. In particular, we can set $r$ appropriately to obtain the following result on the asymptotic TDR. 

\begin{corollary}\label{cor:BHestimated}
Consider the setting of Theorem~\ref{th:BHestimated}. Fix any $\epsilon > 0$. Assume $m_1 \ge 1$ and $ (\ell+1)\delta\alpha\lceil m_1\epsilon\rceil/m\geq 2$. Then
\begin{align}\label{powerBONuSstarTDR}
 \TDR(\mbox{AdaDetect}_{\alpha'}) \geq \TDR(\mbox{AdaDetect}^*_\alpha)  -\P\left(\hat{\eta}> \eta \right)-2 m e^{-(3/28)  (\ell+1)\delta^2\alpha\lceil m_1\epsilon\rceil/m} - \epsilon,
\end{align}
 where $\alpha'=\alpha(1+3\delta)(1+\zeta_{\lceil m_1\epsilon\rceil }(\eta))$.
 In particular, if there exist sequences $\delta=\delta(k,\l,m,m_1)$, $\epsilon=\epsilon(k,\l,m,m_1)$, and $\eta=\eta(k,\l,m,m_1)$ such that, as $\l,m,m_1$ tend to infinity, 
 \begin{equation}\label{equ:condparaconsist}
\mbox{ $\delta,\epsilon\to 0$, $\ell\delta^2 \epsilon m_1/m\to \infty$, $ \P\left(\hat{\eta}> \eta \right)\to 0$ and $\zeta_{\lceil m_1\epsilon\rceil }(\eta)\to 0$, }
\end{equation}
then
\begin{equation}\label{equconsistency}
\liminf_{\l,m,m_1} \left\{\TDR(\mbox{AdaDetect}_{\tilde{\alpha}}) - \TDR(\mbox{AdaDetect}^*_\alpha) \right\}\geq 0,\:\:\: \mbox{ for any fixed }\tilde{\alpha} > \alpha.
\end{equation}
Furthermore, these results hold with $\mbox{AdaDetect}^*_\alpha$ replaced by $\BH^*_\alpha$ defined in Theorem \ref{th:BHestimated}.
\end{corollary}

The proof is presented in Section~\ref{proof:cor:BHestimated}. Corollary \ref{cor:BHestimated} shows that AdaDetect is nearly as powerful as the oracle version, as well as the BH procedure with the optimal score. 

Now we discuss the choice of $\eta$. Note that $\eta$ is a parameter that only shows up in the bound but not in the algorithm. It incurs a tradeoff that a larger $\eta$ would improve the tail bound by decreasing $\P\left(\hat{\eta}> \eta \right)$ but inflate $\alpha'$ through increasing $\zeta_r(\eta)$. Ideally, we would want $\eta$ so that $\P\left(\hat{\eta}> \eta \right)$ and $\zeta_r(\eta)$ are both negligible. For illustration, assume 
\begin{equation}\label{condetachap}
\P\left(\hat{\eta}> (n+m)^{-\kappa} \right)=o(1),
\end{equation}
for some $\kappa \in (0, 1/2)$ and $\zeta_r(\eta)\lesssim \eta / \gamma$, where $\gamma$ is defined in \eqref{equfgamma}. In this case, $\P\left(\hat{\eta}> \eta \right)$ and $\zeta_r(\eta)$ are both $o(1)$ if
$(n+m)^{-\kappa} = o(\gamma) = o(m_1/(m+\ell)).$ 
Again, this would hold unless the novelties are too sparse.

We show in Lemma~\ref{lem:etachapdensityestimation} that the score function given by density estimation satisfies \eqref{condetachap} under regularity conditions. Another example is given by Theorem~3.2 in \cite{audibert2007fast} in the case where $g^*$ is the posterior probability under a different setting; see Section~\ref{sec:mr}). For $\zeta_r(\eta)$, we provide bounds in Section~\ref{sec:boundzeta} for two examples. In the Gaussian example, we show that $\zeta_r(\eta)\lesssim \eta / \gamma$, where $\gamma$ is defined in \eqref{equfgamma}. 

\begin{remark}\label{rem:improvpower}
Theorem~3 in \cite{yang2021bonus} provides another asymptotic power analysis  
 showing that the symmetric difference between the rejection set for the data-driven score function and its oracle version has a size $o_P(m)$. Unlike Theorem \ref{th:BHestimated} and  Corollary \ref{cor:BHestimated}, it does not have implications when the oracle procedure can only reject $o_P(m)$ hypotheses, as in the case where $m_1/m = o(1)$.
\end{remark}

\section{Experiments}\label{sec:numexp}\label{sec:realdata}



In this section, we examine the performance of AdaDetect on real data (experiments on simulated data are postponed to Section~\ref{sec:simdata}). 
Each dataset contains measurements that are labeled as either typical or novelty. We summarize the datasets in Table~\ref{tab:datasets}  and provide further details in Section~\ref{sec:descriptiondataset}. 

We apply AdaDetect with various score functions, including the density estimation-based score (\texttt{AdaDetect parametric} and \texttt{AdaDetect KDE}), the PU classification-based score (\texttt{AdaDetect SVM}, \texttt{AdaDetect RF}, \texttt{AdaDetect NN}, and \texttt{AdaDetect NN cv}). We also include the conformal novelty detection procedures proposed by \cite{bates2021testing} (\texttt{CAD SVM} and \texttt{CAD IForest}). Note that both \texttt{CAD SVM} and \texttt{CAD IForest} are instances of AdaDetect with one-class classification-based scores. For all our experiments, we use the Python package \texttt{scikit-learn} for Expectation Maximization (EM) algorithm, kernel density estimation, random forests, and neural networks, with the default hyper-parameters from the packages unless otherwise specified (a full description of these methods is provided in Section~\ref{sec:simumethod}).
For the MNIST dataset, we consider two more methods based on a convolutional neural network (CNN), with two convolution layers and one fully connected layer. The first method \texttt{CAD SVDD CNN} is the conformal novelty detection procedure of  \cite{bates2021testing} with a special one-class classifier, given by the Support Vector Data Description (SVDD) method introduced in \cite{ruff18a} used with a family of functions given by the CNN. The second method is AdaDetect with the two-class classifier based on the CNN, denoted by \texttt{AdaDetect CNN}.

We construct test samples and null training samples by subsampling the dataset with $n=5000$, $m = 1000$, and a fixed null proportion $\pi_0=m_0/m = 0.9$. For AdaDetect, we choose $k=4m$, $\ell=m$, $s=3m$ for cross-validation, and the target level $\alpha=0.1$.
The FDR and TDR for the methods are evaluated by using $100$ runs and the results are reported in Table \ref{tab:datasetsperf}.  As expected, all methods control the FDR. Compared to \cite{bates2021testing}, AdaDetect with classification-based scores substantially boosts the power because it incorporates the novelties in learning the score function. Overall, the best performing method is \texttt{AdaDetect RF}, with \texttt{AdaDetect NN} (possibly cross-validated) coming in second. \texttt{AdaDetect CNN} is particularly efficient on the classical MNIST dataset, which is unsurprising because CNN-type classifiers are appropriate for such an image dataset \citep{goodfellow2016}. \rev{We however note that the one-class classifier based upon CNN behaves poorly, which shows that two-class classification is the key for the power boost instead of the better representation given by CNN}.
\rev{In addition, further comparisons are made in Section~\ref{sec:addnumexp} for other values of $n,m,m_1$ in more challenging regimes and the conclusions are qualitatively similar. }

To conclude, 
if a classification method is expected to distinguish between typical and anomalous measurements, combining it with AdaDetect is expected to achieve high power without threatening FDR control.

\begin{table}
\caption{Summary of datasets. \label{tab:datasets}}
 \begin{tabular}{lcccccc}
& Shuttle & Credit card & KDDCup99 & Mammography & Musk & MNIST  \\
\hline
Dimension $d$ & 9 & 30 & 40 & 6 & 166 & $28\times28$  \\
Feature type & Real & Real & \makecell{Real, \\ categorical} & Real & Real & Real \\
Inliers & 45586 & 284315 & 47913 & 10923 & 5581 & 5842 \\
Novelties  & 3511 & 492 & 200 & 260 & 1017 & 5949 \\
\hline
\end{tabular}
\end{table}

\begin{table}
\caption{FDR (top) and TDR (bottom) of AdaDetect with different score functions on real datasets. The target FDR level is $\alpha=0.1$. We report the mean value and the standard deviation (in brackets) over 100 runs. The two best-performing methods are highlighted in bold. \label{tab:datasetsperf}}
{\tiny\begin{tabular}{lcccccc}
& Shuttle & Credit card & KDDCup99 &  Mammography & Musk & MNIST \\
\hline
\hline
& \multicolumn{6}{c}{FDR} \\
\hline
\hline
CAD SVM & 0.04 (0.08) & 0.00 (0.00) & 0.00 (0.000) & 0.05(0.10) & 0.00 (0.00) &  0.00 (0.00)\\
CAD IForest & 0.10 (0.07) & 0.09 (0.06) & 0.08 (0.07) & 0.05 (0.09) & 0.00 (0.00) & 0.00 (0.00) \\
AdaDetect parametric & 0.01 (0.05) & 0.00 (0.00) & 0.00 (0.00) & 0.00 (0.00) & 0.00 (0.00) &  0.00 (0.00)\\
AdaDetect KDE & 0.07 (0.07) & 0.05 (0.08) & 0.02 (0.06) & 0.08 (0.07) & 0.02 (0.08) &  0.00 (0.00)\\
AdaDetect SVM &  0.08 (0.04) & 0.07 (0.05) &  0.07 (0.05) & 0.07 (0.06) & 0.08 (0.06) & 0.02 (0.03)   \\
AdaDetect RF & 0.08 (0.04) & 0.09 (0.04) & 0.08 (0.04) & 0.04 (0.10) & 0.03 (0.06) & 0.03 (0.07) \\
AdaDetect NN & 0.07 (0.05) & 0.09 (0.04) & 0.06 (0.07) & 0.09 (0.06) & 0.06 (0.09) & 0.06 (0.08) \\
AdaDetect cv NN & 0.08 (0.04) & 0.09 (0.05)& 0.08 (0.11) & 0.08 (0.05) & 0.06 (0.08) & 0.01 (0.03) \\
\hline
CAD SVDD + CNN &  - & - & - &  - & - & 0.03 (0.14) \\
AdaDetect CNN &  - & - & - &  - & - & 0.09 (0.05) \\
\hline
\hline
& \multicolumn{6}{c}{TDR} \\
\hline
\hline
CAD SVM & 0.10 (0.18) & 0.00 (0.00) & 0.00 (0.00) & 0.03 (0.06) & 0.00 (0.00) & 0.00 (0.00) \\
CAD IForest & 0.45	 (0.09) & 0.39 (0.22) & 0.56 (0.35) & 0.05 (0.09) & 0.00 (0.00) & 0.00 (0.00)\\
AdaDetect parametric & 0.02 (0.07) & 0.00 (0.00) & 0.00 (0.00) & 0.07 (0.09) & 0.00 (0.00)& 0.00 (0.00)\\
AdaDetect KDE & 0.44 (0.33) & 0.12 (0.20) & 0.11 (0.24) & 0.22 (0.17) & 0.02 (0.06) & 0.00 (0.00)\\
AdaDetect SVM &  \textbf{0.85 (0.17)} & 0.68 (0.28) &  0.66 (0.32) & 0.43 (0.13) &  \textbf{0.40 (0.17)}  & \textbf{0.52 (0.21)} \\
AdaDetect RF & \textbf{0.99 (0.01)} & \textbf{0.85 (0.03)} & \textbf{0.99 (0.01)} & \textbf{0.48 (0.10)} & 0.04 (0.09) & 0.03 (0.08) \\
AdaDetect NN & 0.76 (0.15) & \textbf{0.80 (0.07)} & 0.52 (0.41) & \textbf{0.47 (0.14)} &  0.11 (0.13) & 0.01 (0.03) \\
AdaDetect cv NN & 0.84 (0.12) & 0.76 (0.13) & \textbf{0.74 (0.41)} & 0.42 (0.16) & \textbf{0.13 (0.12)} & 0.01 (0.03)  \\
\hline
CAD SVDD + CNN &  - & - & - &  - & - & 0.03 (0.15) \\
AdaDetect CNN & - & - & - & - & - & \textbf{0.93 (0.06)} \\
\hline
\end{tabular}
}
\end{table}

\section{An astronomy application}
\label{sec:appli}

In this section, we apply AdaDetect to detect variable stars using the Sloan Digital Sky Survey \citep{ivezic2005selection}, a large labeled dataset with $92,658$ nonvariable (null) and $483$ variable (novelties) stars. Each star is encoded as a $4$-dimensional vector containing the star's flux in specific bands (colors) of the visible light. This dataset is particularly appealing for demonstrating our method. First, the two classes occupy similar regions in the considered color space, with slight overlap leading to complex decision boundaries. Second, the large number of nonvariable stars allows us to vary the size of the NTS in a large range in the Monte-Carlo simulations. Third, this dataset has been extensively studied by astronomers and has become a standard for benchmarking classification methods (see Chapter~9 of  \citealp{ivezic2019statistics}). Lastly, we can compute the achieved FDR and TDR for any novelty detection method based on the labeled data.

For each experiment, we sample $n$ nonvariable stars as the NTS along with $m_1$ variable stars and $m_0 = m - m_1$ additional nonvariable stars as the test sample. We set $m = 100$ and vary $n$ and $m_1$ across experiments. We apply AdaDetect with Kernel Density Estimation (KDE), Random Forest (RF), and Neural Networks (NN). For comparison, we also include two Empirical BH procedures \citep{mary2021semisupervised}, which are special cases of AdaDetect with non-adaptive scores as the squared $\ell_2$ norm of the demeaned vectors, where the mean is calculated on all nulls outside of the NTS (``Emp BH full'') and on the NTS (``Emp BH current''), respectively. The ``Emp BH current'' method is closer to the current practice, though it is not granted to control the FDR since the score function does not satisfy \eqref{constrainedg}. In addition, we apply the Empirical BH procedure without demeaning the data as well as the SC procedure with estimated local FDR. Neither detects any novelties so we will not report them. 

Figure \ref{app1} presents the results for $m_1 = 50$ and varying $n$ with target FDR level $\alpha = 0.05$. The FDR and TDR are calculated based on $100$ Monte-Carlo simulations. To aid visualization, we represents the uncertainty by a shaded area whose width is equal to the standard error of estimated FDR/TDR divided by $10$. This can be viewed as an approximation of the standard error with $10,000$ Monte-Carlo simulations. In this setting, $\pi_0 = 0.5$ and thus all methods provably control the FDR at level $\pi_0\alpha = 0.025$ (except ``Emp BH current''). This is confirmed in the left panel of Figure \ref{app1}. From the right panel, we observe that AdaDetect with RF achieves the highest power, substantially improving upon AdaDetect with non-adaptive scores (Emp BH). This demonstrates the advantages of utilizing classification-based score functions. In Section \ref{sec:appliplus}, we present results in additional experimental settings that exhibit qualitative similarities to Figure~\ref{app1}.

\begin{figure}[h!]
\centerline{
\includegraphics[scale=0.5]{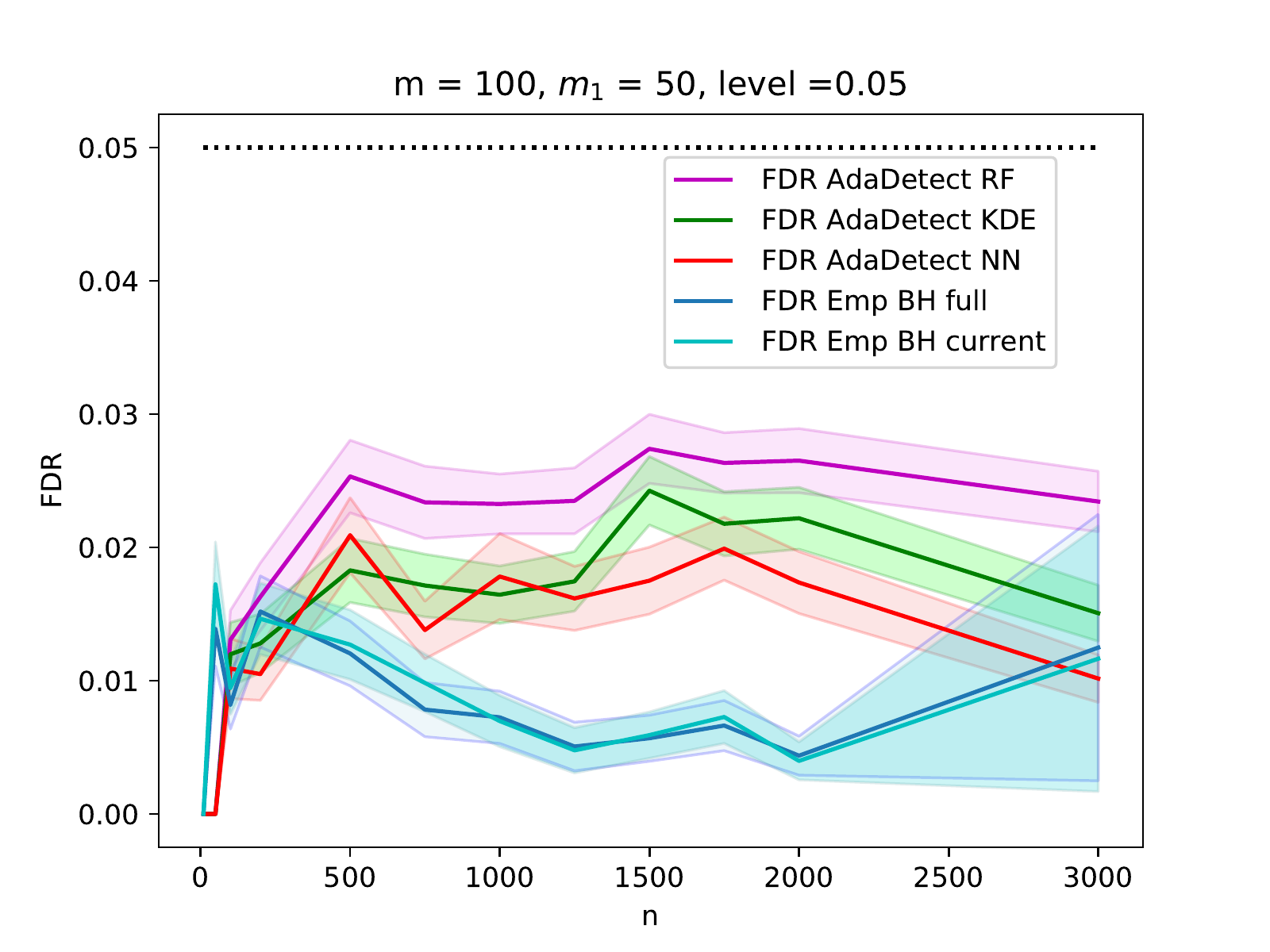}
\includegraphics[scale=0.5]{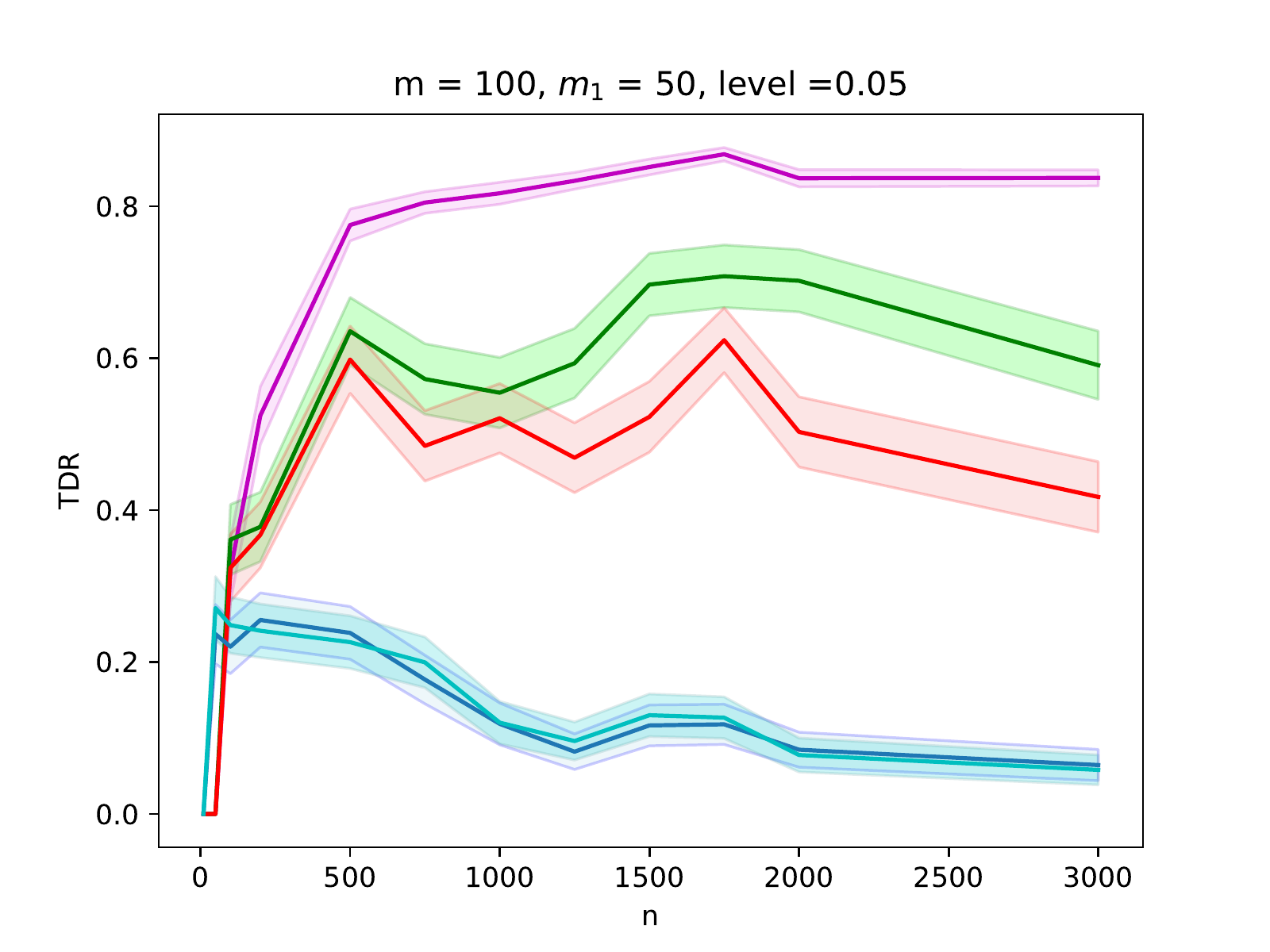}
}
\caption{Estimated FDR (left) and TDR (right) as a function of $n$, the size of the NTS, with $m = 100, m_1 = 50$ and $\alpha = 0.05$. All methods shown in the plot provably control the FDR at level $\pi_0\alpha = 0.025$ (except ``Emp BH current'').}
\label{app1}
\end{figure}

 

\section{Conclusion and discussion}\label{sec:discussion}


\rev{In this work, we propose AdaDetect as a generic framework that can wrap around any classification methods and provably control the FDR in finite samples when the null measurements are exchangeable conditional on the novelties. It generalizes and often substantially outperforms previous methods that only work with one-class classification methods which are not adaptive to the novelty distribution. We also develop the $\pi_0$-adaptive AdaDetect that further improves the power in the presence of many novelties as well as the cross-validated AdaDetect that allows model selection. The theoretical analysis is based on a novel FDR expression that unifies and generalizes the existing results. In addition, we provide power analysis showing that (1) the optimal score function is given by any monotonic transformation of the ratio between the average density of novelties and the null density and (2) the estimated score function can be asymptotically optimal in terms of power. We demonstrate the versatility of AdaDetect on a variety of tasks.}

\subsection{\rev{Limitations of AdaDetect}}

\rev{Here we discuss several limitations of our method and potential solutions.
\begin{itemize}
\item Heterogeneous null distributions. A key assumption for the FDR control is that the null distribution $P_0$ is the same across the NTS and the test sample. This excludes the case where the null can be generated from a bag 
  of distributions $\{P_{0,k},1\leq k\leq K\}$. Under heterogeneity, the empirical $p$-values can be invalid even marginally since the nulls are no longer exchangeable. One possible way to reconcile this issue is to assume the nulls are generated from a mixture distribution $\sum_{k=1}^{K}\pi_k P_{0,k}$, thereby retaining the exchangeability. We leave this for future research.
\item Directional null hypotheses. Throughout the paper we focus on testing whether a new observation has the same distribution as the typical measurements. In some applications, it may be more appropriate to test directional nulls, which are often characterized by the sign of a parameter for parametric models. However, it is unclear how this can be done in nonparametric cases. One possibility is to consider the nulls $P_i\preceq P_0$ where $\preceq$ denotes the stochastic dominance, meaning that there exists a random vector $(A, B)$ such that $A\sim P_0, B\sim P_i$ and $A\ge B$ in an entrywise fashion. 
By restricting the score function to be entrywise increasing, we may still apply AdaDetect and retain the FDR control. 
\item Randomness of data splitting. AdaDetect is intrinsically randomized due to the data splitting step. Without carefully documenting random seeds, the researcher can ``hack'' the results by reporting the best results across different splits. A subsequent work by \cite{bashari2023derandomized} proposes an elegant solution to derandomize AdaDetect by treating the test statistics as $e$-values and aggregating over all data splits. They show that the $e$-AdaDetect successfully stabilizes the output of AdaDetect.
\item Semi-supervised data. In some applications, labeled novelties are available in the training sample. For example, the researcher may have historical data on fraud transactions recorded in the system and can train a two-class classifier to distinguish between the nulls and labeled novelties. When future novelties are similar to labeled novelties, it should yield an efficient score function. This has been studied by \cite{liang2022integrative}. Combining their approach with ours in a nonstationary setting where future novelties behave differently from the past ones is a promising avenue for future research. 
\item Sparse novelties: when the novelties are too sparse, two-class classifiers may not be the best at discriminating between nominals and novelties and can be out-performed by simpler one-class classifiers; see \cite{liang2022integrative}. One possible solution is to apply AdaDetect cv by including both one-class and two-class classification methods and let data decide which score function is more efficient. We leave the full examination of this approach for future research. 
\end{itemize}
}


\subsection{Other future works}

First, we could provide a more detailed power analysis by quantifying the bias term  ${R}_1(g^\sharp_{\cG})- {R}_1(g^\sharp)$ for a broader class of algorithms. For example, we can consider $\cG=\cG_{N,L,s}$, the set of realizations of NN with width $N$, depth $L$ and sparsity $s$ \citep{bos2021convergence}. Such a quantitative analysis could provide guidelines for choosing hyper-parameters or at least a default range in the cross-validated AdaDetect procedure. 

Next, a core assumption of the FDR controlling theory is exchangeability of the null scores conditional on the novelties. This can be satisfied beyond our setting, e.g., the knockoff setting discussed in Remark~\ref{rem:lasso}. This suggests a possible path to further improve the knockoffs method.

Lastly, the BONuS algorithm in \cite{yang2021bonus} can iteratively remove null observations and update the score function correspondingly using a masking technique introduced by \cite{lei2018adapt}. While this increases the computation cost, it gradually reduces the attenuation caused by the null sample in the mixed sample and hence improves the accuracy of the estimated score function. It would be interesting to apply their idea in AdaDetect. 


\section*{Acknowledgements}

We would like to thank Gilles Blanchard, Will Fithian, Aaditya Ramdas, Fanny Villers  and  Asaf Weinstein for constructive discussions and feedbacks. 
The authors acknowledge the grants ANR-16-CE40-0019 (project SansSouci), ANR-17-CE40-0001 (BASICS)  and ANR-21-CE23-0035 (ASCAI) of the French National Research Agency ANR and the GDR ISIS through the ``projets exploratoires" program (project TASTY). 

\bibliographystyle{apalike}
\bibliography{biblio}

\begin{thebibliography}{}

\bibitem[Angelopoulos and Bates, 2021]{angelopoulos2021gentle}
Angelopoulos, A.~N. and Bates, S. (2021).
\newblock A gentle introduction to conformal prediction and distribution-free
  uncertainty quantification.
\newblock {\em arXiv preprint arXiv:2107.07511}.

\bibitem[Audibert and Tsybakov, 2007]{audibert2007fast}
Audibert, J.-Y. and Tsybakov, A.~B. (2007).
\newblock Fast learning rates for plug-in classifiers.
\newblock {\em The Annals of statistics}, 35(2):608--633.

\bibitem[Balasubramanian et~al., 2014]{balasubramanian2014conformal}
Balasubramanian, V., Ho, S.-S., and Vovk, V. (2014).
\newblock {\em Conformal prediction for reliable machine learning: theory,
  adaptations and applications}.
\newblock Newnes.

\bibitem[Barber and Cand\`es, 2015]{BC2015}
Barber, R.~F. and Cand\`es, E.~J. (2015).
\newblock Controlling the false discovery rate via knockoffs.
\newblock {\em Ann. Statist.}, 43(5):2055--2085.

\bibitem[Barber et~al., 2021]{barber2021predictive}
Barber, R.~F., Candes, E.~J., Ramdas, A., and Tibshirani, R.~J. (2021).
\newblock Predictive inference with the jackknife+.
\newblock {\em The Annals of Statistics}, 49(1):486--507.

\bibitem[{Barber} et~al., 2020]{Barber2020robust}
{Barber}, R.~F., {Cand\`es}, E.~J., and {Samworth}, R.~J. (2020).
\newblock {Robust inference with knockoffs}.
\newblock {\em {Ann. Stat.}}, 48(3):1409--1431.

\bibitem[Bartlett et~al., 2019]{bartlett2019nearly}
Bartlett, P.~L., Harvey, N., Liaw, C., and Mehrabian, A. (2019).
\newblock Nearly-tight vc-dimension and pseudodimension bounds for piecewise
  linear neural networks.
\newblock {\em The Journal of Machine Learning Research}, 20(1):2285--2301.

\bibitem[Bashari et~al., 2023]{bashari2023derandomized}
Bashari, M., Epstein, A., Romano, Y., and Sesia, M. (2023).
\newblock Derandomized novelty detection with fdr control via conformal
  e-values.

\bibitem[Bates et~al., 2021]{bates2021testing}
Bates, S., Cand\`es, E., Lei, L., Romano, Y., and Sesia, M. (2021).
\newblock Testing for outliers with conformal p-values.

\bibitem[Bekker and Davis, 2020]{bekker2020learning}
Bekker, J. and Davis, J. (2020).
\newblock Learning from positive and unlabeled data: A survey.
\newblock {\em Machine Learning}, 109(4):719--760.

\bibitem[Benjamini, 2010]{benjamini2010discovering}
Benjamini, Y. (2010).
\newblock Discovering the false discovery rate.
\newblock {\em Journal of the Royal Statistical Society: series B (statistical
  methodology)}, 72(4):405--416.

\bibitem[Benjamini and Hochberg, 1995]{BH1995}
Benjamini, Y. and Hochberg, Y. (1995).
\newblock Controlling the false discovery rate: a practical and powerful
  approach to multiple testing.
\newblock {\em J. Roy. Statist. Soc. Ser. B}, 57(1):289--300.

\bibitem[Benjamini et~al., 2006]{BKY2006}
Benjamini, Y., Krieger, A.~M., and Yekutieli, D. (2006).
\newblock Adaptive linear step-up procedures that control the false discovery
  rate.
\newblock {\em Biometrika}, 93(3):491--507.

\bibitem[Benjamini and Yekutieli, 2001]{BY2001}
Benjamini, Y. and Yekutieli, D. (2001).
\newblock The control of the false discovery rate in multiple testing under
  dependency.
\newblock {\em Ann. Statist.}, 29(4):1165--1188.

\bibitem[Blanchard et~al., 2010]{BLS2010}
Blanchard, G., Lee, G., and Scott, C. (2010).
\newblock Semi-supervised novelty detection.
\newblock {\em J. Mach. Learn. Res.}, 11:2973--3009.

\bibitem[Blanchard and Roquain, 2009]{BR2009}
Blanchard, G. and Roquain, E. (2009).
\newblock Adaptive false discovery rate control under independence and
  dependence.
\newblock {\em J. Mach. Learn. Res.}, 10:2837--2871.

\bibitem[Bogdan et~al., 2015]{Slope2015}
Bogdan, M., van~den Berg, E., Sabatti, C., Su, W., and Cand\`es, E.~J. (2015).
\newblock S{LOPE}---adaptive variable selection via convex optimization.
\newblock {\em Ann. Appl. Stat.}, 9(3):1103--1140.

\bibitem[Bos and Schmidt-Hieber, 2021]{bos2021convergence}
Bos, T. and Schmidt-Hieber, J. (2021).
\newblock Convergence rates of deep relu networks for multiclass
  classification.

\bibitem[Cai et~al., 2019]{CARS2019}
Cai, T., Sun, W., and Wang, W. (2019).
\newblock Covariate-assisted ranking and screening for large-scale two-sample
  inference.
\newblock {\em Journal of the Royal Statistical Society: Series B},
  81(2):187--234.

\bibitem[Cai and Sun, 2009]{CS2009}
Cai, T.~T. and Sun, W. (2009).
\newblock Simultaneous testing of grouped hypotheses: finding needles in
  multiple haystacks.
\newblock {\em J. Amer. Statist. Assoc.}, 104(488):1467--1481.

\bibitem[Calvo et~al., 2007]{calvo2007learning}
Calvo, B., Larranaga, P., and Lozano, J.~A. (2007).
\newblock Learning bayesian classifiers from positive and unlabeled examples.
\newblock {\em Pattern Recognition Letters}, 28(16):2375--2384.

\bibitem[Cannon et~al., 2002]{cannon2002learning}
Cannon, A., Howse, J., Hush, D., and Scovel, C. (2002).
\newblock Learning with the neyman-pearson and min-max criteria.
\newblock {\em Los Alamos National Laboratory, Tech. Rep. LA-UR}, pages
  02--2951.

\bibitem[Du~Plessis et~al., 2014]{du2014analysis}
Du~Plessis, M.~C., Niu, G., and Sugiyama, M. (2014).
\newblock Analysis of learning from positive and unlabeled data.
\newblock {\em Advances in neural information processing systems}, 27:703--711.

\bibitem[{Efron}, 2004]{Efron2004}
{Efron}, B. (2004).
\newblock {Large-scale simultaneous hypothesis testing: the choice of a null
  hypothesis.}
\newblock {\em {J. Am. Stat. Assoc.}}, 99(465):96--104.

\bibitem[Efron, 2007]{Efron2007b}
Efron, B. (2007).
\newblock Doing thousands of hypothesis tests at the same time.
\newblock {\em Metron - International Journal of Statistics}, LXV(1):3--21.

\bibitem[Efron, 2008]{Efron2008}
Efron, B. (2008).
\newblock Microarrays, empirical {B}ayes and the two-groups model.
\newblock {\em Statist. Sci.}, 23(1):1--22.

\bibitem[{Efron}, 2009]{Efron2009b}
{Efron}, B. (2009).
\newblock {Empirical Bayes estimates for large-scale prediction problems.}
\newblock {\em {J. Am. Stat. Assoc.}}, 104(487):1015--1028.

\bibitem[Efron et~al., 2001]{ETST2001}
Efron, B., Tibshirani, R., Storey, J.~D., and Tusher, V. (2001).
\newblock Empirical {B}ayes analysis of a microarray experiment.
\newblock {\em J. Amer. Statist. Assoc.}, 96(456):1151--1160.

\bibitem[Ferreira and Zwinderman, 2006]{FZ2006}
Ferreira, J.~A. and Zwinderman, A.~H. (2006).
\newblock On the {B}enjamini-{H}ochberg method.
\newblock {\em Ann. Statist.}, 34(4):1827--1849.

\bibitem[Fisher, 2021]{fisher2021saffron}
Fisher, A. (2021).
\newblock Saffron and lord ensure online control of the false discovery rate
  under positive dependence.
\newblock {\em arXiv preprint arXiv:2110.08161}.

\bibitem[Fithian and Lei, 2020]{fithian2020conditional}
Fithian, W. and Lei, L. (2020).
\newblock Conditional calibration for false discovery rate control under
  dependence.

\bibitem[Foygel~Barber and Ramdas, 2015]{foygel2015p}
Foygel~Barber, R. and Ramdas, A. (2015).
\newblock The p-filter: multi-layer fdr control for grouped hypotheses.
\newblock {\em ArXiv e-prints}, pages arXiv--1512.

\bibitem[Friedman, 2003]{friedman2003multivariate}
Friedman, J.~H. (2003).
\newblock On multivariate goodness--of--fit and two--sample testing.
\newblock {\em Statistical Problems in Particle Physics, Astrophysics, and
  Cosmology}, 1:311.

\bibitem[{Giraud}, 2022]{Gir2022}
{Giraud}, C. (2022).
\newblock {\em {Introduction to high-dimensional statistics}}, volume 168.
\newblock Boca Raton, FL: CRC Press.

\bibitem[Goeman and Solari, 2011]{GS2011}
Goeman, J.~J. and Solari, A. (2011).
\newblock Multiple testing for exploratory research.
\newblock {\em Statist. Sci.}, 26(4):584--597.

\bibitem[Goodfellow et~al., 2016]{goodfellow2016}
Goodfellow, I., Bengio, Y., and Courville, A. (2016).
\newblock {\em Deep Learning}.
\newblock MIT Press.
\newblock \url{http://www.deeplearningbook.org}.

\bibitem[Guo et~al., 2020]{Guo_2020_CVPR}
Guo, T., Xu, C., Huang, J., Wang, Y., Shi, B., Xu, C., and Tao, D. (2020).
\newblock On positive-unlabeled classification in gan.
\newblock In {\em Proceedings of the IEEE/CVF Conference on Computer Vision and
  Pattern Recognition (CVPR)}.

\bibitem[Hastie et~al., 2009]{hastie2009elements}
Hastie, T., Tibshirani, R., Friedman, J.~H., and Friedman, J.~H. (2009).
\newblock {\em The elements of statistical learning: data mining, inference,
  and prediction}, volume~2.
\newblock Springer.

\bibitem[Ivanov, 2020]{ivanov2020dedpul}
Ivanov, D. (2020).
\newblock Dedpul: Difference-of-estimated-densities-based positive-unlabeled
  learning.
\newblock In {\em 2020 19th IEEE International Conference on Machine Learning
  and Applications (ICMLA)}, pages 782--790. IEEE.

\bibitem[Ivezi{\'c} et~al., 2019]{ivezic2019statistics}
Ivezi{\'c}, {\v{Z}}., Connolly, A.~J., VanderPlas, J.~T., and Gray, A. (2019).
\newblock {\em Statistics, data mining, and machine learning in astronomy: A
  practical python guide for the analysis of survey data}.
\newblock Princeton University Press.

\bibitem[Ivezi{\'c} et~al., 2005]{ivezic2005selection}
Ivezi{\'c}, {\v{Z}}., Vivas, A.~K., Lupton, R.~H., and Zinn, R. (2005).
\newblock The selection of {RR} lyrae stars using single-epoch data.
\newblock {\em The Astronomical Journal}, 129(2):1096.

\bibitem[Javanmard and Javadi, 2019]{javanmard2019false}
Javanmard, A. and Javadi, H. (2019).
\newblock False discovery rate control via debiased lasso.
\newblock {\em Electronic Journal of Statistics}, 13(1):1212--1253.

\bibitem[Korn et~al., 2004]{Korn2004}
Korn, E.~L., Troendle, J.~F., McShane, L.~M., and Simon, R. (2004).
\newblock Controlling the number of false discoveries: application to
  high-dimensional genomic data.
\newblock {\em J. Statist. Plann. Inference}, 124(2):379--398.

\bibitem[LeCun and Cortes, 2010]{MNIST}
LeCun, Y. and Cortes, C. (2010).
\newblock {MNIST} handwritten digit database.

\bibitem[Lei et~al., 2021]{lei2021distribution}
Lei, L., D’Amour, A., Ding, P., Feller, A., and Sekhon, J. (2021).
\newblock Distribution-free assessment of population overlap in observational
  studies.
\newblock Technical report.

\bibitem[Lei and Fithian, 2018]{lei2018adapt}
Lei, L. and Fithian, W. (2018).
\newblock Adapt: an interactive procedure for multiple testing with side
  information.
\newblock {\em Journal of the Royal Statistical Society: Series B (Statistical
  Methodology)}, 80(4):649--679.

\bibitem[Liang et~al., 2022]{liang2022integrative}
Liang, Z., Sesia, M., and Sun, W. (2022).
\newblock Integrative conformal p-values for powerful out-of-distribution
  testing with labeled outliers.

\bibitem[Loper et~al., 2019]{loper2019smoothed}
Loper, J.~H., Lei, L., Fithian, W., and Tansey, W. (2019).
\newblock Smoothed nested testing on directed acyclic graphs.
\newblock {\em arXiv preprint arXiv:1911.09200}.

\bibitem[Ma et~al., 2021]{ma2021global}
Ma, R., Tony~Cai, T., and Li, H. (2021).
\newblock Global and simultaneous hypothesis testing for high-dimensional
  logistic regression models.
\newblock {\em Journal of the American Statistical Association},
  116(534):984--998.

\bibitem[Mary and Roquain, 2022]{mary2021semisupervised}
Mary, D. and Roquain, E. (2022).
\newblock {Semi-supervised multiple testing}.
\newblock {\em Electronic Journal of Statistics}, 16(2):4926 -- 4981.

\bibitem[Ramdas et~al., 2019a]{ramdas2019sequential}
Ramdas, A., Chen, J., Wainwright, M.~J., and Jordan, M.~I. (2019a).
\newblock A sequential algorithm for false discovery rate control on directed
  acyclic graphs.
\newblock {\em Biometrika}, 106(1):69--86.

\bibitem[Ramdas et~al., 2019b]{ramdas2019unified}
Ramdas, A.~K., Barber, R.~F., Wainwright, M.~J., and Jordan, M.~I. (2019b).
\newblock A unified treatment of multiple testing with prior knowledge using
  the p-filter.
\newblock {\em The Annals of Statistics}, 47(5):2790--2821.

\bibitem[Rava et~al., 2021]{rava2021burden}
Rava, B., Sun, W., James, G.~M., and Tong, X. (2021).
\newblock A burden shared is a burden halved: A fairness-adjusted approach to
  classification.
\newblock {\em arXiv preprint arXiv:2110.05720}.

\bibitem[Romano and Wolf, 2005]{RW2005}
Romano, J.~P. and Wolf, M. (2005).
\newblock Exact and approximate stepdown methods for multiple hypothesis
  testing.
\newblock {\em J. Amer. Statist. Assoc.}, 100(469):94--108.

\bibitem[Roquain and Villers, 2011]{RV2011}
Roquain, E. and Villers, F. (2011).
\newblock Exact calculations for false discovery proportion with application to
  least favorable configurations.
\newblock {\em Ann. Statist.}, 39(1):584--612.

\bibitem[Rosset et~al., 2022]{rosset2022optimal}
Rosset, S., Heller, R., Painsky, A., and Aharoni, E. (2022).
\newblock Optimal and maximin procedures for multiple testing problems.
\newblock {\em Journal of the Royal Statistical Society Series B: Statistical
  Methodology}, 84(4):1105--1128.

\bibitem[Ruff et~al., 2018]{ruff18a}
Ruff, L., Vandermeulen, R.~A., G{\"o}rnitz, N., Deecke, L., Siddiqui, S.~A.,
  Binder, A., M{\"u}ller, E., and Kloft, M. (2018).
\newblock Deep one-class classification.
\newblock In {\em Proceedings of the 35th International Conference on Machine
  Learning}, volume~80, pages 4393--4402.

\bibitem[Sarkar, 2008]{Sar2008}
Sarkar, S.~K. (2008).
\newblock On methods controlling the false discovery rate.
\newblock {\em Sankhya, Ser. A}, 70:135--168.

\bibitem[Sch{\"o}lkopf et~al., 2001]{scholkopf2001estimating}
Sch{\"o}lkopf, B., Platt, J.~C., Shawe-Taylor, J., Smola, A.~J., and
  Williamson, R.~C. (2001).
\newblock Estimating the support of a high-dimensional distribution.
\newblock {\em Neural computation}, 13(7):1443--1471.

\bibitem[Scott and Nowak, 2005]{scott2005neyman}
Scott, C. and Nowak, R. (2005).
\newblock A neyman-pearson approach to statistical learning.
\newblock {\em IEEE Transactions on Information Theory}, 51(11):3806--3819.

\bibitem[Sen, 2018]{sen2018gentle}
Sen, B. (2018).
\newblock A gentle introduction to empirical process theory and applications.
\newblock {\em Lecture Notes, Columbia University}, 11:28--29.

\bibitem[Storey et~al., 2004]{STS2004}
Storey, J.~D., Taylor, J.~E., and Siegmund, D. (2004).
\newblock Strong control, conservative point estimation and simultaneous
  conservative consistency of false discovery rates: a unified approach.
\newblock {\em J. R. Stat. Soc. Ser. B Stat. Methodol.}, 66(1):187--205.

\bibitem[Sugiyama et~al., 2012]{sugiyama2012density}
Sugiyama, M., Suzuki, T., and Kanamori, T. (2012).
\newblock {\em Density ratio estimation in machine learning}.
\newblock Cambridge University Press.

\bibitem[{Sun} and {Cai}, 2007]{SC2007}
{Sun}, W. and {Cai}, T.~T. (2007).
\newblock {Oracle and adaptive compound decision rules for false discovery rate
  control}.
\newblock {\em {J. Am. Stat. Assoc.}}, 102(479):901--912.

\bibitem[Sun and Cai, 2009]{SC2009}
Sun, W. and Cai, T.~T. (2009).
\newblock Large-scale multiple testing under dependence.
\newblock {\em J. R. Stat. Soc. Ser. B Stat. Methodol.}, 71(2):393--424.

\bibitem[{Vapnik}, 1998]{vapnik1998statistical}
{Vapnik}, V.~N. (1998).
\newblock {\em {Statistical learning theory}}.
\newblock Chichester: Wiley.

\bibitem[Vovk, 2015]{vovk2015cross}
Vovk, V. (2015).
\newblock Cross-conformal predictors.
\newblock {\em Annals of Mathematics and Artificial Intelligence}, 74(1):9--28.

\bibitem[Vovk et~al., 2005]{vovk2005algorithmic}
Vovk, V., Gammerman, A., and Shafer, G. (2005).
\newblock {\em Algorithmic learning in a random world}.
\newblock Springer Science \& Business Media.

\bibitem[Wang et~al., 2022]{wang2022approximate}
Wang, Y., Kaji, T., and Rockova, V. (2022).
\newblock Approximate bayesian computation via classification.
\newblock {\em Journal of Machine Learning Research}, 23(350):1--49.

\bibitem[Weinstein, 2021]{weinstein2021permutation}
Weinstein, A. (2021).
\newblock On permutation invariant problems in large-scale inference.
\newblock {\em arXiv preprint arXiv:2110.06250}.

\bibitem[Weinstein et~al., 2017]{weinstein2017power}
Weinstein, A., Barber, R., and Cand{\`e}s, E. (2017).
\newblock A power and prediction analysis for knockoffs with lasso statistics.

\bibitem[Yang et~al., 2021]{yang2021bonus}
Yang, C.-Y., Lei, L., Ho, N., and Fithian, W. (2021).
\newblock Bonus: Multiple multivariate testing with a data-adaptive test
  statistic.

\bibitem[Zrnic et~al., 2021]{zrnic2021asynchronous}
Zrnic, T., Ramdas, A., and Jordan, M.~I. (2021).
\newblock Asynchronous online testing of multiple hypotheses.
\newblock {\em J. Mach. Learn. Res.}, 22:33--1.

\end{thebibliography}

\newpage
\appendix

\setcounter{page}{1}

\title{Supplementary material \\
Adaptive novelty detection with false discovery rate guarantee}
\begin{center}
by A. Marandon, L. Lei, D. Mary and E. Roquain
\end{center}

\section{Proofs and results for Section~3}

\subsection{Proof of Lemma~\ref{lem:ariane}}\label{sec:prooflem:ariane}

Let 
\begin{align*}
U=(U_1,\dots,U_{n+m_0})&=(Y_{1},\dots,Y_{n},X_i,i\in \cH_0);\\
V=(V_1,\dots,V_{m_1})&=(X_i,i\in \cH_1);\\
W=h(U,V)&=((Z_1,\dots,Z_k) , \{Z_{k+1},\dots,Z_{n+m}\});\\
S_i &= g(U_i,W), \:\: i\in \{1,\dots,n+m_0\},
\end{align*}
for given measurable function $g$ that satisfies the condition \eqref{constrainedg}. Then, for any permutation $\pi$ of $\{1,\dots,n+m_0\}$ that do not permute $\{1,\dots,k\}$, 
Assumption~\ref{as:exchangeable0} implies that $U\:|\: V\sim U^\pi\:|\: V$ and thus $(U,V)\sim (U^\pi,V)$. This entails 
 $(U,V,W)\sim (U^\pi,V,h(U^\pi,V))= (U^\pi,V,h(U,V)) =  (U^\pi,V,W)$. Hence,
\begin{align*}
(g(U_1,W),\dots,g(U_{n+m_0},W))\:|\: V,W\sim (g(U_{\pi(1)},W),\dots,g(U_{\pi(n+m_0)},W))\:|\: V,W
\end{align*}
Since $\pi(i)=i$ for all $i\in \{1,\dots, k\}$, we obtain that
\begin{align*}
&(g(U_{k+1},W),\dots,g(U_{n+m_0},W))\:|\: (g(V_1,W),\dots,g(V_{m_1},W))\\
& \sim (g(U_{\pi(k+1)},W),\dots,g(U_{\pi(n+m_0)},W))\:|\: (g(V_1,W),\dots,g(V_{m_1},W)),
\end{align*}
which completes the proof.

\subsection{Key properties}\label{sec:emp}

In this section, we present key properties of empirical $p$-values derived from exchangeable scores.
The first result provides a representation that characterizes the dependence structure of the empirical $p$-values that is a key step for the proof of FDR control. It generalizes the representation by \cite{bates2021testing} for independent scores. The proof is deferred to Section \ref{proof:th-key}.

\begin{theorem}\label{th-key}
Consider any family of scores $(S_{k+1},\dots,S_{n+m})$ that satisfy Assumptions~\ref{as:newexch}~and~\ref{as:noties} and the corresponding family of empirical $p$-values \eqref{emppvalues}.
For any $i\in \cH_0$, define
  \begin{equation}\label{equWi}
  W_i = \big( \{S_{k+1}, \ldots, S_n, S_{n+i}\}, (S_{n+j}, j \in \cH_0, j \neq i), (S_{n+j}, j \in \cH_1) \big)
  \end{equation}
  and, for any $j\in \{1,\dots,m\}\backslash\{i\}$,
  \begin{equation}\label{equCij}
  C_{i,j}=(\ell+1)^{-1}\sum_{s\in \{S_{k+1}, \ldots, S_{n}, S_{n+i}\}}\ind{s>S_{n+j}}.
  \end{equation}
Further, let $S_{(1)}> \dots > S_{(\ell+1)}$ be the order statistics of $\{S_{k+1}, \ldots, S_{n}, S_{n+i}\}$. Then the following holds:
\begin{itemize}
\item[(i)] 
  The sequences $(S_{n+j})_{j\in \{1,\dots,m\}\backslash\{i\}}$ and  $(C_{i,j})_{j\in \{1,\dots,m\}\backslash\{i\}}$ are both measurable with respect to $W_i$.
\item[(ii)] For all $j\in \{1,\dots,m\}\backslash\{i\}$, 
\begin{equation}\label{equpjfunctionpi}
p_j=C_{i,j} + \ind{S_{n+i}\leq S_{n+j}}/(\ell +1)=C_{i,j} + \ind{S_{((\l+1)p_i )}\leq S_{n+j}}/(\ell +1),
\end{equation}
which is a nondecreasing function of $p_i$ for any given $W_i$.
\item[(iii)] $p_i$ is independent of $W_i$.
\item[(iv)] $(\l+1)p_i$ is uniform distributed on $\{1,\dots,\l+1\}$.
\end{itemize}
\end{theorem}

The second result characterizes the distribution of other null $p$-values conditional on a given null $p$-value as well as the ordered scores. We defer the proof to Section \ref{proof:th-key2}.

\begin{theorem}\label{th-key2}
In the setting of Theorem~\ref{th-key}, the distribution of the family $(p_j, \:j\in \cH_0)$ satisfies 
\begin{align*}
&(p_j, \:j\in \cH_0\backslash\{i\}) \:|\:  p_i=1/(\ell +1), \{S_{(1)},\dots\, S_{(\ell+1)}\} \\
&\sim( C_{i,j} + \ind{S_{(1)}\leq S_{n+j}}/(\ell +1),\:j\in \cH_0\backslash\{i\}) \:|\:  \{S_{(1)},\dots\, S_{(\ell+1)}\}\\
&\sim (p'_j, \:j\in \cH_0\backslash\{i\}) \:|\:  \{U_{(1)},\dots\, U_{(\ell+1)}\} 
\end{align*}
for which $U_{(1)}>\dots> U_{(\ell+1)}$ denote the order statistics of $\ell +1$ i.i.d. $U(0,1)$ random variables $U_1,\dots,U_{\ell +1}$, and where $p'_j$, $j\in \cH_0\backslash\{i\}$, are conditionally on $\{U_{(1)},\dots\, U_{(\ell+1)}\}$, i.i.d. with common distribution $F(x)=(1-U_{(\lfloor x(\ell+1)\rfloor+1)} )\ind{1/(\ell+1)\leq x<1 }+\ind{x\geq 1}$, $x\in \R$.
\end{theorem}

Noting that $U_{(b)}$ follows a beta distribution, Theorem~\ref{th-key2}
can be used to compute the distribution of $(p_j, \:j\in \cH_0\backslash\{i\})$ conditionally on  $p_i=1/(\ell +1)$ by a simple integration. This is used in the proof of FDR control for $\pi_0$-adaptive AdaDetect (Theorem~\ref{thm:AdaptBONuS} and Corollary~\ref{StoreyBH_multidim}).

\subsection{Proof of Theorem~\ref{thm:PRDS_multidim}}\label{sec:proofthm:PRDS_multidim}

By Theorem~\ref{th-key} (iii), $p_i$ is independent of $W_i$. 
Hence, by integration, it is sufficient to establish that for any nondecreasing measurable set $D\subset [0,1]^{m-1}$, the function
\begin{equation}\label{equ:toprove}
r\in \{1,\dots,\ell+1\}\mapsto  \P((p_j, j\neq i)\in D\:|\: p_i=r/(\ell+1),W_i)
\end{equation}
is nondecreasing. 
By Theorem~\ref{th-key} (ii), we have that $(p_j, j\neq i)$ is a deterministic function of $p_i$ and $W_i$, which is nondecreasing in $p_i$. This gives \eqref{equ:toprove} and proves Theorem~\ref{thm:PRDS_multidim}.

\subsection{Proof of Theorem~\ref{thm:FDRBONuS}}\label{proof:thm:FDRBONuS}

This proof combines Theorem~\ref{th-key} with Lemma~\ref{BHsmallerp}, a property of the BH algorithm that slightly extends the classical result.

Recall that the AdaDetect procedure is the BH algorithm applied to the empirical $p$-values given by \eqref{emppvalues}. We apply Lemma~\ref{BHsmallerp} with the empirical $p$-values $(p_j,1\leq j\leq m)$ being the empirical $p$-values, any $i\in \cH_0$, $p'_i=1/(\ell +1)$, and $p'_j=C_{i,j}$ defined in \eqref{equCij} for $j\neq i$. By definition, the condition \eqref{propppprime} holds. Moreover, if $p_j > p_i$, we have $S_{n+i}>S_{n+j}$ in which case $\ind{S_{n+i}\leq S_{n+j}}=0$, implying that $p'_j=p_j$. Letting  $\wh{k}=\wh{k}(p_i,1\leq i\leq m)$ and $\wh{k}'_i=1 \vee \wh{k}(p'_i,1\leq i\leq m)$, Lemma~\ref{BHsmallerp} entails that
$$
\{p_i\leq \alpha \wh{k}/m \}=\{ p_i\leq \alpha \wh{k}'_i/m \}\subset \{ \wh{k}=\wh{k}'_i \}.
$$
Let $W_i$ be defined in \eqref{equWi}. Then
\begin{align*}
\FDR(\mbox{AdaDetect}_{\alpha}) &= \sum_{i\in \cH_0}  \E\left[\frac{\ind{p_i\leq \alpha \hat{k}/m}}{\hat{k}\vee 1}\right]\\
&= \sum_{i\in \cH_0}  \E\left[\frac{\ind{p_i\leq \alpha \hat{k}'_i/m}}{\hat{k}'_i}\right]\\
&=  \sum_{i\in \cH_0}  \E\left[\E\left[\frac{\ind{p_i\leq \alpha \hat{k}'_i/m}}{\hat{k}'_i}\:\big|\: W_i\right]\right]\\
&=  \sum_{i\in \cH_0}  \E\left[\frac{\P\left( (\l+1)p_i\leq \alpha (\l+1)\hat{k}'_i/m \:\big|\: W_i\right)}{\hat{k}'_i}\right],\end{align*}
where the last line is due to that $\hat{k}'_i$ is measurable with respect to $W_i$, which is implied by Theorem~\ref{th-key} (i). Then Theorem~\ref{th-key} (iii) and (iv) implies
\begin{align*}
\FDR(\mbox{AdaDetect}_{\alpha})=  \sum_{i\in \cH_0}  \E\left(\frac{ \lfloor \alpha (\l+1)\hat{k}'_i/m\rfloor }{(\ell +1) \hat{k}'_i}\right).
\end{align*}
The result is then proved by letting $K_i=\hat{k}'_i$.

\subsection{Proof of Theorem~\ref{thm:AdaptBONuS}}\label{proof:thm:AdaptBONuS}

Letting $\wh{k}=\wh{k}(p_i,1\leq i\leq m)$ the number of rejections of $\mbox{AdaDetect}_{\alpha m/G(p) }$, we have
\begin{align*}
  \FDR(\mbox{AdaDetect}_{\alpha m/G(p)}) &= \sum_{i\in \cH_0}  \E\left[\frac{\ind{p_i\leq \alpha \hat{k}(p)/G(p)}}{\hat{k}(p)\vee 1}\right]\\
  & = \sum_{i\in \cH_0}  \E\left[\frac{\ind{p_i\leq \alpha (\hat{k}(p)\vee 1)/G(p)}}{\hat{k}(p)\vee 1}\right]
\end{align*}
 By Theorem~\ref{th-key} (ii), we can write $(p_j,j\neq i)$ as $\Psi_1(W_i,p_i)$ for some deterministic function $\Psi_1$ that is nondecreasing in $p_i$ where $W_i$ is defined in \eqref{equWi}. As a result, we can write 
$\hat{k}(p)\vee 1$ (resp. $1/G(p)$) as $\Psi_2(W_i,p_i)$ (resp. $\Psi_3(W_i,p_i)$) for some deterministic functions $\Psi_2$, $\Psi_3$ that are nonincreasing in $p_i$ since $\hat{k}$ and $1/G$ are both coordinate-wise nonincreasing.
Let
\begin{align*}
    c^*(W_i)&=\max \mathcal{N}(W_i)\\
    \mathcal{N}(W_i)&=\{j/(\ell +1)\::\: 1\leq j\leq \ell +1, j/(\ell +1)\leq \alpha \Psi_2(W_i,j/(\ell +1))\Psi_3(W_i,j/(\ell +1))\}.
\end{align*}
Above, we define $c^*(W_i)=1/(\ell+1)$ if $\mathcal{N}(W_i) = \emptyset$.
\rev{By definition, $\mathcal{N}(W_i)$ is thus the set of all values that the empirical $p$-value $p_i$ can take if it is rejected and  $c^*(W_i)$ is the largest possible value.
}
Thus, $\{p_i\leq \alpha \hat{k}(p)/G(p)\}= \{p_i\in\mathcal{N}(W_i)\}= \{p_i\leq c^*(W_i), \mathcal{N}(W_i)\neq \emptyset\}$. This entails
\begin{align*}
\FDR(\mbox{AdaDetect}_{\alpha m/G(p)}) &\leq \sum_{i\in \cH_0}  \E\left[\E\left[\frac{\ind{p_i\leq c^*(W_i)}}{\hat{k}(p)\vee 1} \ind{\mathcal{N}(W_i)\neq \emptyset}\:|\: W_i\right]\right]\\
&\leq \sum_{i\in \cH_0}  \E\left[\E\left[\frac{\ind{p_i\leq c^*(W_i)}}{\Psi_2(W_i, c^*(W_i))}\ind{\mathcal{N}(W_i)\neq \emptyset}\:|\: W_i\right]\right]\\
& = \sum_{i\in \cH_0}  \E\left[\frac{\P\left[p_i\leq c^*(W_i)\:|\: W_i\right] }{\Psi_2(W_i, c^*(W_i))}\ind{\mathcal{N}(W_i)\neq \emptyset}\right]\\
&\leq \sum_{i\in \cH_0}  \E\left[\frac{ c^*(W_i)}{\Psi_2(W_i, c^*(W_i))}\ind{\mathcal{N}(W_i)\neq \emptyset}\right],
\end{align*}
where the last two lines use Theorem~\ref{th-key} (iii) and (iv), respectively.
By definition of $c^*(W_i)$, we obtain
\begin{align*}
\FDR(\mbox{AdaDetect}_{\alpha m/G(p)}) &\leq \alpha \sum_{i\in \cH_0}  \E\left[\Psi_3(W_i, c^*(W_i))\right]\\
                                       &\leq \alpha \sum_{i\in \cH_0}  \E\left[\Psi_3(W_i, 1/(\ell+1))\right] \\
  & \leq  \alpha \sum_{i\in \cH_0}  \E\left(\frac{1}{G(q'_j,1\leq j\leq m)} \right),
\end{align*}
where $q'_i=1/(\ell +1)$, $q'_j=0$ for $j\in \cH_1$ and $q'_j=C_{i,j} + \ind{S_{(1)}\leq S_{n+j}}/(\ell +1)$ for $j\in \cH_0\backslash\{i\}$. The proof is completed by applying Theorem~\ref{th-key2}.

\subsection{Proof of Corollary~\ref{StoreyBH_multidim}}\label{proof:StoreyBH_multidim}

By using \eqref{equboundsadaptive} in Theorem~\ref{thm:AdaptBONuS}, the result is established if we prove in each case
\begin{equation}\label{toproveadaptive}
\sum_{i\in \cH_0}  \E\left(\frac{1}{G(p'_j,1\leq j\leq m)} \right)\leq 1,
\end{equation}
with $p'=(p'_j,1\leq j\leq m)\sim \mathcal{D}_i$ defined in \eqref{def:Di}.

\paragraph{Proof for the Storey estimator}
In this case, 
$$G(p')=\frac{1+\sum_{j=1}^m \ind{p'_j\geq \lambda}}{1-\lambda} = \frac{1+\sum_{j\in \cH_0\backslash\{i\}} \ind{p'_j\geq \lambda}}{1-\lambda},$$
 with $\lambda= K/(\ell+1)$ for $K\in \{2,\dots,\ell\}$. 
 Recall that, conditionally on $\{U_{(1)},\dots,U_{(\ell+1)}\}$,  $(p'_j: j\in \cH_0\backslash\{i\})$ are i.i.d. with the c.d.f. $F(x)=(1-U_{(\lfloor x(\ell+1)\rfloor+1)} )\ind{1/(\ell+1)\leq x<1 }+\ind{x\geq 1}$, $x\in \R$. Therefore, we have for $j\in \cH_0\backslash\{i\}$,
 \begin{align*}
    \P(p'_j\geq \lambda \:|\: \{U_{(1)},\dots,U_{(\ell+1)}\}) &=\P(p'_j\geq K/(\ell +1) \:|\: \{U_{(1)},\dots,U_{(\ell+1)}\}) \\
    &= \P(p'_j>(K-1)/(\ell +1) \:|\:  \{U_{(1)},\dots,U_{(\ell+1)}\}) \\
    &= 1-(1- U_{(K)})=U_{(K)}.
\end{align*}
 Thus, it is enough to prove
\begin{equation}\label{relStorey}
 \E\left(\frac{1}{1+\mathcal{B}(m_0-1, U_{(K)})}   \right)\leq \frac{1}{m_0 (1-\lambda)},
\end{equation}
where $\mathcal{B}(m_0-1, U_{(K)})$ denotes a Binomial random variable with parameters $m_0-1$ and $U_{(K)}$.
By Lemma~1 in \cite{BKY2006}, 
\begin{align*}
\E\left(\frac{1}{1+\mathcal{B}(m_0-1,U_{(K)})} \:|\:  \{U_{(1)},\dots,U_{(\ell+1)}\} \right)&\leq \frac{1}{m_0 U_{(K)}}.
\end{align*}
Hence, the LHS of \eqref{relStorey} is bounded by $\E\left( \frac{1}{m_0 
U_{(K)}}\right).$ 
It is well-known that $U_{(K)}\sim \beta(\ell+2-K,K)$  (note that $U_{(1)},\dots,U_{(\ell+1)}$ is a decreasing sequence) and the expectation of the inverse of a Beta random variable with scale parameters $a$ and $b$ is $(a+b-1)/(a-1)$. Hence, 
\begin{align*}
\E\left( \frac{1}{m_0 
U_{(K)}}\right)= \frac{\ell +1}{m_0 (\ell+1 -K)} = \frac{1}{m_0 (1-\lambda)}. 
\end{align*}

\paragraph{Proof for the Quantile estimator}

In that case, $$G(p')=\frac{m-k_0+1 }{1-p'_{(k_0)}},$$
where $p'_{(k_0)}$ denotes the $k_0$-smallest element of $(p'_j,1\leq j\leq m)$.
 If 
$k_0\leq m-m_0 +1$, we have $G(p') \geq m_0$ and  \eqref{toproveadaptive} holds.
Thus, we assume $k_0\geq m-m_0 +2$, in which case $k_0+m_0-m-1\geq 1$. Let $j_0=k_0+m_0-m-1 \in [1,m_0]$. For the rest of the proof, we will fix $i$ and write $\E[\cdot]$ for $\E_{p'\sim \mathcal{D}_i}[\cdot]$ in short. Then
$$
\E\left(\frac{1}{G(p'_j,1\leq j\leq m)} \right) = \E\left(\frac{1-p'_{(j_0:\cH_0\backslash\{i\})}}{m-k_0+1} \right),
$$
 where $p'_{(j_0:\cH_0\backslash\{i\})}$ denotes the $j_0$-smallest element of $(p'_j,j\in \cH_0\backslash\{i\})$. To prove \eqref{toproveadaptive}, it is left to prove
\begin{equation}\label{toprove}
 \E\left(1-p'_{(j_0:\cH_0\backslash\{i\})} \right)\leq \frac{m-k_0+1}{m_0}\Longleftrightarrow \E(p'_{(j_0:\cH_0\backslash\{i\})}) \geq j_0/m_0.
\end{equation}
By definition of $\mathcal{D}_i$,
\begin{align*}
\E(p'_{(j_0:\cH_0\backslash\{i\})} \:|\:  \{U_{(1)}, \dots, U_{(\ell+1)}\}) &= \int_0^{\infty} \P(p'_{(j_0:\cH_0\backslash\{i\})}> x\:|\: \{U_{(1)}, \dots, U_{(\ell+1)}\}) dx\\
&= \int_0^{\infty} \P\left(\sum_{j\in \cH_0\backslash\{i\}} \ind{p'_{j}\leq x} < j_0\:|\:  \{U_{(1)}, \dots, U_{(\ell+1)}\}\right) dx\\
&= \int_0^{\infty} \P\left(\mathcal{B}(m_0-1, F(x))  < j_0\right) dx,
\end{align*}
where $\mathcal{B}(m_0-1,F(x))$ denotes a Binomial random variable with parameters $m_0-1$ and $F(x)$, where $F(x)=(1-U_{(\lfloor x(\ell+1)\rfloor+1)} )\ind{1/(\ell+1)\leq x<1 }+\ind{x\geq 1}$, $x\in \R$. Hence, the last display is equal to
\begin{align*}
& (\ell+1)^{-1}+\sum_{b=2}^{\ell +1} \int_{0}^{1} \ind{\lfloor x(\ell+1)\rfloor+1 = b} \P\left(\mathcal{B}(m_0-1, 1-U_{(b)})  < j_0\right) dx\\
&=(\ell+1)^{-1}+(\ell+1)^{-1} \sum_{b=2}^{\ell +1}  \P\left(\mathcal{B}(m_0-1, 1-U_{(b)})  < j_0\right) \\
&\geq (\ell+1)^{-1} \sum_{b=1}^{\ell +1}  \P\left(\mathcal{B}(m_0-1, 1-U_{(b)})  < j_0\right).
\end{align*}
Hence,
\begin{align*}
\E(p'_{(j_0:\cH_0\backslash\{i\})}) &\ge \sum_{k=0}^{j_0-1} {m_0-1 \choose  k} (\ell+1)^{-1} \sum_{b=1}^{\ell+1} \E [(1-U_{(b)})^k U_{(b)} ^{m_0-1-k} ] .
\end{align*}
Since $U_{(b)}\sim \beta(\ell+2-b,b)$  (recall that $U_{(1)},\dots,U_{(\ell+1)}$ being a decreasing sequence), we have $1-U_{(b)}\sim \beta(b,\ell+2-b)$. Hence,
\begin{align*}
\E [(1-U_{(b)})^k U_{(b)}^{m_0-1-k}] &=\frac{(\ell+1)!}{(b-1)! (\ell+1-b)!}\int_\R x^k (1-x)^{m_0-1-k} x^{b-1} (1-x)^{\ell+2-b-1}dx\\
&=\frac{(\ell+1)!}{(b-1)! (\ell+1-b)!}\frac{(k+b-1)!(m_0+\ell-(k+b))!}{(m_0+\ell)!},
\end{align*}
where the last line uses the fact that the integrand is proportional to the density of $\beta(k+b,m_0-k+\ell+1-b)$. As a result, we get
\begin{align*}
& {m_0-1 \choose  k} (\ell+1)^{-1}\sum_{b=1}^{\ell+1}\E [(1-U_{(b)})^k U_{(b)}^{m_0-1-k}] \\
&=m_0^{-1}{m_0+\ell \choose m_0}^{-1}  \:\:\sum_{b=1}^{\ell+1} {k +b-1\choose k}{m_0-k-1+\ell -b+1 \choose m_0-k-1} \\
&=m_0^{-1}{m_0+\ell \choose m_0}^{-1}  \:\:\sum_{b'=0}^{\ell} {k +\ell-b'\choose k}{m_0-k-1+b' \choose m_0-k-1} .
\end{align*}
By applying Lemma~\ref{lemvander} with $j=m_0-k-1$, $u=k$, $v=\ell$ (hence $j+u=m_0-1$ and $j+u+v=\ell+m_0-1$), the RHS is equal to $1/m_0$. This proves \eqref{toprove} and hence the theorem..

\subsection{Proof of Theorem~\ref{th-key}}\label{proof:th-key}

By Assumption \ref{as:newnoties}, we can assume $(S_{k+1}, \ldots, S_{n+m})$ has no ties throughout the proof. 
The result (i) is obvious.   For (ii), note that $(\ell+1)p_i=1+\sum_{s\in \{S_{k+1}, \ldots, S_{n}\}}\ind{s>S_{n+i}} $ is the rank of $S_{n+i}$ within the set $\{S_{k+1}, \ldots, S_n,S_{n+i}\}$. Hence, $S_{((\ell+1)p_i)}=S_{n+i}$ and, for all $j\in \{1,\dots,m\}\backslash\{i\}$, 
   \begin{align*}
p_j&=(\ell+1)^{-1}\left(1+\sum_{\substack{s\in \{S_{k+1}, \ldots, S_{n},S_{n+i}\}\\ s\neq S_{((\ell+1)p_i)}}}\ind{s>S_{n+j}}\right)\\
&=(\ell+1)^{-1} \left(\sum_{s\in \{S_{k+1}, \ldots,S_{n},S_{n+i}\}}\ind{s>S_{n+j}} +\ind{S_{((\ell+1)p_i)}\leq S_{n+j}} \right).
\end{align*}
  This proves (ii). By Assumption~\ref{as:newexch}, for any permutation $\sigma$ of $\{k+1,\dots,n, n+i\}$, we have 
$$
((S_{k+1},\dots,S_{n},S_{n+i}), W_i)\sim ((S_{\sigma(k+1)},\dots,S_{\sigma(n)},S_{\sigma(n+i)}), W_i^\sigma) =((S_{\sigma(k+1)},\dots,S_{\sigma(n)},S_{\sigma(n+i)}), W_i),  
$$
where 
$$
W_i^\sigma = \big( \{S_{\sigma(k+1)}, \ldots, S_{\sigma(n)}, S_{\sigma(n+i)}\}, (S_{n+j}, j \in \cH_0, j \neq i), (S_{n+j}, j \in \cH_1) \big)
   = W_i.
$$
This implies that $(S_{k+1}, \ldots, S_{n}, S_{n+i})$ is exchangeable conditionally on $W_i$.
Now let $R_{1}, \ldots, R_{\ell+1}$ be the ranks of $(S_{k+1}, \ldots, S_{n}, S_{n+i})$ within the same set. Then $S_{j+k} = S_{(R_{j})}$ ($j = 1,\ldots, \ell$) and $S_{n+i}=S_{(R_{\ell+1})}$, where $S_{(1)} >S_{(2)} > \cdots > S_{(\ell+1)}$ are the order statistics.
Since $(S_{k+1}, \ldots, S_{n}, S_{n+i})$ are exchangeable conditionally on $W_i$ and almost surely mutually distinct, we have that
  \[(R_{1}, \ldots, R_{\ell+1})\indep W_i, \quad \text{and }\quad (R_{1}, \ldots, R_{\ell+1})\sim \text{Unif}(\mathfrak{S}(\{1, \ldots, n-k+1\})),\]
  where $\mathfrak{S}(\{1, \ldots, \ell+1\})$ denotes the set of permutations of $\{1, \ldots, \ell+1\}$.
 The results (iii) and (iv) then follow from the fact that $p_i = R_{\ell+1} / (\ell+1)$.
 
 \subsection{Proof of Theorem~\ref{th-key2}}\label{proof:th-key2}

By \eqref{equpjfunctionpi}, $p_j$ is a function of $p_i$ and $W_i$ for all $j\in \cH_0\backslash\{i\}$. Replacing $p_i$ by $1/(\ell +1)$ in that expression, we get
\begin{align*}
p_j&=(\ell+1)^{-1}\sum_{s\in \{S_{k+1}, \ldots, S_{n}, S_{n+i}\}}\ind{s>S_{n+j}} + \ind{S_{(1)}\leq S_{n+j}}/(\ell +1)\\
&=(\ell+1)^{-1}\sum_{q=1}^{\l+1}\ind{S_{(q)}>S_{n+j}} + \ind{S_{(1)}\leq S_{n+j}}/(\ell +1)\\
&=(\ell+1)^{-1}\left(1+\sum_{q=2}^{\l+1}\ind{S_{(q)}>S_{n+j}} \right).
\end{align*}
By Lemma~\ref{lem:reduction}, we have
\begin{align*}
&(p_j,\: j\in \cH_0\backslash\{i\}) \:|\:  p_i=1/(\ell +1), \{S_{(1)},\dots\, S_{(\ell+1)}\}\\
&\sim (p'_j,\: j\in \cH_0\backslash\{i\}) \:|\:  p'_i=1/(\ell +1), \{U_{(1)},\dots\, U_{(\ell+1)}\}\\
&\sim (p'_j,\: j\in \cH_0\backslash\{i\}) \:|\:   \{U_{(1)},\dots\, U_{(\ell+1)}\}
\end{align*}
where $p'_j=\frac{1+\sum_{q=2}^{\ell+1} \ind{U_{(q)}> V_j}}{\ell+1}$,  $U_1,\dots,U_{\ell +1},V_j (j\in \cH_0)$ are i.i.d. from $U(0,1)$, and $U_{(1)}>\dots> U_{(\ell+1)}$ denote the order statistics of $U_1,\dots,U_{\ell +1}$. As a result, conditional on $ \{U_{(1)},\dots\, U_{(\ell+1)}\}$, $(p'_j,\: j\in \cH_0\backslash\{i\})$ are i.i.d. with a c.d.f.
 \begin{align*}
 F(x)&=\P(p'_j\leq x\:|\:  U_{(1)}, \dots, U_{(\ell+1)} )\\
 &=\P\left(\sum_{q=2}^{\ell+1} \ind{U_{(q)}>V_j}\leq \lfloor x(\ell+1)\rfloor-1\:|\:  U_{(1)}, \dots, U_{(\ell+1)} \right)\\
 &=\P( U_{(\lfloor x(\ell+1)\rfloor+1)}\leq V_j)=1-U_{(\lfloor x(\ell+1)\rfloor+1)}.
 \end{align*}
because $\sum_{q=2}^{\ell+1} \ind{U_{(q)}>v}\geq \lfloor x(\ell+1)\rfloor$ is equivalent to $U_{(\lfloor x(\ell+1)\rfloor+1)}>v$. This completes the proof.

\section{Proofs for Section~4}

\subsection{Proof of Theorem~\ref{th:SCextended}}\label{sec:th:SCextended}

Let $T(x)=1-\lrt(x)=\pi_0 f_0(x)/f(x)$ and $t(\alpha)=1-c(\alpha)\in (0,1)$. Then $R=\{i \::\: T(X_i)\leq t(\alpha)\}$. Consider any procedure $R'=\{i \::\: T'(X_i)\leq t'\}$ with $\mbox{mFDR}(R')\leq \alpha$.
Since $\mbox{mFDR}(R)=\alpha$, we have both
\begin{align*}
0&=\int \ind{T(x)\leq t(\alpha)} (T(x)-\alpha) f(x) d\nu(x)\\
0&\geq \int \ind{T'(x)\leq t'} (T(x)-\alpha) f(x) d\nu(x).
\end{align*}
The first equality implies $t(\alpha)\geq \alpha$. If $t(\alpha) = \alpha$, then $T(X) = \alpha$ almost surely under $f$. This implies that all hypotheses are rejected with probability $1$ and hence $R$ is never less powerful than $R'$. 

Assume $t(\alpha) > \alpha$. Then the two equalities imply
 \begin{align}\label{equintermopt}
\int (\ind{T(x)\leq t(\alpha)}- \ind{T'(x)\leq t'}) (T(x)-\alpha) f(x) d\nu(x) \geq  0.
\end{align}
Since $T(x)\leq t(\alpha)$ is equivalent to $\frac{T(x)-\alpha}{1-T(x)}\leq \frac{t(\alpha)-\alpha}{1-t(\alpha)}$ (even when $T(x)=1$), we obtain
 \begin{align*}
\frac{t(\alpha)-\alpha}{1-t(\alpha)}\int (\ind{T(x)\leq t(\alpha)}- \ind{T'(x)\leq t'}) (1-T(x)) f(x) d\nu(x) \geq  0.
 \end{align*}
 Since $t(\alpha) > \alpha$,
 \begin{align}\label{equintermopt2}
\int (\ind{T(x)\leq t(\alpha)}- \ind{T'(x)\leq t'})  \bar{f}_1(x) d\nu(x) \geq  0.
\end{align}

\subsection{Proof of Lemma~\ref{lem:PUclassiforacle}}\label{proofPUclassiforacle}


For case (i), \eqref{equJlambda2} can be expressed as
\begin{align*}
2J_\lambda(g)&= \int \Big\{k (1+g(x))_+ +  \lambda (\ell+m) (1-g(x))_+  f_\gamma(x)/f_0(x) \Big\} f_0(x)d\nu(x).
\end{align*}
For any $u, v > 0$ and $a \in \R$,
$$
u (1+a)_+ +v (1-a)_+ = v (1-a) \ind{a< -1} + ( u+ v + a (u-v))\ind{-1\leq a\leq 1} + u (1+a) \ind{a> 1}.
$$
As a function of $a$, it is continuous and piecewise linear with two turning points $(-1,2v)$ and $(1,2u)$. When $u\neq v$, the unique minimum is attained at $a=\mbox{sign}(v/u-1)$. The proof is then completed by setting $u=k$ and $v=\lambda (\ell+m)  f_\gamma(x)/f_0(x)$ and the assumption that $\P(f_{\gamma}(X)/f_0(X) = k/\lambda(\ell+m)) = 0$.
 

For case (ii), \eqref{equJlambda2} is given by
\begin{align*}
J_\lambda(g)&= \int \Big\{-k\log(1-g(x)) - \lambda (\ell+m) \log(g(x))  f_\gamma(x)/f_0(x) \Big\} f_0(x)d\nu(x).
\end{align*}
For any $u>0$, $v\geq 0$, the map $a\in [0,1]\mapsto u \log(1-a) + v \log(a)$ has a unique maximizer at $a=v/(u+v)$.

\section{Proofs for Section~5}

\subsection{Proof of Theorem~\ref{powerhighdim}}\label{sec:proofpowerhighdim}

For all $g\in \cG$, let 
\begin{align}
\tilde{R}_0(g)&= \ell^{-1}  \sum_{i=k+1}^{n}  \ind{g(Z_i)\geq 0} ;\\
\hat{R}_{\gamma,0}(g)&= (m_0+\ell)^{-1} \left( \sum_{i=k+1}^{n}  \ind{g(Z_i)< 0} + \sum_{i\in \cH_0}  \ind{g(Z_{n+i})< 0}\right) ;\\
\hat{R}_{\gamma,1}(g)&= m_1^{-1}  \sum_{i\in \cH_1}  \ind{g(Z_{n+i})< 0},
\end{align} 
so that $\hat{R}_{\gamma}(g)=(m+\ell)^{-1} \sum_{i=k+1}^{n+m} \ind{g(Z_i)<  0}=(1-\gamma)\hat{R}_{\gamma,0}(g)+\gamma \hat{R}_{\gamma,1}(g)$. For notational convenience, for any $i\ge 1$, let
\[e_{i} = \sqrt{\frac{V(\cG) + \log(1/\delta)}{i}}.\]
Define the following events:
\begin{align*}
\Omega_0 &= \left\{ \sup_{g\in \cG}|\hat{R}_0(g)-R_0(g)|\leq b e_{k}\right\};\\
\tilde{\Omega}_0 &= \left\{ \sup_{g\in \cG}|\tilde{R}_0(g)-R_0(g)|\leq b e_{\ell}\right\};\\
\Omega_{\gamma,0} &= \left\{ \sup_{g\in \cG}|\hat{R}_{\gamma,0}(g)-(1-R_0(g))|\leq b e_{m_0+\ell}\right\};\\
\Omega_{\gamma,1} &= \left\{ \sup_{g\in \cG}|\hat{R}_{\gamma,1}(g)-R_1(g)|\leq b e_{m_1}\right\}.
\end{align*}
We choose $b$ such that
\begin{equation}
  \label{eq:Pevent}
  \P(\Omega_0\cap \tilde{\Omega}_0\cap \Omega_{\gamma,0}\cap  \Omega_{\gamma,1})\geq 1-\delta.
\end{equation}
The well-known result for empirical processes on finite VC classes (e.g., Example 7.10 of \cite{sen2018gentle}) implies that $b$ only depends on $\delta$. Throughout the rest of the proof, we choose
\begin{equation}
  \label{eq:CC'}
  C = b, \quad C' = 30b,
\end{equation}
where $C$ and $C'$ are the constants in expressions of $\epsilon_0$ and $\Delta$, respectively. 

Note that on $\Omega_{\gamma,0}\cap \Omega_{\gamma,1}$, we have
\begin{align}
  &\sup_{g\in \cG}|\hat{R}_\gamma(g)-R_\gamma(g)|\leq  b (1-\gamma) e_{m_0+\ell} + b \gamma  e_{m_1}\nonumber\\
  & = b \left(\sqrt{\frac{m_0+\ell}{m+\ell}} + \sqrt{\frac{m_1}{m+\ell}}\right) e_{m+\ell} \le 2be_{m+\ell}.\label{equintremRgamma}
\end{align}
On $\Omega_0$, \eqref{defghat} implies that
$${R}_0(\hat{g}) =  \hat{R}_0(\hat{g}) + {R}_0(\hat{g}) -\hat{R}_0(\hat{g}) \leq \beta + \epsilon_0 + b e_{k} = \beta + 2be_{k}.$$
Clearly, $2be_{k} \le 30\gamma^{-1}b e_{k}\le \Delta$. This proves the first claim of (i). Moreover, on $\Omega_0\cap \tilde{\Omega}_0$, we have
$$
\tilde{R}_0(\hat{g}) = {R}_0(\hat{g}) + \tilde{R}_0(\hat{g})-{R}_0(\hat{g}) \leq \beta  + 3b (e_{k} \vee e_{\ell}) \le \beta + 0.1 \gamma\Delta.
$$
Equivalently,
\begin{equation}\label{equintermR0}
 \sum_{i=k+1}^{n}  \ind{\hat{g}(Z_i)\geq 0} \leq \ell (\beta  + 0.1\gamma\Delta).
\end{equation}
By the assumption that ${R}_0(g^\sharp_\cG) = \beta$, on the event $\Omega_0$, $\hat{R}_0(g^\sharp_\cG)\leq {R}_0(g^\sharp_\cG)+b e_{k} = \beta + \epsilon_0$. By definition \eqref{defghat} of $\hat{g}$, we have\begin{equation}\label{majkey}
\hat{R}_\gamma(\hat{g})\leq \hat{R}_\gamma(g^\sharp_\cG).
\end{equation}
By \eqref{equintremRgamma} and \eqref{majkey}, on $\Omega_0\cap \Omega_{\gamma,0}\cap \Omega_{\gamma,1}$, 
$${R}_\gamma(\hat{g}) = \hat{R}_\gamma(\hat{g})  + {R}_\gamma(\hat{g}) -\hat{R}_\gamma(\hat{g}) \leq \hat{R}_\gamma(g^\sharp_\cG) + 2b e_{m+\ell}\leq {R}_\gamma(g^\sharp_\cG) + 4b e_{m+\ell}.$$
Together with \eqref{NPexactRgamma}, this implies
\[(1 - \gamma)(1 - R_0(\hat{g})) + \gamma R_1(\hat{g}) \le (1 - \gamma)(1 - R_0(g^\sharp_\cG)) + \gamma R_1(g^\sharp_\cG) + 4b e_{m+\ell}\]
and thus \[ R_{1}(\hat{g})\le R_1(g^\sharp_\cG) + \gamma^{-1}(R_0(\hat{g}) - R_0(g^\sharp_\cG) +4b e_{m+\ell}).\]
By definition, $R_0(\hat{g})\le \beta + \epsilon_0 = {R}_0(g^\sharp_\cG) + \epsilon_0 = {R}_0(g^\sharp_\cG) + be_{k}$. Thus,
\begin{align*}
  {R}_1(\hat{g}) \leq {R}_1(g^\sharp_\cG) + \gamma^{-1}b [4 e_{m+\ell} + e_{k}]\leq {R}_1(g^\sharp_\cG) + 5\gamma^{-1}b e_{k}.
\end{align*}
Clearly, $5\gamma^{-1}be_{k} \le 30\gamma^{-1}b e_{k}\le \Delta$. This proves the second claim of (i). Then, on the event $\Omega_{\gamma, 1}$, we have
\begin{align*}
\hat{R}_{\gamma,1}(\hat{g}) &\leq R_{1}(\hat{g}) + be_{m_1} \le {R}_1(g^\sharp_\cG) + be_{m_1} + 5\gamma^{-1}b e_{k}.
\end{align*}
Since $e_{m_1} = \gamma^{-1/2}e_{m+\ell} \le \gamma^{-1}e_{m+\ell}\le \gamma^{-1}e_{k}$, we have
\[\hat{R}_{\gamma,1}(\hat{g}) \leq {R}_1(g^\sharp_\cG) + 6\gamma^{-1} be_{k}.\]
Since $C' \ge 6b$,
$$\Delta=C' \gamma^{-1} (e_{k}\vee e_{\ell}) \ge 6\gamma^{-1}b e_{k}.$$
Thus, on $\Omega_0\cap \Omega_{\gamma,0}\cap \Omega_{\gamma,1}$,
\[\frac{1}{m_1}\sum_{i\in \cH_1}  \ind{\hat{g}(Z_{n+i}) < 0} = \hat{R}_{\gamma,1}(\hat{g}) \leq {R}_1(g^\sharp_\cG) + \Delta.\]
Equivalently, recalling that $X_i = Z_{n+i}$,
\begin{equation}\label{equintermR1}
  \sum_{i\in \cH_1}  \ind{\hat{g}(X_{i})\geq  0}\ge M.
\end{equation}

Let $\eta = 1/\ell + \beta +0.1\gamma\Delta$. Then \eqref{equintermR0} implies that
\[ 1 + \sum_{i=k+1}^{n}  \ind{\hat{g}(Z_i)\geq 0} \leq \eta\ell.\]
To validate the conditions in Lemma \ref{positivescoreBH}, we only need to show $\eta \le \alpha M / m$. Since $1 - R_{1}(g^\sharp_\cG)\ge (1 + \alpha^{-1})\Delta$, we have
\[\Delta \le \alpha(1 - R_{1}(g^\sharp_\cG) - \Delta) \le \frac{\alpha M}{m_1}\le \frac{\alpha M}{\gamma m}.\]
By the assumptions that $\ell \ge 2m/(\alpha M)$ and $\beta \le 0.4\alpha M / m$, we have
\[\eta \le \frac{\alpha M}{2m} + \frac{0.4\alpha M}{m} + \frac{0.1\alpha M}{m} = \frac{\alpha M}{m}.\]
By Lemma~\ref{positivescoreBH}, on the event  $\Omega_0\cap \tilde{\Omega}_0\cap \Omega_{\gamma,0}\cap  \Omega_{\gamma,1}$,
\[\mbox{AdaDetect}_{\alpha}\supset \{i \in \{1, \ldots, m\}: \hat{g}(X_i)\ge 0\}.\]
By \eqref{equintermR1},
\[|\mbox{AdaDetect}_{\alpha}\cap \cH_1| / m_1 \ge M / m_1 \ge 1-{R}_1(g^\sharp_\cG)-\Delta.\]
By \eqref{eq:Pevent}, this occurs with probability at least $1 - \delta$.

\subsection{Proof of Theorem~\ref{th:BHestimated} }\label{proof:BHestimated}

We first prove the result for 
 $\BH^*_\alpha$, i.e.,
\begin{align}\label{powerBHstar}
\P\left(\mathcal{R}' \cap \{\BH^*_\alpha\subset \mbox{AdaDetect}_{\alpha(1+\delta)(1+\zeta_r(\eta))}\}^c\right)\leq \P\left(\hat{\eta}> \eta \right)-m e^{-(3/28)  (\ell+1)\delta^2\alpha(r\vee 1)/m} ,
\end{align}
where $\mathcal{R}'=\{|\BH^*_\alpha|\geq r\}$.
First we note that, while $\hat{g}(Y_i)$ are dependent through the score function $\hat{g}(\cdot)$, the $g^*(Y_i)$ are i.i.d., allowing us to apply concentration inequalities.
For $s\in \R$, define
\begin{align*}
\wh{G}_0(s) &= (\ell+1)^{-1}\left(1+\sum_{i=k+1}^n \ind{\hat{g}(Y_i)\geq s}\right);\\
\wh{G}^*_0(s) &= (\ell+1)^{-1}\left(1+\sum_{i=k+1}^n \ind{g^*(Y_i)\geq s}\right).
\end{align*}
For notational convenience, let $\tilde{\alpha} = \alpha(r\vee 1) / m$. Consider in addition the following events:
\begin{align*}
\Omega_1 &=\left\{\max_{k+1\leq i\leq n+m} |\hat{g}(Z_i)- g^*(Z_i) |\leq \eta \right\};\\
\Omega_2 &= \left\{ \sup_{1\leq i\leq m}\left(\frac{\wh{G}^*_0(g^*(X_i)-2\eta)-\ol{G}_0(g^*(X_i)-2\eta)\vee \tilde{\alpha}}{\ol{G}_0(g^*(X_i)-2\eta)\vee \tilde{\alpha}}\right)\leq \delta \right\}.
\end{align*}
On $\Omega_1\cap \Omega_2$, we have for all $i\in \{1,\dots,m\}$, 
\begin{align}
&\wh{G}_0(\hat{g}(X_i)) \leq \wh{G}^*_0(\hat{g}(X_i)-\eta)
\leq \wh{G}^*_0(g^*(X_i)-2\eta)
\leq (\ol{G}_0(g^*(X_i)-2\eta)\vee \tilde{\alpha} )(1+\delta)\nonumber\\
&= (\ol{G}_0(g^*(X_i))\vee \tilde{\alpha} )\left(1 + \frac{\ol{G}_0(g^*(X_i)-2\eta)\vee \tilde{\alpha} - \ol{G}_0(g^*(X_i))\vee \tilde{\alpha}}{\ol{G}_0(g^*(X_i))\vee \tilde{\alpha}}\right)(1 +\delta) .\nonumber
\end{align}
Let $u = \ol{G}_0(g^*(X_i))\vee \tilde{\alpha}$. Since $\ol{G}_0$ is nonincreasing, $u = \ol{G}_0(g^*(X_i)\wedge \ol{G}_0^{-1}(\tilde{\alpha}))$. Thus, 
\begin{align*}
  & \frac{\ol{G}_0(g^*(X_i)-2\eta)\vee \tilde{\alpha} - \ol{G}_0(g^*(X_i))\vee \tilde{\alpha}}{\ol{G}_0(g^*(X_i))\vee \tilde{\alpha}}\\
  & =\frac{\ol{G}_0((g^*(X_i)-2\eta)\wedge \ol{G}_0^{-1}(\tilde{\alpha})) - u}{u}\\
  & \le \frac{\ol{G}_0(\ol{G}_0^{-1}(u)-2\eta) - u}{u}.
\end{align*}
Since $u\ge \tilde{\alpha}$, the LHS is bounded by $\zeta_r(\eta)$. As a result, on $\Omega_1\cap \Omega_2$,
\begin{equation}\label{equintermG0bar}
 \wh{G}_0(\hat{g}(X_i))\le (\ol{G}_0(g^*(X_i))\vee \tilde{\alpha} )(1 + \zeta_r(\eta))(1 + \delta).
\end{equation}
Then, for all $t\in \{\alpha k/m,r\vee 1\leq k\leq m\}$,
\begin{align*}
\ind{\ol{G}_0(g^*(X_i)) \leq t}=\ind{\ol{G}_0(g^*(X_i))\vee \tilde{\alpha} \leq t}
&\leq 
\ind{
\wh{G}_0(\hat{g}(X_i))\leq t (1+\delta)(1+\zeta_r(\eta))}.
\end{align*}
By applying Lemma~\ref{BHdomination} with $p_i=\ol{G}_0(g^*(X_i))$, $p_i'= \wh{G}_0(\hat{g}(X_i))$, $\beta=\alpha$, and
$\beta'=\alpha(1+\zeta_r(\eta))(1+\delta)$, we obtain that $\BH^*_\alpha\subset \mbox{AdaDetect}_{\alpha(1+\delta)(1+\zeta_r(\eta))}$ on $\Omega_1\cap \Omega_2 \cap \mathcal{R}'$. We are left to show that
\[\P(\Omega_2^c)\le m\exp(-(3/28)  (\ell+1)\delta^2\tilde{\alpha}).\]
Since $X_i$'s and $Y_i$'s are independent, the union bound implies
\begin{align*}
\P(\Omega_2^c )&\leq  \P\left( \sup_{1\leq i\leq m}\left(\frac{\wh{G}^*_0(g^*(X_i)-2\eta)-\ol{G}_0(g^*(X_i)-2\eta)\vee \tilde{\alpha}}{\ol{G}_0(g^*(X_i)-2\eta)\vee \tilde{\alpha}}\right)\geq \delta \right)\\
               &\leq \sum_{i=1}^m \P\left(\frac{\wh{G}^*_0(g^*(X_i)-2\eta)-\ol{G}_0(g^*(X_i)-2\eta)\vee \tilde{\alpha}}{\ol{G}_0(g^*(X_i)-2\eta)\vee \tilde{\alpha}}\geq \delta \right)\\
               &\leq \sum_{i=1}^m \P\left(\frac{\wh{G}^*_0(g^*(X_i)-2\eta)-\ol{G}_0((g^*(X_i)-2\eta)\wedge \ol{G}_0^{-1}(\tilde{\alpha}))}{\ol{G}_0((g^*(X_i)-2\eta)\wedge \ol{G}_0^{-1}(\tilde{\alpha}))}\geq \delta \right)\\
               &\leq \sum_{i=1}^m \P\left(\frac{\wh{G}^*_0((g^*(X_i)-2\eta)\wedge \ol{G}_0^{-1}(\tilde{\alpha}))-\ol{G}_0((g^*(X_i)-2\eta)\wedge \ol{G}_0^{-1}(\tilde{\alpha}))}{\ol{G}_0((g^*(X_i)-2\eta)\wedge \ol{G}_0^{-1}(\tilde{\alpha}))}\geq \delta \right)\\
&\leq m \sup_{s\leq \ol{G}_0^{-1}(\tilde{\alpha})} \P\left(\frac{\wh{G}^*_0(s)-\ol{G}_0(s)}{\ol{G}_0(s)}\geq \delta \right).
\end{align*}
For all $s\leq \ol{G}_0^{-1}\tilde{\alpha}$, 
\begin{align*}
 &\P\left(\wh{G}^*_0(s)-\ol{G}_0(s)\geq \delta \ol{G}_0(s)\right)\\
 \leq\:&\P\left( \sum_{i=k+1}^n (\ind{g^*(Y_i)\geq s}-\ol{G}_0(s))\geq (\ell+1)\delta \ol{G}_0(s)-1\right)\\
 \leq\: &\P\left( \sum_{i=k+1}^n (\ind{g^*(Y_i)\geq s}-\ol{G}_0(s))\geq 0.5 (\ell+1)\delta \ol{G}_0(s)\right),
\end{align*}
where the last line uses the fact that $(\ell+1)\delta \ol{G}_0(s)\ge  (\ell+1)\delta\tilde{\alpha}\geq 2$. 
Let $A=0.5 (\ell+1)\delta \ol{G}_0(s)$. Since the $g^*(Y_i)$'s are independent, By Bernstein's inequality, 
\begin{align}
\P(\Omega_2^c )&\leq
                 m \sup_{s\leq \ol{G}_0^{-1}(\tilde{\alpha})} \exp\left(-\frac{A^2}{2(\ell+1)\ol{G}_0(s)+2A/3}\right)\nonumber\\
               & = m \sup_{s\leq \ol{G}_0^{-1}(\tilde{\alpha})} \exp\left(-0.5 \frac{A^2}{4A/\delta+2A/3}\right)\nonumber\\
               & \le m \sup_{s\leq \ol{G}_0^{-1}(\tilde{\alpha})} \exp\left(-0.5 \frac{A^2}{4A/\delta+2A/3\delta}\right)\nonumber\\
  & = m \sup_{s\leq \ol{G}_0^{-1}(\tilde{\alpha})} \exp\left(-\frac{3(\ell + 1)\delta^2 \ol{G}_0(s)}{28}\right)\nonumber\\
 & \leq m\exp(-(3/28)  (\ell+1)\delta^2\tilde{\alpha}).\label{applibernsteinomega2}
 \end{align}
The proof of \eqref{powerBHstar} is then completed.

Next, we prove \eqref{powerBONuSstar}, i.e., the result for $\mbox{AdaDetect}_{\alpha}$. Recall that $\mathcal{R}= \{|\mbox{AdaDetect}^*_\alpha|\geq r\}$. Similar to $\Omega_2$, we define
\begin{align*}
  \Omega_3 &= \left\{ \sup_{1\leq i\leq m}\left(\frac{\ol{G}_0(g^*(X_i))\vee \tilde{\alpha} -\wh{G}^*_0(g^*(X_i))\vee \tilde{\alpha}}{\wh{G}^*_0(g^*(X_i))\vee \tilde{\alpha}}\right)\leq \delta \right\}\\
           &= \left\{ \sup_{1\leq i\leq m}\left(\frac{\ol{G}_0(g^*(X_i)\wedge \ol{G}_0^{-1}(\tilde{\alpha})) -\wh{G}^*_0(g^*(X_i)\wedge \ol{G}_0^{-1}(\tilde{\alpha}))}{\wh{G}^*_0(g^*(X_i)\wedge \ol{G}_0^{-1}(\tilde{\alpha}))}\right)\leq \delta \right\}.
\end{align*}
By \eqref{equintermG0bar}, on $\Omega_1\cap \Omega_2\cap \Omega_3$, for all $i\in \{1,\dots,m\}$ and $t\in \{\alpha k/m,r\vee 1\leq k\leq m\}$,
\begin{align*}
  \wh{G}_0(\hat{g}(X_i))&\le (\ol{G}_0(g^*(X_i))\vee \tilde{\alpha} )(1 + \zeta_r(\eta))(1 + \delta)\\
                        & \le (\wh{G}^*_0(g^*(X_i))\vee \tilde{\alpha})(1 + \zeta_r(\eta))(1 + \delta)^2\\
  & \le (\wh{G}^*_0(g^*(X_i))\vee \tilde{\alpha})(1 + \zeta_r(\eta))(1 + 3\delta).
\end{align*}
Thus, on $\Omega_1\cap \Omega_2\cap \Omega_3$,
\begin{align*}
\ind{\wh{G}^*_0(g^*(X_i))\vee \tilde{\alpha} \leq t}
&\leq 
\ind{
\wh{G}_0({g}(X_i))\leq t (1+3\delta)(1+\zeta_r(\eta))}.
\end{align*}
Applying Lemma~\ref{BHdomination} with $p_i=\wh{G}^*_0(g^*(X_i))$, $p_i'= \wh{G}_0({g}(Z_i))$, $\beta=\alpha$, and $\beta'=\alpha(1+\zeta_r(\eta))(1+3\delta)$), we obtain that $\mbox{AdaDetect}^*_\alpha\subset \mbox{AdaDetect}_{\alpha(1+3\delta)(1+\zeta_r(\eta))}$ on $\Omega_1\cap \Omega_2\cap \Omega_3 \cap \mathcal{R}$. It remains to prove that
\[\P(\Omega_3^c)\le m\exp(-(3/28)  (\ell+1)\delta^2\tilde{\alpha}).\]
Similar to $\P(\Omega_2^c)$, we have
\begin{align*}
\P(\Omega_3^c )&\leq  \P\left( \sup_{1\leq i\leq m}\left(\frac{\ol{G}_0(g^*(X_i)\wedge \ol{G}_0^{-1}(\tilde{\alpha})) -\wh{G}^*_0(g^*(X_i)\wedge \ol{G}_0^{-1}(\tilde{\alpha}))}{\wh{G}^*_0(g^*(X_i)\wedge \ol{G}_0^{-1}(\tilde{\alpha}))}\right)\geq \delta \right)\\
&\leq m \sup_{s\leq \ol{G}_0^{-1}(\tilde{\alpha})} \P\left(\frac{\ol{G}_0(s)-\wh{G}^*_0(s)}{\wh{G}^*_0(s)}\geq \delta \right)\\
&\leq m \sup_{s\leq \ol{G}_0^{-1}(\tilde{\alpha})} \P\left(\wh{G}^*_0(s) - \ol{G}_0(s) \leq -0.5 \delta  \ol{G}_0(s)\right),
\end{align*}
where the last line uses the fact that $(1+\delta)^{-1}\leq 1-0.5 \delta$ for any $\delta \in [0, 1]$. Now, for 
for all $s\leq \ol{G}_0^{-1}(\tilde{\alpha})$, 
\begin{align*}
&\P\left(\wh{G}^*_0(s) - \ol{G}_0(s) \leq -0.5 \delta  \ol{G}_0(s)\right)\\
 =\:&\P\left( \sum_{i=k+1}^n (\ind{g^*(Y_i)\geq s}-\ol{G}_0(s))\leq -0.5(\ell+1)\delta  \ol{G}_0(s)-1+\ol{G}_0(s) \right)\\
 \leq\: &\P\left( \sum_{i=k+1}^n (\ind{g^*(Y_i)\geq s}-\ol{G}_0(s))\leq -0.5(\ell+1) \delta  \ol{G}_0(s)\right).
\end{align*}
Applying the Bernstein inequality as in \eqref{applibernsteinomega2}, we conclude that $\P(\Omega_3^c )\leq m\exp(-(3/28)  (\ell+1)\delta^2\tilde{\alpha})$.

\begin{remark}
The inequality \eqref{powerBHstar} generalizes the power oracle inequality of \cite{mary2021semisupervised} for a fixed score function and $r=0$.
\end{remark}

\subsection{Proof of Corollary~\ref{cor:BHestimated}}\label{proof:cor:BHestimated}

Letting $\mathcal{R}=\{|\mbox{AdaDetect}^*_\alpha|\geq m_1\epsilon\}$, we have 
\begin{align*}
\TDR(\mbox{AdaDetect}^*_\alpha) &\leq \E[\TDP(\mbox{AdaDetect}^*_\alpha)]\\
 &\leq \epsilon + \int_\epsilon^1 \P(\TDP(\mbox{AdaDetect}^*_\alpha)\geq u)du\\
 &= \epsilon + \int_\epsilon^1 \P(\mathcal{R}\cap\{\TDP(\mbox{AdaDetect}^*_\alpha)\geq u\})du
\end{align*}
Now applying Theorem~\ref{th:BHestimated} with $r=\lceil m_1\epsilon\rceil$, we obtain  
\begin{align*}
  \TDR(\mbox{AdaDetect}^*_\alpha) &\leq  \epsilon + \int_\epsilon^1 \P(\TDP(\mbox{AdaDetect}_{\alpha'}) \geq u)du  +\P\left(\hat{\eta}> \eta \right) + 2 m e^{-(3/28)  (\ell+1)\delta^2\alpha\lceil m_1\epsilon\rceil/m}\\
  & \le   \TDR(\mbox{AdaDetect}_{\alpha'}) + \epsilon  +\P\left(\hat{\eta}> \eta \right) + 2 m e^{-(3/28)  (\ell+1)\delta^2\alpha\lceil m_1\epsilon\rceil/m}
\end{align*}
which gives \eqref{powerBONuSstarTDR}.

\section{Useful lemmas}

\begin{lemma}\label{BHdomination}
Let $(p_i,1\leq i\leq m)$ and $(p'_i,1\leq i\leq m)$ be two sets of $p$-values and $\beta,\beta'\in (0,1)$,  $r\in \{0,\dots,m\}$.  Assume that for all $t\in \{\beta k/m,r\vee 1\leq k\leq m\}$, 
\begin{equation}\label{equlemBHdomination}
\ind{p_i\leq t}\leq  \ind{p'_i\leq t \beta'/\beta}.
\end{equation}
Then, on the event where the BH algorithm applied to $p$-values $(p_i,1\leq i\leq m)$ at level $\beta$ has at least $r$ rejections, all these rejections would be rejected by the BH algorithm applied to $p$-values $(p'_i,1\leq i\leq m)$ at level $\beta'$.
\end{lemma}
\begin{proof}
Recall that BH algorithm applied to $p$-values $(p_i,1\leq i\leq m)$ at level $\beta$ rejects the $i$-th hypothesis iff $p_i\leq \beta \hat{k}/m$ where
$$
\hat{k} = \max\left\{k\in \{0,\dots,m\}\::\: \sum_{i=1}^m \ind{p_i\leq \beta k/m}\geq k\right\}.
$$
Similarly, the BH algorithm at level $\beta'$ rejects $i$-th hypothesis iff $p'_i\leq \beta \hat{k}'/m$ where
$$
\hat{k}' = \max\left\{k\in \{0,\dots,m\}\::\: \sum_{i=1}^m \ind{p'_i\leq \beta k/m}\geq k\right\}.
$$
When $\hat{k} = 0$, the conclusion is trivial. Now assume $\hat{k}\ge r\vee 1$. Setting $t=\beta \hat{k}/m$ in \eqref{equlemBHdomination}, we obtain that
$$\ind{p_i\leq \beta \hat{k}/m}\leq  \ind{p'_i\leq \beta' \hat{k}/m}, \:\:\:1\leq i\leq m.$$
This implies $\hat{k}'\geq \hat{k}\geq r\vee 1$. Hence, for all $i$,
$$
\ind{p_i\leq \beta \hat{k}/m}\leq \ind{p_i'\leq \beta' \hat{k}/m} \leq \ind{p_i'\leq \beta' \hat{k}'/m}.
$$
\end{proof}

\begin{lemma}[Vandermonde's equality]\label{lemvander}
For all $0\leq j \leq k\leq n$, we have
$$
{n+1 \choose k+1} = \sum_{m=j}^{n-k+j} {m\choose j} {n-m\choose k-j}.
$$
Equivalently, for all $j,u,v\geq 0$
$$
{j+u+v+1 \choose j+u+1} = \sum_{b=0}^{v} {j+b\choose j} {u+v-b\choose u}.
$$
\end{lemma}

\begin{lemma}\label{lem:emppvaluesfunctrank}
Consider any set of scores $(S_{k+1},\dots,S_{n+m})$ and the corresponding empirical $p$-values $(p_{1}, \ldots, p_{m})$ defined in \eqref{emppvalues}.
Then on the event where the scores $(S_j,j\in \{k+1,\dots,n\}\cup (n+ \cH_0))$ are mutually distinct, the null $p$-values $(p_j,j\in \cH_0)$ are mesurable with respect to the ranks of $(S_j,j\in \{k+1,\dots,n\}\cup (n+\cH_0))$.
\end{lemma}

\begin{proof}
Let $R_1,\dots,R_{n-k+m_0}$ and $S_{(1)},\dots,S_{(n-k+m_0)}$ be the ranks and the order statistics of  $(S_j,j\in \{k+1,\dots,n\}\cup (n+\cH_0))$, respectively. 
Since the score $(S_j,j\in \{k+1,\dots,n\}\cup (n+\cH_0))$ are mutually distinct, we have that $S_i>S_j$ iff $R_{i}< R_{j}$. Hence, by \eqref{emppvalues}, we obtain for all $j\in \cH_0$,
\begin{align*}
p_j &= \frac{1}{n-k+1}\left(1+\sum_{i=k+1}^{n}\ind{S_i>S_{n+j}}\right)= \frac{1}{n-k+1}\left(1+\sum_{i=k+1}^{n}\ind{R_i<R_{n+j}}\right).
\end{align*}
\end{proof}

Lemma~\ref{lem:emppvaluesfunctrank} implies the following equivalent representation of the distribution of $p$-values.

\begin{lemma}\label{lem:reduction}
  For any set of scores $(S_{k+1},\dots,S_{n+m})$ satisfying Assumptions~\ref{as:newexch}~and~\ref{as:noties}, the joint distribution of the null empirical $p$-values $(p_j,j\in \cH_0)$ does not depend on the joint distribution of scores. In particular, it can be characterized by generating $(S_j,j\in \{k+1,\dots,n\}\cup (n+ \cH_0))$ i.i.d. from $U(0,1)$.
\end{lemma}

\begin{proof}
By Lemma~\ref{lem:emppvaluesfunctrank}, almost surely, $(p_j,j\in \cH_0)$ is a function of the ranks of $(S_j,j\in \{k+1,\dots,n\}\cup (n+\cH_0))$, which is uniformly distributed on the permutations of $\{1,\dots,n-k+m_0\}$ according to Assumption~\ref{as:newexch}. The result follows.
\end{proof}

\begin{lemma}\label{positivescoreBH} 
Fix $M\in\{1,\dots,m\}$ and $\eta\in (0,1)$. Assume that the scores $S_{k+1},\dots,S_n$ in the NTS satisfy
$
1+\sum_{i=k+1}^n \ind{S_i\geq 0}\leq \eta \ell
$
and the scores $S_{n+1},\dots,S_{n+m}$ in the test sample satisfy
$
\sum_{j=1}^m \ind{S_{n+j}\geq 0}\geq M.
$
Then, if $\eta\leq \alpha M/m$, $\mbox{AdaDetect}_\alpha$ would reject all hypotheses with nonnegative scores, i.e., 
$$
\{j\in \{1,\dots,m\} \::\: S_{n+j}\geq 0\}\subset \mbox{AdaDetect}_\alpha.
$$
\end{lemma}

\begin{proof}
By definition, for each $j\in \{1,\dots,m\}$, the empirical $p$-value $p_j$ \eqref{emppvalues} satisfies
\begin{align*}
p_j&= \frac{1}{\ell+1}\left(1+\sum_{i=k+1}^{n}\ind{S_i>S_{n+j}}\right)\\
&\leq \frac{1}{\ell+1}\left(1+\sum_{i=k+1}^{n}(\ind{S_i\geq 0 } +\ind{ S_{n+j}<0})\right) \\
&\leq \eta + \ind{ S_{n+j}<0},
\end{align*}
by the assumptions. Hence, letting $M'=\sum_{j=1}^m \ind{S_{n+j}\geq 0}\geq M$, we have
$$
\sum_{j=1}^m \ind{p_j\leq \alpha M'/m}\geq \sum_{j=1}^m \ind{p_j\leq \alpha M/m}\geq \sum_{j=1}^m \ind{p_j\leq \eta}\geq   \sum_{j=1}^m \ind{S_{n+j}\geq 0} = M'
$$
Since $\mbox{AdaDetect}_\alpha$ is the BH algorithm applied to the empirical $p$-values, the result follows.
\end{proof}

\begin{lemma}\label{BHsmallerp} 
Write the number of rejections $\wh{k}=\wh{k}(p_i,1\leq i\leq m)$ given by \eqref{equkchapeau} a function of $p$-values. Fix any $i\in \{1,\dots,m\}$ and consider two sets of $p$-values $(p_j,1\leq j\leq m)$ and $(p'_j,1\leq j\leq m)$ which satisfy almost surely that
\begin{align}\label{propppprime}
\forall j\in \{1,\dots,m\}, \left\{\begin{array}{cc} p'_j\leq p_j& \mbox{ if } p_j\leq p_i\\p'_j = p_j& \mbox{ if } p_j> p_i\end{array}\right.  
\end{align}
Let $\wh{k}=\wh{k}(p_i,1\leq i\leq m)$ and $\wh{k}'=1\vee \wh{k}(p'_i,1\leq i\leq m)$. Then
$$
\{p_i\leq \alpha \wh{k}/m \}=\{ p_i\leq \alpha \wh{k}'/m \}\subset \{ \wh{k}=\wh{k}'\}.
$$
\end{lemma}

This lemma is closely related to many previous results on the structure of the BH algorithm; see, e.g., \cite{FZ2006,RV2011,ramdas2019unified}. It states that the rejected $p$-values can be made arbitrarily smaller without changing the number of rejections.

\begin{proof}
First, since $p'_j\leq p_j$ for all $j\in \{1,\dots,m\}$, we clearly have $\wh{k}\leq \wh{k}'$.  Now we prove the  equivalence
\begin{equation}\label{equivtoprove}
p_i\leq \alpha \wh{k}/m \Longleftrightarrow p_i\leq \alpha \wh{k}'/m .
\end{equation}
Clearly, $p_i\leq \alpha \wh{k}/m$ implies $p_i\leq \alpha \wh{k}'/m$.
Now we prove the other direction that $p_i\leq \alpha \wh{k}'/m$ implies $\wh{k}\ge \wh{k}'$. 
Note that
\begin{align*}
\sum_{j=1}^m \ind{p_j\leq \alpha \wh{k}'/m} &= \sum_{j=1}^m \ind{p_j\leq p_i} \ind{p_j\leq \alpha \wh{k}'/m}
+ \sum_{j=1}^m \ind{p_j> p_i} \ind{p_j\leq \alpha \wh{k}'/m}\\
&= \sum_{j=1}^m \ind{p_j\leq p_i} 
+ \sum_{j=1}^m \ind{p_j> p_i} \ind{p'_j\leq \alpha \wh{k}'/m}\\
&= \sum_{j=1}^m \ind{p_j\leq p_i} \ind{p'_j\leq \alpha \wh{k}'/m}
+ \sum_{j=1}^m \ind{p_j> p_i} \ind{p'_j\leq \alpha \wh{k}'/m}\\
&=\sum_{j=1}^m \ind{p'_j\leq \alpha \wh{k}'/m} \geq 1\vee \wh{k}(p'_i,1\leq i\leq m)  = \wh{k}',
\end{align*}
where the second and third lines is due to \eqref{propppprime} and that $p_i\leq \alpha \wh{k}'/m$. The fourth line uses the definition of $\wh{k}(p'_i,1\leq i\leq m)$ and that $p_i\leq \alpha \wh{k}'/m$). By definition of $\wh{k}$, the above inequality implies $\wh{k}\ge \wh{k}'$. Thus, we must have $\wh{k} = \wh{k}'$ and the result follows.
\end{proof}

\section{Auxiliary results for Section~5}

\subsection{Bounding $\zeta_r(\cdot)$}\label{sec:boundzeta}

In this section, we provide explicit bounds on $\zeta_r(\cdot)$ for any given $r\in \{0,\dots,m\}$ and sample size $m$ in the two following cases.

\begin{lemma}\label{lem:compuzetaunif}
Assume $P_0 = {U}([0,1]^d)$  and $P_i={U}([0,1]^{d-1})\otimes Q$ for any $i\in \cH_1$ where $Q$ is a distribution supported on $[0,1]$ with a strictly decreasing and differentiable density function $h(x)$ on $(0,1)$. Assume $h(0)>0$ and $c=\inf_{x\in (0,1)} |h'(x)|>0$. For any $\alpha\in (0,1)$,  $m\geq 1$, and $r\in\{0,\dots,m\}$, the following results hold.
\begin{itemize}
\item[(i)] If $g^*(x)=A+B f_1(x)/f_0(x)$ for some $A\in \R$ and $B>0$, then, for any $\eta<(B/2)\cdot (h(\alpha)-h(1))$,
\begin{align}\label{equzetafirstunif}
\zeta_r(\eta)\leq \frac{m}{\alpha (r\vee 1)}\frac{2}{B c} \eta.
\end{align}
\item[(ii)] If $g^*(x)=\Psi(f_1(x) / f_0(x))$ where $\Psi(y) = 1/(1+(A+B y)^{-1})$ for some $A, B>0$, then, for any $\eta<(1/2)(\Psi\circ h(\alpha)-\Psi \circ h(1))$, 
\begin{align}\label{equzetasecondunif}
\zeta_r(\eta)\leq  \frac{m}{\alpha (r\vee 1)}\frac{2(A+Bh(0)+1)^2}{B c} \eta.
\end{align}
\end{itemize}
\end{lemma}

\begin{proof}[Proof of Lemma \ref{lem:compuzetaunif}]
  Under the assumptions, $f_1(x) / f_0(x) = h(x_d)$ where $x_d$ is the $d$-th coordinate of $x$. For case (i), $g^*(x)=A+B f_1(x)/f_0(x)=A+B h(x_d)$. Since $h$ is strictly decreasing, for any $s\in [A+Bh(1),A+Bh(0)]$, 
\begin{align*}
\ol{G}_0(s)=\P_{X_d\sim {U}([0, 1])}(A+B h(X_d)\geq s)= \P(X_d\leq h^{-1}((s-A)/B))=h^{-1}((s-A)/B).
\end{align*}
Thus, $\ol{G}_0^{-1}(u)=A+Bh(u)$ for all $u\in [0,1]$. Fix any $\eta<(B/2)(h(\alpha)-h(1))$. Then for any $u\in [0, \alpha]$, 
\[A + Bh(0) \ge \ol{G}_0^{-1}(u)\ge \ol{G}_0^{-1}(u) - 2\eta\ge \ol{G}^{-1}(\alpha) - 2\eta  > A + Bh(\alpha) - B(h(\alpha) - h(1)) = A + Bh(1).\]
Note that both $\ol{G}_0$ and $\ol{G}_0^{-1}$ are decreasing,
\begin{align*}
  \zeta_r(\eta)
  &= \max_{u\in [\alpha (r\vee 1)/m,\alpha]} \left\{\frac{\ol{G}_0(\ol{G}_0^{-1}(u)-2\eta) - u}{u }\right\}\\
&\leq \frac{m}{\alpha (r\vee 1)} \max_{u\in [\alpha (r\vee 1)/m,\alpha]} (h^{-1}(h(u)-2\eta/B) - h^{-1}(h(u)))\\
&\leq \frac{m}{\alpha (r\vee 1)} \frac{2\eta}{B} \max_{u\in [\alpha (r\vee 1)/m,\alpha]} |(h^{-1})'(u)|.
\end{align*}
The result is then proved by noting that $(h^{-1})'(u) = 1 / h'(h^{-1}(u))$ and the assumption that $|h'(h^{-1}(u))| > c$. 

For case (ii), $g^*(x)= 1/(1+(A+B h(x_d))^{-1}) = \Psi\circ h(x_1)$, where $\Psi(u) = (A+Bu)/(A+Bu+1)$ for any $u\in [h(1), h(0)]$. Note that $\Psi'(u)=B/(A+Bu+1)^2$. Then for any $s$ in $[1/(1+(A+B h(1))^{-1}),1/(1+(A+B h(0))^{-1})]$, the image of $\Psi$,
\begin{align*}
\ol{G}_0(s)&=\P_{X_d\sim {U}([0, 1])}(\Psi(h(X_d))\geq s)= (\Psi\circ h)^{-1}(s).
\end{align*}
Hence, for all $u\in [0,1]$, 
$
\ol{G}_0^{-1}(u)= \Psi\circ h (u).
$
Fix any
$\eta< (1/2)(\Psi\circ h(\alpha) - \Psi\circ h(1))$. Then for any $u\in [\alpha (r\vee 1)/m,\alpha]$, 
\[\Psi \circ h(0) \ge \ol{G}_0^{-1}(u)\ge \ol{G}_0^{-1}(u)-2\eta \ge \ol{G}_0^{-1}(\alpha)-2\eta > \Psi\circ h(1).\]
Note that both $\ol{G}_0$ and $\ol{G}_0^{-1}$ are decreasing,
\begin{align*}
\zeta_r(\eta) &= \max_{u\in [\alpha (r\vee 1)/m,\alpha]} \left\{\frac{\ol{G}_0(\ol{G}_0^{-1}(u)-2\eta) - u}{u }\right\}\\
              &\leq \frac{m}{\alpha (r\vee 1)} \max_{u\in [\alpha (r\vee 1)/m,\alpha]} \left((\Psi\circ h)^{-1}(\Psi\circ h(u)-2\eta) - u\right)\\
  &\leq \frac{m}{\alpha (r\vee 1)} \max_{u\in [\alpha (r\vee 1)/m,\alpha]} \left((\Psi\circ h)^{-1}(\Psi\circ h(u)-2\eta) - (\Psi\circ h)^{-1}(\Psi\circ h(u))\right)\\
&\leq \frac{m}{\alpha (r\vee 1)}  (2\eta) \max_{u\in [\alpha (r\vee 1)/m,\alpha]} |((\Psi\circ h)^{-1})'(u)|\\
&\leq \frac{m}{\alpha (r\vee 1)} (2\eta) \frac{1}{\inf_{x\in(0,1)}|h'(x)|\times \inf_{y\in [h(1),h(0)]}|\Psi'(y)|}.
\end{align*}
The result is proved by noting that $\inf_{y\in [h(1),h(0)]}|\Psi'(y)|= B/(A+Bh(0)+1)^2$.
\end{proof}

\begin{lemma}\label{lem:compuzeta}
Assume $P_0=\mathcal{N}(\mu_0,I_d)$ and $P_i=\mathcal{N}(\mu,I_d)$ for any $i\in \cH_1$, where $\mu_0\neq \mu_1$ are the null and alternative mean vectors, respectively. For any $\alpha\in (0,1)$, $m\geq 1$, and $r\in \{0,\dots,m\}$, the following results hold.
\begin{itemize}
\item[(i)] If $g^*(x)=A+B f_1(x)/f_0(x)$ for some $A\in \R$ and $B>0$, then for any $\alpha \in (0, \ol{\Phi}(1))$ and $\eta\in [0, 1]$ with $4\eta\leq (eb^2B')\wedge B'$,
\begin{align}\label{equzetafirst}
\zeta_r(\eta)\leq C \eta,
\end{align}
where $C=C(B,\mu,\mu_0)= \frac{8}{b^2B'}$, $B'= B e^{- b^2/2}$ and $b=\|\mu-\mu_0\|$.
\item[(ii)] If $g^*(x)=1/(A+B f_0(x)/f_1(x))$, $A>0$, $B>0$, and $\alpha \in (0, \ol{\Phi}(1))$, $\eta\in [0, 1]$ with $4\eta(A+B'' e^{ -b \ol{\Phi}^{-1}(\alpha)})\leq 1$ and $ \eta C e^{ (b+1) \sqrt{2\log (m/(\alpha(r\vee 1)))}}/(2e)\leq 1$, 
\begin{align}\label{equzetasecond}
\zeta_r(\eta)\leq  C \eta e^{ (b+1) \sqrt{2\log (m/(\alpha(r\vee 1)))} },
\end{align}
for $C=C(A,B,\mu,\mu_0)=8e  (A+B'')^2/(bB'')$ with $B''= B e^{ b^2/2}$ and $b=\|\mu-\mu_0\|$. 
 \item[(iii)] For $g^*(x)=1/(1+(A+B f_1(x)/f_0(x))^{-1})$, $A>0$, $B>0$, and $\alpha \in (0, \ol{\Phi}(1))$, $\eta\in [0, 1]$ with 
 $16\eta\left((A+B'e^{ b \sqrt{2\log (m/(\alpha(r\vee 1)))} })^2\vee 1\right)/B'\leq 1\wedge  b$, then
\begin{align}\label{equzetathird}
\zeta_r(\eta)\leq  C \eta \left((A+B'e^{ b \sqrt{2\log (m/(\alpha(r\vee 1)))} })^2\vee 1\right),
\end{align}
for $C=C(B,\mu,\mu_0)=64 e /(bB')$, $B'= B e^{- b^2/2}$ and $b=\|\mu-\mu_0\|$.
\end{itemize}
\end{lemma}

\begin{proof}
Let us first consider the case of $g^*(x)=A+B f_1(x)/f_0(x)$. We thus have
\begin{align*}
g^*(x)&=A+ B \exp\left\{(x-\mu_0)^T(\mu-\mu_0) -(1/2) \|\mu-\mu_0\|^2  \right\}\\
&= \Psi((x-\mu_0)^T(\mu-\mu_0)/\|\mu-\mu_0\|) 
\end{align*}
where we let $\Psi(t)=A+B' \exp( b t ),$ $t\in \R$, for $B'= B e^{-(1/2) \|\mu-\mu_0\|^2}$ and $b=\|\mu-\mu_0\|$. Hence $\Psi^{-1}(v)= b^{-1}\log\left((v- A)/ B'\right)$ for $v>A$.
In that case, we have 
\begin{align*}
\Psi^{-1}(\Psi(t)-2\eta)&=\Psi^{-1}(A+B' e^{b t} -2\eta)=b^{-1}\log\left(e^{ b t }-2\eta/ B' \right) \\
&= t  + b^{-1}\log\left(1-2\eta e^{ -b t }/B' \right) 
\end{align*}
 Since $\log(1-x)\geq -2x$ for all $x\in [0,1/2]$, we have $\log\left(1-2\eta e^{ -b t }/B' \right)\geq -4\eta e^{ -b t }/B'$ because $4\eta e^{ -b t }/B'\leq 1$.
This entails that for $4\eta e^{ -b \ol{\Phi}^{-1}(\alpha) }/B'\leq 1$, for $u\in [\alpha (r\vee 1)/m,\alpha]$, (by taking $t=\ol{\Phi}^{-1}(u)$ in the  above relations)
\begin{align*}
\ol{\Phi}\circ\Psi^{-1}(\Psi\circ\ol{\Phi}^{-1}(u)-2\eta) - u &\leq \ol{\Phi}\left(\ol{\Phi}^{-1}(u)    -4  \eta e^{ -b \ol{\Phi}^{-1}(u) }/(bB')\right) - u.
\end{align*}
Now, we can use \eqref{eq:gaussian} in Lemma~\ref{lem:gaussian} with $y=4  \eta e^{ -b \ol{\Phi}^{-1}(u) }/(bB')$ (checking that $u \leq \alpha \leq \ol{\Phi}(1)$) to obtain that
\begin{align*}
\ol{\Phi}\circ\Psi^{-1}(\Psi\circ\ol{\Phi}^{-1}(u)-2\eta)  &\leq
8  u  \eta e^{ -b \ol{\Phi}^{-1}(u) } (bB')^{-1} \ol{\Phi}^{-1}(u)  \exp(4  \eta e^{ -b \ol{\Phi}^{-1}(u) } (bB')^{-1} \ol{\Phi}^{-1}(u) )
 \\
 &\leq \frac{8u}{eb^2B'} \eta \exp(   4 \eta/(eb^2B' ) ),
\end{align*}
because $\forall x\geq 1$, we have $x  e^{-xb} \leq 1/(eb)$. This gives \eqref{equzetafirst}.

Let us now turn to prove \eqref{equzetasecond} by considering $g^*(x)=1/(A+B f_0(x)/f_1(x))$, $A>0$, $B>0$. Similarly to above, we have 
\begin{align*}
g^*(x)&= \Psi((x-\mu_0)^T(\mu-\mu_0)/\|\mu-\mu_0\|) 
\end{align*}
where we let $\Psi(t)=1/(A+B'' \exp( -b t )),$ $t\in \R$, for $B''= B e^{(1/2) \|\mu-\mu_0\|^2}$ and $b=\|\mu-\mu_0\|$. Hence $\Psi^{-1}(v)= - b^{-1}\log\left((1/v- A)/ B''\right)$ for $v<1/A$.
In that case, we have 
\begin{align*}
\Psi^{-1}(\Psi(t)-2\eta)&=\Psi^{-1}(1/(A+B'' e^{ -b t} ))-2\eta)\\
&= - b^{-1}\log\left(\left(\frac{ A+B'' e^{ -b t}}{1-2\eta(A+B'' e^{ -b t})} -A \right)/B''\right). 
\end{align*}
Now using that $1/(1-x)\leq 1+2x$ for all $x\in [0,1/2]$, we have that for $4\eta(A+B'' e^{ -b t})\leq 1$,
\begin{align*}
\left(\frac{ A+B'' e^{ -b t}}{1-2\eta(A+B'' e^{ -b t})} -A \right)/B''
&\leq  e^{ -b t} + 4\eta (A+B'' e^{ -b t})^2/B''
\end{align*}
This entails
\begin{align*}
\Psi^{-1}(\Psi(t)-2\eta)
&\geq  - b^{-1}\log\left(e^{ -b t} + 4\eta (A+B'' e^{ -b t})^2/B''\right)\\
&= t  - b^{-1}\log\left(1+ 4\eta (A+B'' e^{ -b t})^2e^{ b t}/B''\right) \\
&\geq  t  -  4\eta (A+B'' e^{ -b t})^2e^{ b t}/(bB'')\\
&\geq  t  -  4\eta (A+B'')^2 e^{ b t}/(bB''),
\end{align*}
because $\log(1+x)\leq x$ for all $x\geq 0$. 
Hence for $4\eta(A+B'' e^{ -b \ol{\Phi}^{-1}(\alpha)})\leq 1$, for $u\in [\alpha(r\vee 1)/m,\alpha]$, 
\begin{align*}
\ol{\Phi}\circ\Psi^{-1}(\Psi\circ\ol{\Phi}^{-1}(u)-2\eta) - u &\leq \ol{\Phi}\left(\ol{\Phi}^{-1}(u)   -  4\eta (A+B'')^2e^{ b \ol{\Phi}^{-1}(u)}/(bB'')\right) - u.
\end{align*}
Now, we can use \eqref{eq:gaussian} in Lemma~\ref{lem:gaussian} with $y=4\eta (A+B'')^2e^{ b \ol{\Phi}^{-1}(u)}/(bB'')$ to obtain
\begin{align*}
\ol{\Phi}\circ\Psi^{-1}(\Psi\circ\ol{\Phi}^{-1}(u)-2\eta)  
 &\leq 2 u h\left( 4  (A+B'')^2   \eta e^{ (b+1) \ol{\Phi}^{-1}(u) } (bB'')^{-1} \right)\\
 &\leq 2 u h\left( 4  (A+B'')^2   \eta e^{ (b+1)  \sqrt{2\log (m/(\alpha(r\vee 1)))}} (bB'')^{-1} \right),
\end{align*}
for $h(x)=xe^x$ and because $\ol{\Phi}^{-1}(u)\leq \ol{\Phi}^{-1}(\alpha(r\vee 1)/m)\leq  \sqrt{2\log (m/(\alpha(r\vee 1)))}$ (and using that $\forall x\geq 0$, we have $x  \leq e^x$). This gives \eqref{equzetasecond} because $x e^x \leq e x$ when $x\leq 1$. 

Let us now turn to prove \eqref{equzetathird} by considering $g^*(x)=1/(1+(A+B f_1(x)/f_0(x))^{-1})$, $A>0$, $B>0$. Similarly to above, we have 
\begin{align*}
g^*(x)&= \Psi((x-\mu_0)^T(\mu-\mu_0)/\|\mu-\mu_0\|) 
\end{align*}
where $\Psi(t)=1/(1+(A+B' e^{bt})^{-1}),$ $t\in \R$, for $B'= B e^{-(1/2) \|\mu-\mu_0\|^2}$ and $b=\|\mu-\mu_0\|$. Hence $\Psi^{-1}(v)= b^{-1}\log\left(((1/v- 1)^{-1}-A)/ B'\right)$ for $v\in (0,1)$.
In that case, we have for $ 4 \eta \Psi(t)^{-1} \leq 1$, 
\begin{align*}
\frac{1}{\Psi(t)-2\eta}- 1&=\Psi(t)^{-1}\frac{1}{1-2\eta \Psi(t)^{-1}}- 1\\
&\leq \Psi(t)^{-1}(1+4\eta \Psi(t)^{-1})- 1\\
&=(1+(A+B' e^{bt})^{-1})(1+4\eta \Psi(t)^{-1})- 1\\
&=4\eta \Psi(t)^{-1} +(A+B' e^{bt})^{-1}(1+4\eta \Psi(t)^{-1}),
\end{align*}
by using $1/(1-x)\leq 1+2x$, $x\in [0, 1/2 ]$.
Hence, 
\begin{align*}
\left(\frac{1}{\Psi(t)-2\eta}- 1\right)^{-1}
&\geq \frac{A+B' e^{bt}}{1+4\eta \Psi(t)^{-1}} \frac{1}{1+ 4\eta \Psi(t)^{-1}\frac{A+B' e^{bt}}{1+4\eta \Psi(t)^{-1}}}\\
&\geq \frac{A+B' e^{bt}}{1+4\eta \Psi(t)^{-1}} \frac{1}{1+ 4\eta \Psi(t)^{-1}(A+B' e^{bt})}\\
&\geq (A+B' e^{bt}) (1- 4 \eta \Psi(t)^{-1}) \left(1- 4 \eta \Psi(t)^{-1}(A+B' e^{bt})\right)\\
&\geq (A+B' e^{bt}) \left(1- 8\eta \Psi(t)^{-1}\left((A+B' e^{bt})\vee 1\right)\right)\\
&\geq (A+B' e^{bt}) \left(1- 16 \eta \frac{(A+B'e^{ b t })^2\vee 1}{A+B'e^{bt}}\right),
\end{align*}
by using $1/(1+x)\geq 1-x$ and $(1-x)^2\geq 1-2x$, $x\in [0,1]$,  $\Psi(t)^{-1}\left((A+B' e^{bt})\vee 1\right)\leq 2 \frac{(A+B'e^{ b t })^2\vee 1}{A+B'e^{bt}}$, and provided that $16 \eta \frac{(A+B'e^{ b t })^2\vee 1}{A+B'e^{bt}}\leq 1$. This entails
\begin{align*}
\left(\left(\frac{1}{\Psi(t)-2\eta}- 1\right)^{-1}-A\right)/B'
&\geq e^{bt} - 16\eta\left((A+B'e^{ b t })^2\vee 1\right)/B'.
\end{align*}
Thus, we have
\begin{align*}
\Psi^{-1}(\Psi(t)-2\eta)
&\geq  b^{-1}\log\left(e^{bt} - 16\eta\left((A+B'e^{ b t })^2\vee 1\right)/B'\right)\\
&= t  + b^{-1}\log\left(1 - 16\eta\left((A+B'e^{ b t })^2\vee 1\right)e^{-bt}/B' \right)\\
&\geq t  + b^{-1}\log\left(1 - 16\eta\left((A+B'e^{ b t })^2\vee 1\right)/B' \right)\\
&\geq  t -32\eta\left((A+B'e^{ b t })^2\vee 1\right)/(bB'),
\end{align*}
because $\log(1-x)\geq -2x$ for all $x\in [0,1/2]$ and provided that $16\eta\left((A+B'e^{ b t })^2\vee 1\right)/B'\leq 1/2$. 
(Also note that the latter condition implies both previous conditions 
$4\eta \Psi(t)^{-1}=4\eta (1+(A+B' e^{bt})^{-1})\leq 1$ and $16 \eta \frac{(A+B'e^{ b t })^2\vee 1}{A+B'e^{bt}}\leq 1$).

Hence for $16\eta\left((A+B'e^{ b \ol{\Phi}^{-1}(\alpha/m) })^2\vee 1\right)/B'\leq 1/2$ and $u\in [\alpha(r\vee 1)/m,\alpha]$, 
\begin{align*}
\ol{\Phi}\circ\Psi^{-1}(\Psi\circ\ol{\Phi}^{-1}(u)-2\eta) - u &\leq \ol{\Phi}\left(\ol{\Phi}^{-1}(u)  -32\eta\left((A+B'e^{ b \ol{\Phi}^{-1}(u) })^2\vee 1\right)/(bB')\right) - u.
\end{align*}
Now, we can use \eqref{eq:gaussian} in Lemma~\ref{lem:gaussian} with $y=32\eta\left((A+B'e^{ b \ol{\Phi}^{-1}(u) })^2\vee 1\right)/(bB')$ to obtain 
\begin{align*}
\ol{\Phi}\circ\Psi^{-1}(\Psi\circ\ol{\Phi}^{-1}(u)-2\eta)  
 &\leq 2 u h\left( 32\eta\left((A+B'e^{ b \ol{\Phi}^{-1}(u) })^2\vee 1\right)/(bB')\right)\\
 &\leq 2 u h\left( 32\eta\left((A+B'e^{ b \sqrt{2\log (m/(\alpha(r\vee 1)))} })^2\vee 1\right)/(bB') \right),
\end{align*}
for $h(x)=xe^x$ and because $\ol{\Phi}^{-1}(u)\leq \ol{\Phi}^{-1}(\alpha(r\vee 1)/m)\leq  \sqrt{2\log (m/(\alpha(r\vee 1)))}$ (and using that $\forall x\geq 0$, we have $x  \leq e^x$). This gives \eqref{equzetathird}.
\end{proof}

\begin{lemma}\label{lem:gaussian}
For all $y\geq 0$ and $u\in (0,\ol{\Phi}(1)]$, we have
\begin{align}
\ol{\Phi}(\ol{\Phi}^{-1}(u)-y) - u \leq 2  u  y \ol{\Phi}^{-1}(u)  \exp(y \ol{\Phi}^{-1}(u) ).\label{eq:gaussian}
\end{align}
\end{lemma}

\begin{proof}
By using the classical relations on upper-tail distribution of standard Gaussian, we have 
\begin{align*}
\ol{\Phi}(\ol{\Phi}^{-1}(u)-y) - u &\leq y\left( \phi(\ol{\Phi}^{-1}(u)-y ) \vee \phi(\ol{\Phi}^{-1}(u) ) \right)\\
& = y \phi(\ol{\Phi}^{-1}(u)) \left(1\vee \exp(y \ol{\Phi}^{-1}(u) - y^2/2)\right)\\
&\leq  2  u  y \ol{\Phi}^{-1}(u) \exp(y \ol{\Phi}^{-1}(u) ),
\end{align*}
since $\ol{\Phi}^{-1}(u)\geq 1$, because $\phi(x)\leq 2x \ol{\Phi}(x)$ for all $x\geq 1$. 
\end{proof}

\subsection{Case of density estimation}\label{app:densityestim}

\subsubsection{Consistency}\label{sec:consistency}

Consider $g^*$ given by \eqref{equ:gstardensity}, that is $g^*(x)=f_\gamma(x)/f_0(x)$. Assuming for simplicity that $f_0$ is known (hence $k=0$ and $n=\ell$ here), we propose the estimator $g(x)=\hat{f}_\gamma(x)/f_0(x)$, where $\hat{f}_\gamma$ is the  histogram estimator of $f_\gamma$ given by
\begin{equation}\label{equ:histogram}
\hat{f}_\gamma(x)= M^d\sum_{j=1}^{M^d} (n+m)^{-1} \sum_{i=1}^{n+m} \ind{Z_i\in \mathcal{D}_j}  \ind{x\in \mathcal{D}_j},\:\: x\in [0,1]^d,
\end{equation}
where $\{\mathcal{D}_1,\dots,\mathcal{D}_{M^d}\}$ is a regular partition of $[0,1]^d$ formed by $M^d$ $d$-dimensional cubes of side size $1/M$ and Lebesgue measure $|\mathcal{D}_1|=1/M^d$, with $M=\lceil (n+m)^{1/(2+d)}\rceil$.

Remember that the corresponding AdaDetect procedure controls the FDR even if the estimation quality of $\hat{f}_\gamma$ is poor. In addition, we show in this section that, in a suitable case where the estimation quality is good enough, the power of AdaDetect consistently converges to that of the oracle, that is, \eqref{equconsistency} holds. We use for this Corollary~\ref{cor:BHestimated}.

Assume $m\asymp \l = n \gg m/m_1$, let Assumptions~\ref{as:indep} and~\ref{equ-marg}  be true and consider the uniformly bounded case described in Section~\ref{sec:boundzeta}. 
Then we have both \eqref{condetachap} with $\kappa\in (0,1/(2+d))$ (Lemma~\ref{lem:etachapdensityestimation}, see next section) and $\zeta_{\lceil m_1\epsilon\rceil }(\eta)\lesssim \epsilon^{-1} \frac{m}{\alpha m_1} \eta / \gamma \asymp  \epsilon^{-1}\eta / (m_1/m)^2$ for $\eta$ small enough (observe that $\gamma\sim m_1/m$). Hence, in the not too sparse scenario 
$$
m_1/m \gg m^{-\kappa/2}
$$
(including the dense case $m_1\asymp m$), 
choosing $\eta\asymp m^{-\kappa}$, $\delta=\epsilon = 1/\log( (m_1/m)^2/ m^{-\kappa})$, 
we have 
$\ell\delta^2 \epsilon m_1/m = m^{1-\kappa/2}  \frac{(m_1/m)/ m^{-\kappa/2}}{\log^3( (m_1/m)^2/ m^{-\kappa})} \gg 1$, which 
ensures \eqref{equ:condparaconsist} and thus proves the consistency \eqref{equconsistency}.

\subsubsection{Bounding $\hat{\eta}$ in case of density estimation}\label{sec:boundetachap}

\begin{lemma}
\label{lem:etachapdensityestimation}
Let Assumptions~\ref{as:indep} and~\ref{equ-marg} be true and consider the case where $P_0 = {U}([0,1]^d)$  and $P_i$, $i\in \cH_1$ have a common distribution, supported on $[0,1]^d$, with a density $f_1$ which is $L$-Lipschitz and uniformly upper-bounded by $C^d$. Then the estimator $g(x)=\hat{f}_\gamma(x)/f_0(x)$ of $g^*$ given above satisfies that for $n+m\geq N(C,d)$, 
\begin{equation}\label{equetachapdensesti}
\P\left(\hat{\eta}\geq c_0   (n+m)^{-1/(2+d)} \sqrt{\log (n+m)}\right)\leq 2/(n+m),
\end{equation}
where $\hat{\eta}$ is given by \eqref{diffscore} and  $c_0(d,C,L)=2L+8(2C)^{d/2}$.
\end{lemma}

\begin{proof}
First note that when $f_1$ is $L$-Lipschitz and upper bounded by $C^d$, then $f_\gamma$ is also $L$-Lipschitz and upper bounded by $ C^d$ (because $\gamma\leq 1$).
Hence, by Lemma~\ref{lem:kerneldensity} we have for some constant $c_0=c(d,C,L)>0$ and $N=N(C)$, for all $m\geq N$, for all $x\in [0,1]^d$,
\begin{equation}\label{equ:histogramboundapplied}
\P\left( |\hat{f}_\gamma(x)- f_\gamma(x)|\geq  c_0 (n+m)^{-1/(2+d)} \sqrt{\log (n+m)} \right)\leq 2/(n+m)^2.
\end{equation}
Now, we have for all $\eta>0$,
\begin{align*}
\P(\hat{\eta}\geq \eta) =\P\left( \max_{1\leq i\leq n+m}|g^*(Z_i)-g(Z_i)| \geq \eta\right)
\leq \sum_{i=1}^{n+m} \P\left( |\hat{f}_\gamma(Z_i)  -f_\gamma(Z_i) | \geq \eta\right),
\end{align*}
because $f_0$ is the density of the uniform distribution on $[0,1]^d$. 
A technical point here is that $Z_i$ also appears in $\hat{f}_\gamma$ (at one place), nevertheless, denoting $\hat{f}'_\gamma$ the quantities \eqref{equ:histogram} for which $Z_i$ has been changed to an independent copy $Z_i'$ (at that place), we have for all $x\in [0,1]^d$,
$$
|\hat{f}_\gamma(x) - \hat{f}'_\gamma(x)|\leq  M^d  (n+m)^{-1} \leq 2^d (n+m)^{-2/(2+d)}.
$$
Combined with \eqref{equ:histogramboundapplied}, this gives 
\begin{align*}
\P(\hat{\eta}\geq \eta) &\leq \sum_{i=1}^{n+m} \P\left( |\hat{f}_\gamma'(Z_i)  -f_\gamma(Z_i) | \geq  \eta - 2^d (n+m)^{-2/(2+d)}\right)\leq 2 /(n+m),
\end{align*}
by choosing $\eta$ such that
$
 \eta \geq 2^d (n+m)^{-2/(2+d)} + c_0   (n+m)^{-1/(2+d)} \sqrt{\log (n+m)} .
$
We have established Lemma~\ref{lem:etachapdensityestimation}. 

\end{proof}

\begin{lemma}\label{lem:kerneldensity}[Histogram density estimator, non i.i.d. version]
Consider $Z_1,\dots,Z_{n}$ independent random variables take values in $[0,1]^d$ where $Z_i$ has for density $f_i$, $1\leq i\leq n$. We assume that all the $f_i$'s are $L$-Lipschitz and pointwise upper bounded by $C^d$ (for some constant value $C\in (0,1)$), where $L,C$ does not depend on $i$. Let $f(x)=n^{-1}\sum_{i=1}^n f_i(x)$ for $x\in [0,1]^d$.
We consider the histogram estimator of $f$ given by  \begin{equation}\label{equ:histogram}
\hat{f}_n(x)= M^d\sum_{j=1}^{M^d} n^{-1} \sum_{i=1}^n \ind{Z_i\in \mathcal{D}_j}  \ind{x\in \mathcal{D}_j},\:\: x\in [0,1]^d,
\end{equation}
where $\{\mathcal{D}_1,\dots,\mathcal{D}_{M^d}\}$ is a regular partition of $[0,1]^d$ formed by $M^d$ $d$-dimensional cubes of side size $1/M$ and Lebesgue measure $|\mathcal{D}_1|=1/M^d$.
Then, choosing $M=\lceil n^{1/(2+d)}\rceil$, for $n\geq N(C)$, we have for all $x\in [0,1]^d$,
\begin{equation}\label{equ:histogrambound}
\P\left( |\hat{f}_n(x)-f(x)|\geq  c_0(d,C,L)n^{-1/(2+d)} \sqrt{\log n} \right)\leq 2/n^2,
\end{equation}
where  $c_0(d,C,L)=L+4(2C)^{d/2}$.
\end{lemma}

\begin{proof}
We have for all $x\in [0,1]^d$,
$$
|\hat{f}_n(x)-f(x)|\leq |\hat{f}_n(x)-\E \hat{f}_n(x)| + |\E \hat{f}_n(x)-f(x)|,
$$
which contains a bias and a variance term. For the bias, we have
\begin{align*}
|\E \hat{f}_n(x)-f(x)| &\leq  \sum_{j=1}^{M^d} \ind{x\in \mathcal{D}_j} |\mathcal{D}_j|^{-1} \int_{\mathcal{D}_j} n^{-1} \sum_{i=1}^n |f_i(y)-f_i(x)| dy\leq L /M,
\end{align*}
because $\sup_{1\leq k\leq d}|x_k-y_k|\leq 1/M$ when $x,y$ belongs to the same $\mathcal{D}_j$.

For the variance term, by denoting $j_x$ the only $j$ such that $x\in \mathcal{D}_j$ and $p_{i,x}=\P(Z_i\in \mathcal{D}_{j_x})\in[ 0,(C/M)^d] $, we have
\begin{align*}
\P( |\hat{f}_n(x)-\E \hat{f}_n(x)|\geq \delta)&\leq \P\left(  \left|\sum_{i=1}^n (\ind{Z_i\in \mathcal{D}_{j_x}} - p_{i,x})\right| \geq \delta n M^{-d}\right)\\
&\leq 2\exp\left(-\frac{1}{2} \frac{A^2}{ \sum_{i=1}^n p_{i,x} + A/3 }\right)\leq 2\exp\left(-\frac{1}{2} \frac{A^2}{ n (C/M)^d + A/3 }\right),
\end{align*}
by letting $A=\delta n M^{-d}$ and applying Bernstein's inequality. Choosing $\delta$ such that $A\leq n (C/M)^d$, that is, $\delta\leq C^d$, we obtain
\begin{align*}
\P( |\hat{f}_n(x)-\E \hat{f}_n(x)|\geq \delta)&\leq2\exp\left(-(3/8) \frac{A^2}{ n (C/M)^d }\right) = 2\exp\left(-(3/8) n \delta^2 M^{-d} C^{-d}\right),
\end{align*}
by choosing $\delta$ such that $(3/8) n \delta^2 M^{-d} C^{-d}=2 \log n$, that is, $\delta= (4/\sqrt{3}) (MC)^{d/2 } \sqrt{(\log n)/n}$  gives 
\begin{align*}
\P( |\hat{f}_n(x)-\E \hat{f}_n(x)|\geq \delta)&\leq2/n^2,
\end{align*}
provided that $\delta\leq C^d$. Now, we choose $M=\lceil n^{1/(2+d)}\rceil$, so that for $n\geq N(C)$, $\delta= 4 (2C)^{d/2} n^{-1/(2+d)} \sqrt{\log n}$ is a valid choice. This gives the bound \eqref{equ:histogrambound}. 
\end{proof}

\subsection{From the two-sample setting to classical two-group setting}\label{sec:mr}

Existing results of the form \eqref{condetachap} typically assume that the observations are i.i.d. draws from a two-group mixture model under which the class labels are random. By contrast, we focus on a two-sample setting with fixed labels in which case the observations are {\it non-identically distributed} because the assumptions are weaker than the two-group setting. Nevertheless, we can easily adapt our theory in the two-group setting, in which case a plethora of existing results can be applied to understand the scale of $\zeta_r(\eta)$ and $\hat{\eta}$. The two-group setting can be formalized by the following assumptions.



\begin{assumption}\label{as:rm}
The sample $(Z_1,\dots,Z_{n+m})$ is obtained in the following way:
\begin{itemize}
\item $(Z_1,\dots,Z_k)=(W_i,1\leq i\leq n+m\::\:  A_i=0,B_i=0,C_i=1)$;
\item $(Z_{k+1},\dots,Z_{n})=(W_i,1\leq i\leq n+m\::\:  A_i=0,B_i=0,C_i=0)$;
\item $(X_i,i\in \cH_0)=(W_i,1\leq i\leq n+m\::\:  A_i=0,B_i=1,C_i=0)$;
 \item $(X_i,i\in \cH_1)=(W_i,1\leq i\leq n+m\::\:  A_i=1,B_i=1,C_i=0)$,
\end{itemize}
where $(A_i,B_i,C_i,W_i)$, $1\leq i\leq n+m$, are i.i.d. with $A_i\sim\mathcal{B}(\pi_1)$ (indicator of being a novelty), $B_i\sim\mathcal{B}(\pi_B)$ (indicator of being in the test sample),  $C_i\sim\mathcal{B}(\pi_C)$ (indicator of being in the first NTS), for some proportions $\pi_A,\pi_B,\pi_C\in (0,1)$ and 
$W_i\:|\: A_i \sim f_0$ if  $A_i=0$ and $W_i\:|\: A_i=1 \sim f_1$, where $f_0$ is the density of $P_0$ and $f_1$ is the common density of all $P_i$, $i\in \cH_1$.
\end{assumption}

Under Assumption~\ref{as:rm}, the sample size $n+m$ (number of trials) is fixed while the sample sizes $k$, $\l$, $m$ and $m_1$ are random. Also, we easily see that, conditionally on $A,B,C$, the sample $(Z_1,\dots,Z_{n+m})$ satisfies Assumptions~\ref{as:indep} and~\ref{equ-marg} with $f_i=f_1$ for $i\in \cH_1$.

As a result, under Assumption~\ref{as:rm}, and letting $L_i=1-(1-A_i)(1-B_i)C_i$ we have $(Z_1,\dots,Z_k)=(W_i, 1\leq i\leq n+m\::\:  L_i=0)$, and $(Z_{k+1},\dots,Z_{n+m})=(W_i, 1\leq i\leq n+m\::\:  L_i=1)$ with $(W_i,L_i)_{1\leq i\leq n+m}$ i.i.d., $L_i\sim \mathcal{B}(1-(1-\pi_A)(1-\pi_B)\pi_C)$ 
and $W_i\:|\: L_i=0 \sim f_0$ and $W_i\:|\: L_i=1 \sim (1-\pi) f_0 + \pi f_1$, for some $\pi\in(0,1)$.
Also, the knowledge of the samples $(Z_1,\dots,Z_k)$ and $(Z_{k+1},\dots,Z_{n+m})$ is equivalent to the knowledge of  $(W_i,L_i)_{1\leq i\leq n+m}$.
This means that, under Assumption~\ref{as:rm}, the score function \eqref{scorefunction} is based on $(W_i,L_i)_{1\leq i\leq n+m}$, which is a standard classification setting where the covariates and labels are jointly i.i.d..

\section{Material for experiments of  Section~6}

 \subsection{Details of datasets in Section \ref{sec:realdata}}\label{sec:descriptiondataset}

 \begin{itemize}
 \item Shuttle: consists of radiator data onboard space shuttles. There are 7 classes in total. Instances from class 1 are considered nominal, {while instances 
 those from classes 2, 3, 4, 5, 6, 7 are considered as novelties. }
 \item Credit Card: contains transactions made by credit cards over the course of two days, some being frauds.
 \item KDDCup99: describes a set of network connections that includes a variety of simulated intrusions. 
 \item Mammography: consists of features extracted from mammograms, some having microcalcifications.
 \item Musk: describes a set of molecules that are identified as either musk or non-musk. Musk are considered nominals, non-musk are novelties.
 \item MNIST: contains a set of labeled images of size $28\times28$ of handwritten digits `0' to `9'. We restrict the analysis to `4' and `9', considering that `4' are nulls, `9' are alternatives. 
 \end{itemize}

\subsection{Methods}\label{sec:simumethod}

We describe here the score functions used in each version of AdaDetect mentioned  in Section \ref{sec:realdata}. We also include  the oracle score defined in \eqref{equoraclescore} and the SC procedures proposed by \cite{SC2007}.


\begin{itemize}
	\item \texttt{AdaDetect oracle}: the oracle score function $\lrt$ defined in \eqref{equLR}.
	\item \texttt{AdaDetect parametric} and \texttt{AdaDetect KDE}: ${g}(x)=\wh{f}_\gamma(x)/\wh{f}_0(x)$, where $\wh{f}_\gamma$ is a density estimator of $f_\gamma$ \eqref{equfgamma}  computed on mixed sample $(Z_{k+1},\dots,Z_{n+m})$ and $\wh{f}_0$ is a density estimator of $f_0$ based on $(Z_1,\dots,Z_k)$. For \texttt{AdaDetect parametric}, $\hat f_0$ is estimated by the Ledoit-Wolf method and $\hat f_\gamma$ is estimated by a two-component mixture of Gaussians via an expectation-maximization (EM) algorithm with 100 random restarts. For \texttt{AdaDetect KDE}, they are given by non-parametric Gaussian kernel density estimators (KDE). 
	\item \texttt{AdaDetect SVM}: $\wh{g}$ is obtained by minimizing the empirical risk \eqref{equgPU} with the hinge loss, $\lambda=1$ \rev{ and a suitable regularization with the cost parameter $C$ set to $1$; see \cite{hastie2009elements}.}
	\item \texttt{AdaDetect RF}: $\wh{g}$ is obtained by random forest with the maximum depth $10$; 
	\item  \texttt{AdaDetect NN}: $\wh{g}$ is obtained by minimizing the cross entropy loss \eqref{equgPU} with $\lambda=1$ and the NN function class with $1$ hidden layer, $100$ neurons, and the ReLU activation function. 
	 \item  \texttt{AdaDetect NN cv}: AdaDetect NN with the number of hidden layers and the number of neurons per layer chosen by the cross-validation procedure described in Section~\ref{sec:cross}. 
\end{itemize}
They are compared to several existing methods:
\begin{itemize}
\item \texttt{SC parametric} and \texttt{SC KDE}: the procedure of \cite{SC2007} based on local FDR estimates $\ell_i= \hat \pi_0 \hat{f_0}(X_i) / \hat{f}(X_i)$, where the densities are estimated by the same methods for \texttt{AdaDetect parametric} and \texttt{SC KDE}, respectively, except that $\hat{f_0}$ is based on the whole NTS $Y=(Y_1,\dots,Y_n)$ and $\hat{f}$ is only based on the test sample $X=(X_1,\dots,X_m)$. We set $\hat \pi_0=1$ for a fair comparison with the non-adaptive versions of AdaDetect;
\item \texttt{CAD SVM} and \texttt{CAD IForest}: the conformal
  anomaly detection procedure proposed by \cite{bates2021testing} based on one-class SVM and Isolation Forest, respectively. 
\end{itemize}

Note that these procedures all provably control the FDR under Assumption~\ref{as:indep} (see Corollary~\ref{FDRbounds}) except for \texttt{SC parametric} and \texttt{SC KDE}, which only control the FDR asymptotically with a consistent estimator of the density ratio \citep{SC2007}.

\subsection{Additional experiments with varying $n$, $m$, and $m_1$}\label{sec:addnumexp}

\rev{In this section, we report results for additional experiments in more challenging settings for AdaDetect. 
\begin{itemize}
\item Small sample sizes: $m=200$, with $k=4m$ and $\ell=m$ as before. The results are reported in Table~\ref{tab:perfsmallsamplesize}.
\item Small null sample sizes with $n < m$: $n=m/2$, with $k=m/4$ and $\ell=m/4$. The results are reported in Table~\ref{tab:perfsmallnullsamplesize}. 
\item Sparse novelties: we set $m_1/m=2\%$, with $m=1000$, $k=4m$ and $\ell=m$ as before. The results are reported in Table~\ref{tab:perfhighsparsity}.
\end{itemize}
As in Section \ref{sec:realdata}, we set the target FDR level $\alpha = 0.1$ and report the mean value and the standard deviation (in brackets) over 100 runs. We only highlight in bold the best-performing method if its FDR is above $\alpha$.
}

\begin{table}
\caption{Same as Table~\ref{tab:datasetsperf} for $m=200$ ($k=4m$, $\ell=m$). 
\label{tab:perfsmallsamplesize}}
{\tiny\begin{tabular}{lcccccc}
& Shuttle & Credit card & KDDCup99 &  Mammography & Musk & MNIST \\
\hline
\hline
& \multicolumn{6}{c}{FDR} \\
\hline
\hline
CAD SVM &  0.04 (0.09) &  0.00 (0.00) & 0.00 (0.00) &  0.02 (0.09) & 0.00 (0.00) & 0.00 (0.00)\\
CAD IForest &  0.07 (0.10) & 0.05 (0.08) & 0.07 (0.11) &  0.01 (0.06) & 0.00 (0.00)  & 0.00 (0.00) \\
AdaDetect parametric & 0.03 (0.09) & 0.00 (0.00) & 0.00 (0.00) &  0.01 (0.08) & 0.00 (0.00)  & 0.00 (0.00) \\
AdaDetect KDE &  0.05 (0.10) & 0.01 (0.05) & 0.00 (0.00) &  0.02 (0.06) & 0.00 (0.00)  & 0.00 (0.00)\\
AdaDetect SVM &  0.04 (0.10) & 0.02 (0.07) & 0.01 (0.04) &  0.01 (0.06) & 0.00 (0.00)  & 0.01 (0.03) \\
AdaDetect RF & 0.07 (0.09)  & 0.07 (0.09) & 0.08 (0.09) &  0.05 (0.11) & 0.00 (0.00)  & 0.02 (0.12) \\
AdaDetect NN &  0.05 (0.09) & 0.06 (0.09) & 0.06 (0.14) &  0.03 (0.09) & 0.01 (0.10)  & 0.02 (0.10) \\
AdaDetect cv NN &  0.05 (0.08) & 0.06 (0.08) & 0.05 (0.11) &  0.04 (0.10) & 0.01 (0.10)  & 0.00 (0.00)\\
\hline
CAD SVDD + CNN &  - & - & - &  - & - & 0.01 (0.07)  \\
AdaDetect CNN &  - & - & - &  - & -  & 0.01 (0.10) \\
\hline
\hline
& \multicolumn{6}{c}{TDR} \\
\hline
\hline
CAD SVM &  0.12 (0.22) & 0.00 (0.00) & 0.00 (0.00) & 0.03 (0.10) & 0.00 (0.00)  & 0.00 (0.00)\\
CAD IForest &  0.28 (0.28) & 0.25 (0.32) & 0.42 (0.47) &  0.02 (0.11) & 0.00 (0.00)  & 0.00 (0.00)\\
AdaDetect parametric &  0.13 (0.25) & 0.00 (0.00) & 0.00 (0.00) &  0.01 (0.07) & 0.00 (0.00)  &0.00 (0.00)\\
AdaDetect KDE &  0.22 (0.33) & 0.01 (0.07) & 0.00 (0.00) & 0.06 (0.n17) & 0.00 (0.00)  &0.00 (0.00)\\
AdaDetect SVM &  0.16 (0.26) & 0.03 (0.11) & 0.06 (0.21) &  0.01 (0.05) & 0.00 (0.00)  & 0.02 (0.11) \\
AdaDetect RF &  \textbf{0.57 (0.30)} & \textbf{0.71 (0.21)} & \textbf{0.97 (0.04)} &  \textbf{0.11 (0.21)}  & 0.00 (0.00)  & 0.00 (0.01)\\
AdaDetect NN &  0.31 (0.35) & 0.46 (0.34) & 0.28 (0.40) &  0.07 (0.19) & 0.00 (0.00)  & 0.01(0.05) \\
AdaDetect cv NN &  0.28 (0.36) & 0.47 (0.33) & 0.31 (0.39) &  0.10 (0.22) & 0.00 (0.00)  &0.00 (0.00)\\
\hline
CAD SVDD + CNN &  - & - & - &  - & - & 0.00 (0.02)  \\
AdaDetect CNN &  - & - & - &  - & -  & 0.02 (0.10) \\
\hline
\end{tabular}
}
\end{table} 

\begin{table}
\caption{Same as Table~\ref{tab:datasetsperf} for $n=m/2$ ($k=m/4$, $\ell=m/4$). 
\label{tab:perfsmallnullsamplesize}}  
{\tiny\begin{tabular}{lcccccc}
& Shuttle & Credit card & KDDCup99 &  Mammography & Musk & MNIST \\
\hline
\hline
& \multicolumn{6}{c}{FDR} \\
\hline
\hline
CAD SVM & 0.03 (0.06) & 0.00 (0.00) &  0.00 (0.00) &  0.02 (0.08) & 0.00 (0.00) &0.00 (0.00) \\
CAD IForest & 0.05 (0.07) & 0.04 (0.07) & 0.04	(0.07) &  0.01 (0.06) & 0.00 (0.00) & 0.00 (0.00) \\
AdaDetect parametric &  0.03 (0.05) & 0.00 (0.00) & 0.00 (0.00)& 0.23 (0.39) & 0.00 (0.00) & 0.00 (0.00)\\
AdaDetect KDE &  0.05 (0.08) & 0.01 (0.07) &  0.00 (0.00) &  0.04 (0.09) & 0.00 (0.00) & 0.00 (0.00)\\
AdaDetect SVM &  0.06 (0.07) & 0.03 (0.06)& 0.00 (0.04) & 0.00 (0.04) & 0.00 (0.00) & 0.00 (0.00)\\
AdaDetect RF &  0.07 (0.06) & 0.07	(0.06) & 0.08 (0.06) &  0.05 (0.09) & 0.00 (0.00) &  0.00 (0.00)\\
AdaDetect NN & 0.05 (0.06) & 0.06 (0.07) & 0.01 (0.05) &  0.05 (0.08) & 0.01 (0.10) &  0.00 (0.00)\\
AdaDetect cv NN & 0.05 (0.07) & 0.04 (0.08) & 0.02 (0.06) & 0.04 (0.08) & 0.02 (0.14)& 0.00 (0.00)\\
\hline
CAD SVDD + CNN &  - & - & - &  - & - &  0.03 (0.03)  \\
AdaDetect CNN &  - & - & - &  - & - & 0.03 (0.09)  \\
\hline
\hline
& \multicolumn{6}{c}{TDR} \\
\hline
\hline
CAD SVM &  0.07 (0.16) &  0.00 (0.00) &  0.00 (0.00) & 0.02 (0.08) & 0.00 (0.00) &0.00 (0.00) \\
CAD IForest &  0.22 (0.24) & 0.20 (0.29) & 0.34 (0.46) &  0.01 (0.07) & 0.00 (0.00) &  0.00 (0.00)\\
AdaDetect parametric & 0.12 (0.23) & 0.00 (0.00) &  0.00 (0.00) &  0.05 (0.09) & 0.00 (0.00) & 0.00 (0.00)\\
AdaDetect KDE & 0.19 (0.30) & 0.01	 (0.06) &  0.00 (0.00) &  0.07 (0.16) & 0.00 (0.00) & 0.00 (0.00)\\
AdaDetect SVM &  0.47 (0.29) & 0.20 (0.32) & 0.00 (0.02) &  0.00	(0.03) & 0.00 (0.00)  & 0.00 (0.00)\\
AdaDetect RF &  \textbf{0.66 (0.20)} & \textbf{0.72 (0.16)} & \textbf{0.98 (0.02)} &  0.11 (0.18) & 0.00 (0.00) & 0.00 (0.00) \\
AdaDetect NN &  0.30 (0.36) & 0.39	(0.35) & 0.07 (0.24) &  \textbf{0.16 (0.22)} & 0.00 (0.00) & 0.00 (0.00) \\
AdaDetect cv NN &  0.30 (0.36) & 0.21 (0.31) & 0.09	 (0.27)&  0.13 (0.21) & 0.00 (0.00)  & 0.00 (0.00)\\
\hline
CAD SVDD + CNN &  - & - & - &  - & - & 0.03 (0.03)  \\
AdaDetect CNN &  - & - & - &  - & -  & 0.06 (0.17) \\
\hline
\end{tabular}
}
\end{table}

\begin{table}
\caption{Same as Table~\ref{tab:datasetsperf} for $m_1/m=2\%$ ($k=4m$, $\ell=m$).} 
\label{tab:perfhighsparsity}
{\tiny\begin{tabular}{lccccccc}
& Shuttle & Credit card & KDDCup99 &  Mammography & Musk & MNIST  \\
\hline
\hline
& \multicolumn{6}{c}{FDR} \\
\hline
\hline
CAD SVM & 0.00 (0.00) & 0.00 (0.00) &  0.00 (0.00) &  0.00 (0.00) & 0.00 (0.00) &  0.00 (0.00)\\
CAD IForest & 0.05 (0.11) & 0.03 (0.11) & 0.03	(0.08) &  0.00 (0.00) & 0.00 (0.00) & 0.00 (0.00)  \\
AdaDetect parametric &  0.00 (0.00) & 0.00 (0.00) & 0.00 (0.00)& 0.00 (0.00) & 0.00 (0.00) &0.00 (0.00) \\
AdaDetect KDE &  0.00 (0.00) & 0.00 (0.00) &  0.00 (0.00) &  0.00 (0.00) & 0.00 (0.00) & 0.00 (0.00) \\
AdaDetect SVM &  0.03 (0.09) & 0.03 (0.08)& 0.00 (0.00) & 0.00 (0.00) & 0.00 (0.00) & 0.00 (0.00)\\
AdaDetect RF & 0.09 (0.10)  & 0.07	(0.09) & 0.08 (0.08) &  0.03 (0.11) & 0.00 (0.00) & 0.00 (0.00) \\
AdaDetect NN & 0.06 (0.12) & 0.06 (0.10) & 0.03 (0.07) &  0.03	 (0.09) &  0.00 (0.00) & 0.00 (0.00)\\
AdaDetect cv NN & 0.04 (0.10) & 0.05 (0.09) & 0.02	(0.05) & 0.02 (0.09) & 0.00 (0.00) & 0.00 (0.00)\\
\hline
CAD SVDD + CNN &  - & - & - &  - & - & 0.01 (0.10)  \\
AdaDetect CNN &  - & - & - &  - & - & 0.02 (0.14)  \\ 
\hline
\hline
& \multicolumn{6}{c}{TDR} \\
\hline
\hline
CAD SVM &  0.00 (0.00) &  0.00 (0.00) &  0.00 (0.00) & 0.00 (0.00) & 0.00 (0.00) &0.00 (0.00) \\
CAD IForest &  0.10 (0.20) & 0.04 (0.13) & 0.13 (0.31) &  0.00 (0.00) & 0.00 (0.00) &  0.00 (0.00) \\
AdaDetect parametric & 0.00 (0.00) & 0.00 (0.00) &  0.00 (0.00) &  0.00 (0.00) & 0.00 (0.00) &0.00 (0.00) \\
AdaDetect KDE & 0.00 (0.00) & 0.00 (0.00) &  0.00 (0.00) &  0.00 (0.00) & 0.00 (0.00) &0.00 (0.00) \\
AdaDetect SVM &  0.10 (0.27) & 0.11 (0.26) & 0.00 (0.00) &  0.00 (0.00) & 0.00 (0.00)  & 0.00 (0.00)\\
AdaDetect RF &  \textbf{0.60 (0.29)} & \textbf{0.64 (0.27)} & \textbf{0.98 (0.03)} &  0.05 (0.14) &  0.00 (0.00) &0.00 (0.00) \\
AdaDetect NN &  0.12 (0.24) & 0.38	(0.34) & 0.22 (0.38) &  0.04 (0.13) &  0.00 (0.00) &  0.00 (0.00)\\
AdaDetect cv NN &  0.08 (0.21) & 0.37 (0.35) & 0.18	 (0.36) &  0.03 (0.11) &  0.00 (0.00) & 0.00 (0.00)\\
\hline
CAD SVDD + CNN &  - & - & - &  - & - & 0.01 (0.09)   \\
AdaDetect CNN &  - & - & - &  - & -  & 0.01 (0.07) \\
\hline
\end{tabular}
}
\end{table}

\subsection{Simulated data}\label{sec:simdata}

We consider here experiments based on simulated data sets.  
In all experiments considered in this section, we generate measurements under Assumption~\ref{as:indep} with all novelties generated from the same distribution $P_1$. Unless otherwise specified, we set $n = 3000, m=1000$, $\pi_0=m_0/m=0.9$, and calculate FDR and TDR based on $100$ Monte-Carlo simulations. Following Remark \ref{rem:samplesize}, we set $k = 2m$ and $\ell = m$ for all AdaDetect methods. 

\paragraph{Gaussian setting}

We start with a setting where $P_0 = \mathcal{N}(0, I_d)$ and $P_1=\mathcal{N}(\mu,I_d)$, where $\mu \in \R^d$ is a sparse vector with the first $5$ coordinates equal to $\sqrt{2 \log(d)}$ and the remaining ones equal to $0$. The results are presented in Figure \ref{fig:gauss} with the dimension $d$ varying. First, we note that neither \texttt{SC~parametric} nor \texttt{SC~KDE} control the FDR, even if the model is correctly specified for the former, and the FDR inflation is substantial in high dimensions. By contrast, as implied by our theory, all other procedures control the FDR at level $\pi_0\alpha$. Next, we compare the TDR of procedures which control the FDR. In low dimensions ($d \leq 100$), \texttt{AdaDetect parametric} has the highest power that is close to \texttt{AdaDetect oracle}, which is expected since the model is correctly specified and parametric estimation is accurate when the dimension is low. In high dimensions ($d = 500$), \texttt{AdaDetect parametric} becomes much more noisy while \texttt{AdaDetect RF} maintains a stable and high power.

\begin{figure}
		\begin{minipage}{0.85\linewidth}
		\raisebox{-0.5 mm}{ \includegraphics[width=0.215\linewidth]{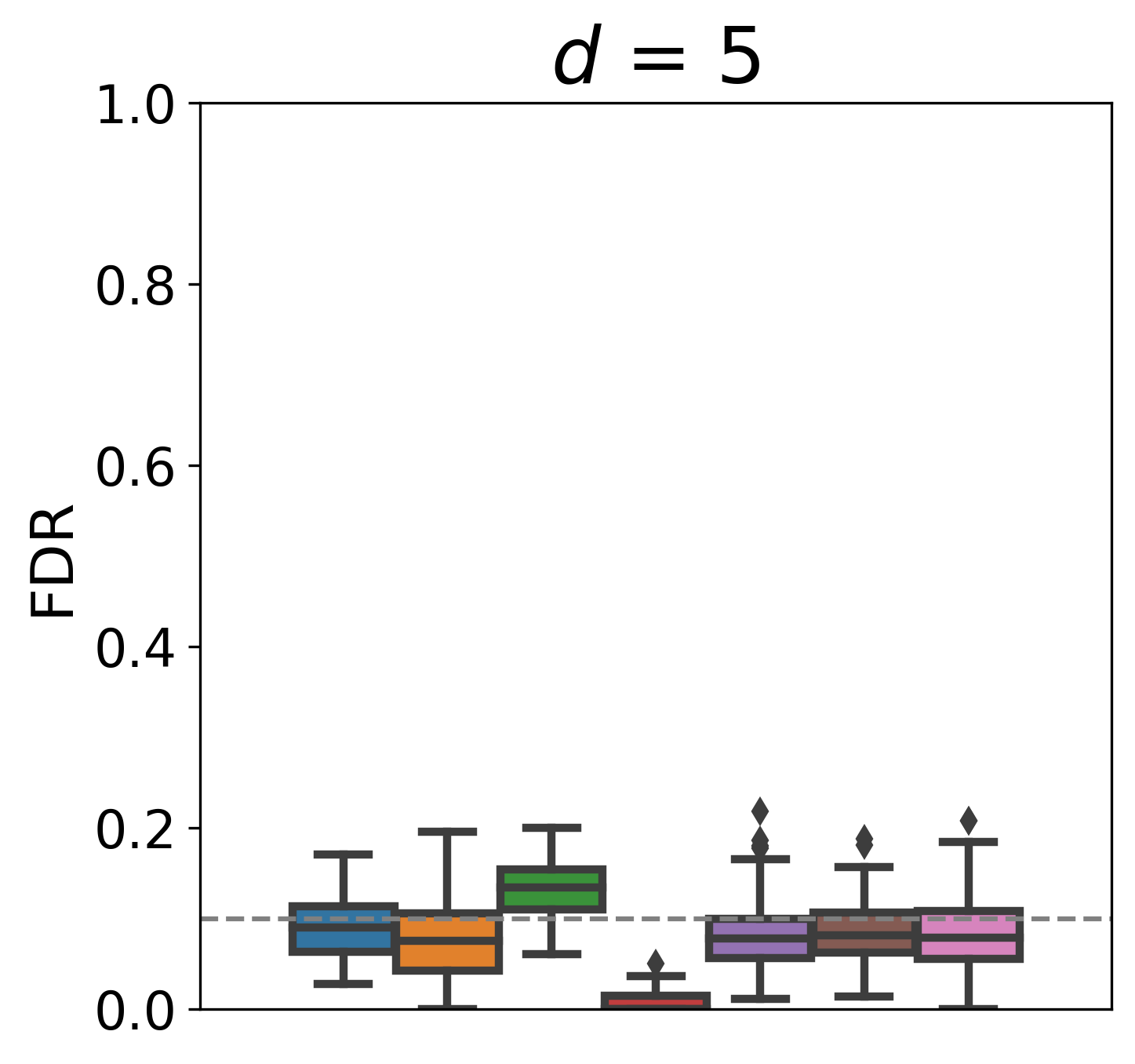} }
		\includegraphics[width=0.18\linewidth]{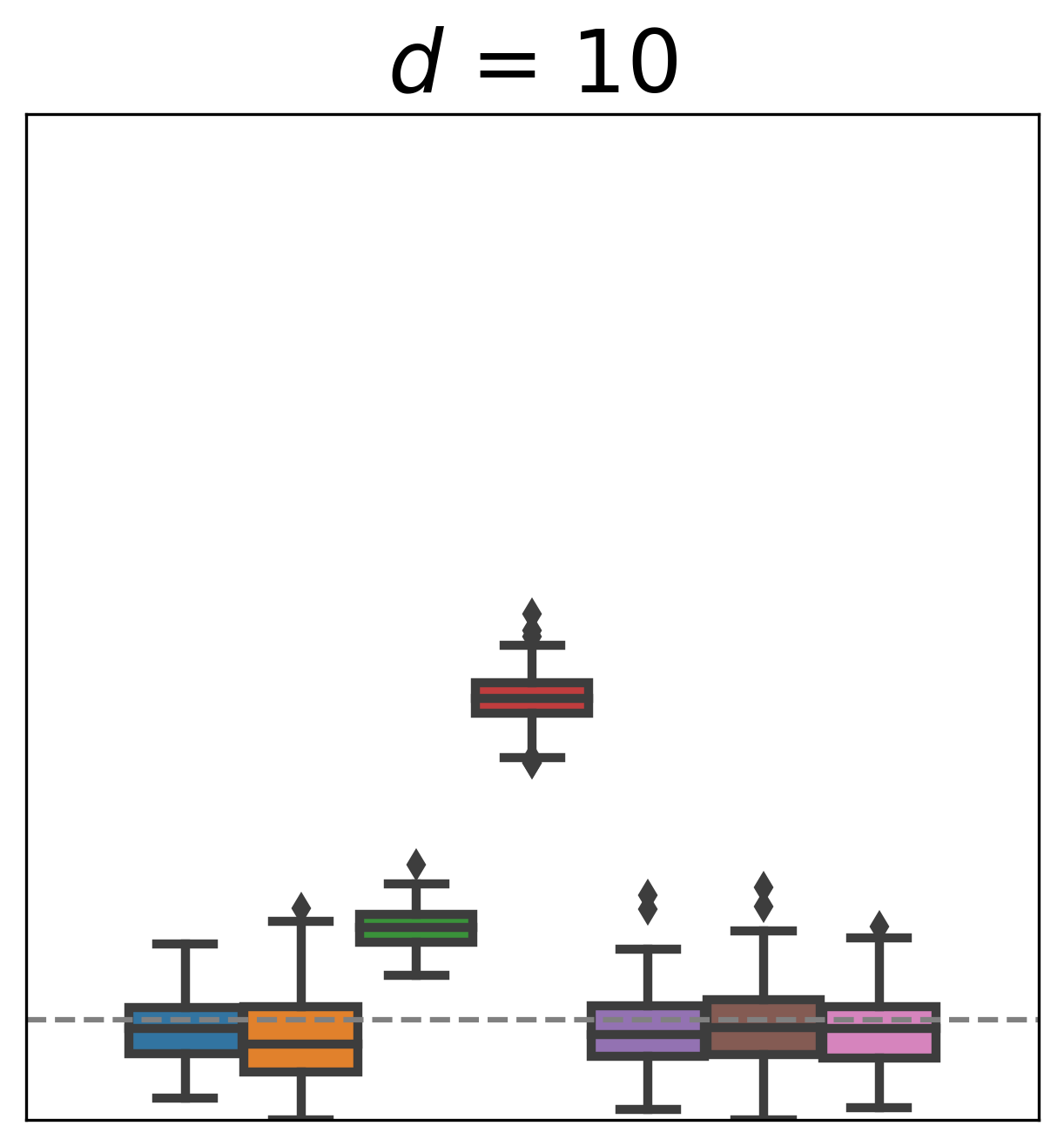} 
		\includegraphics[width=0.18\linewidth]{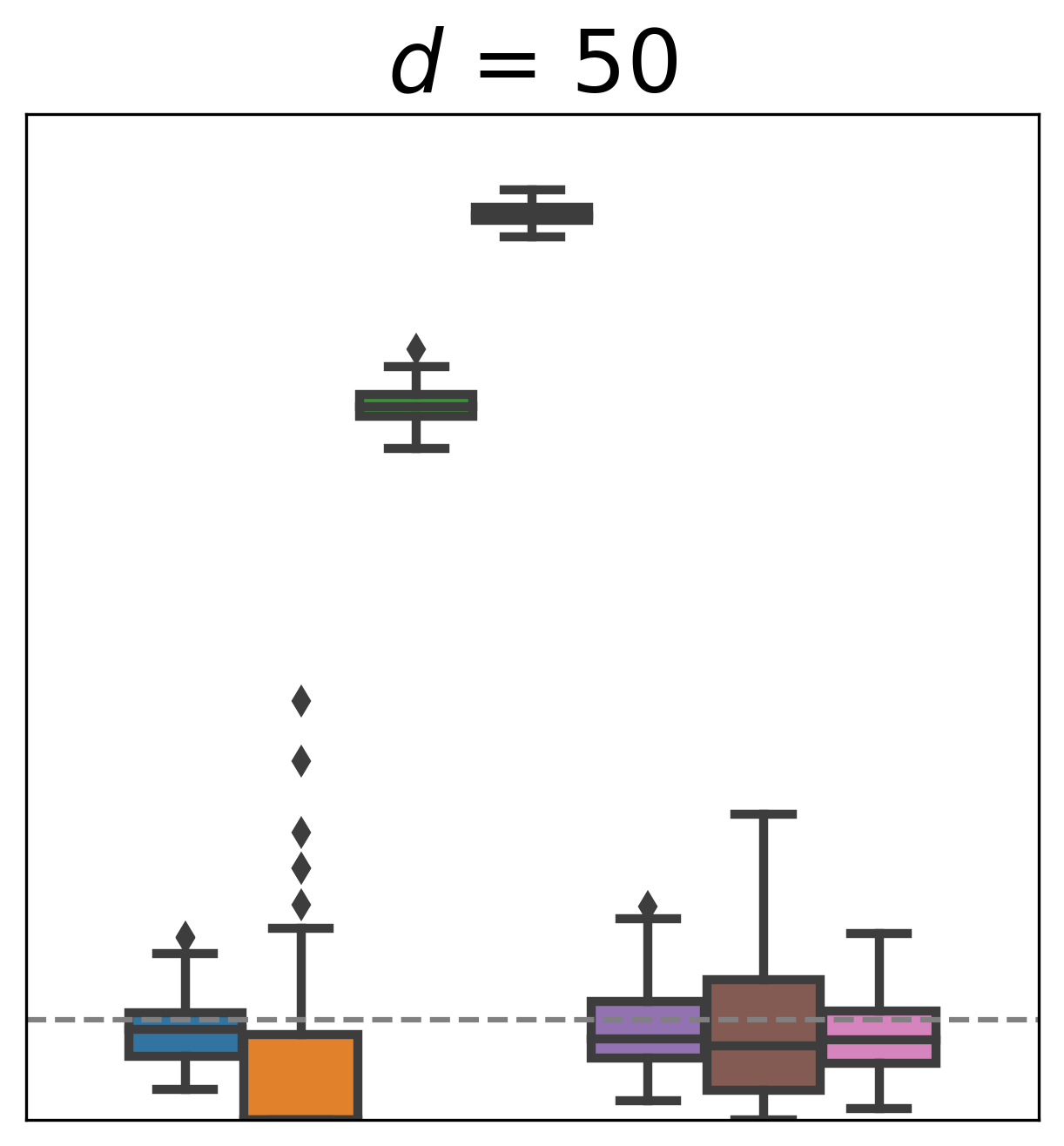} 
		\includegraphics[width=0.18\linewidth]{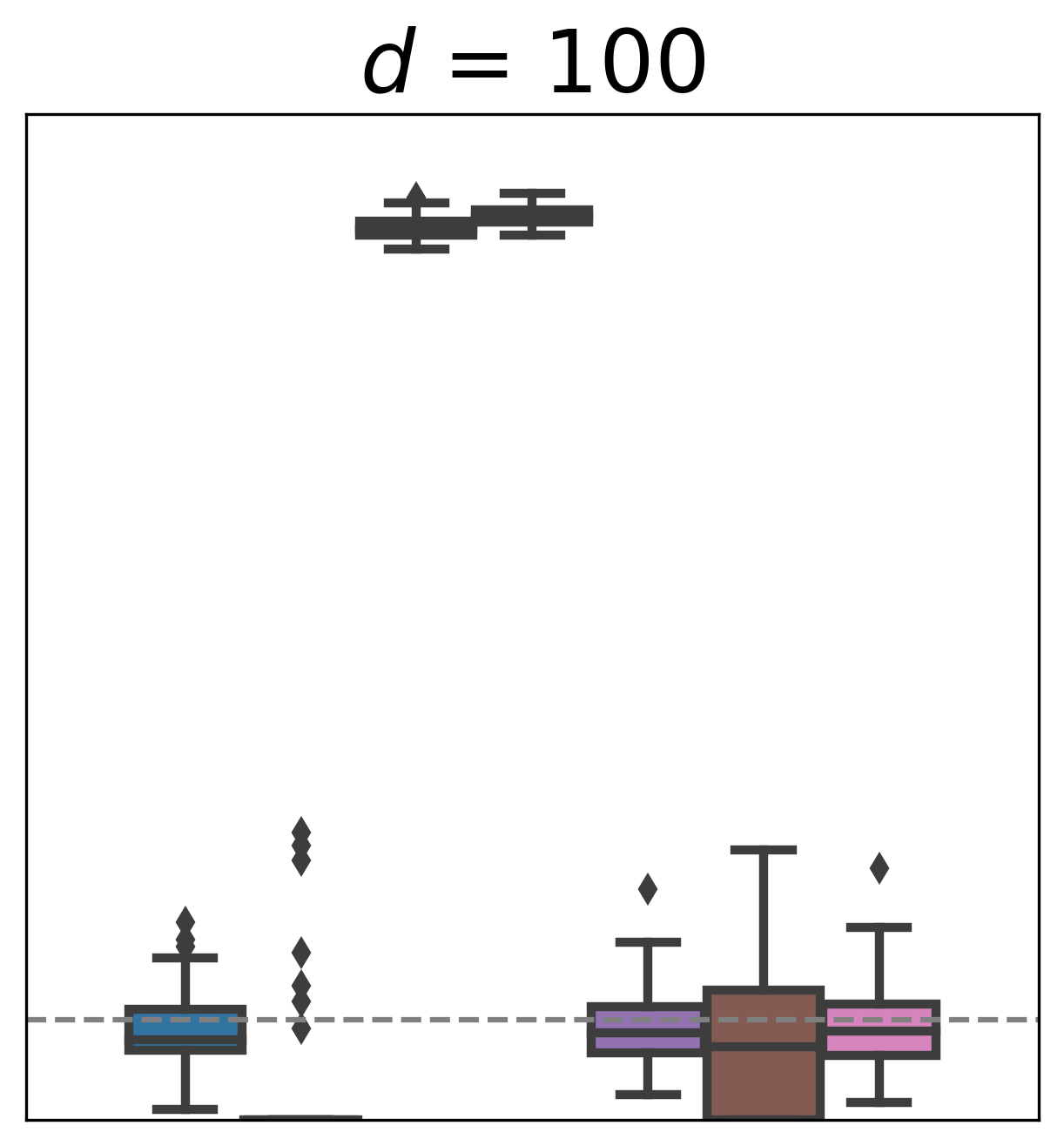} 
		\includegraphics[width=0.18\linewidth]{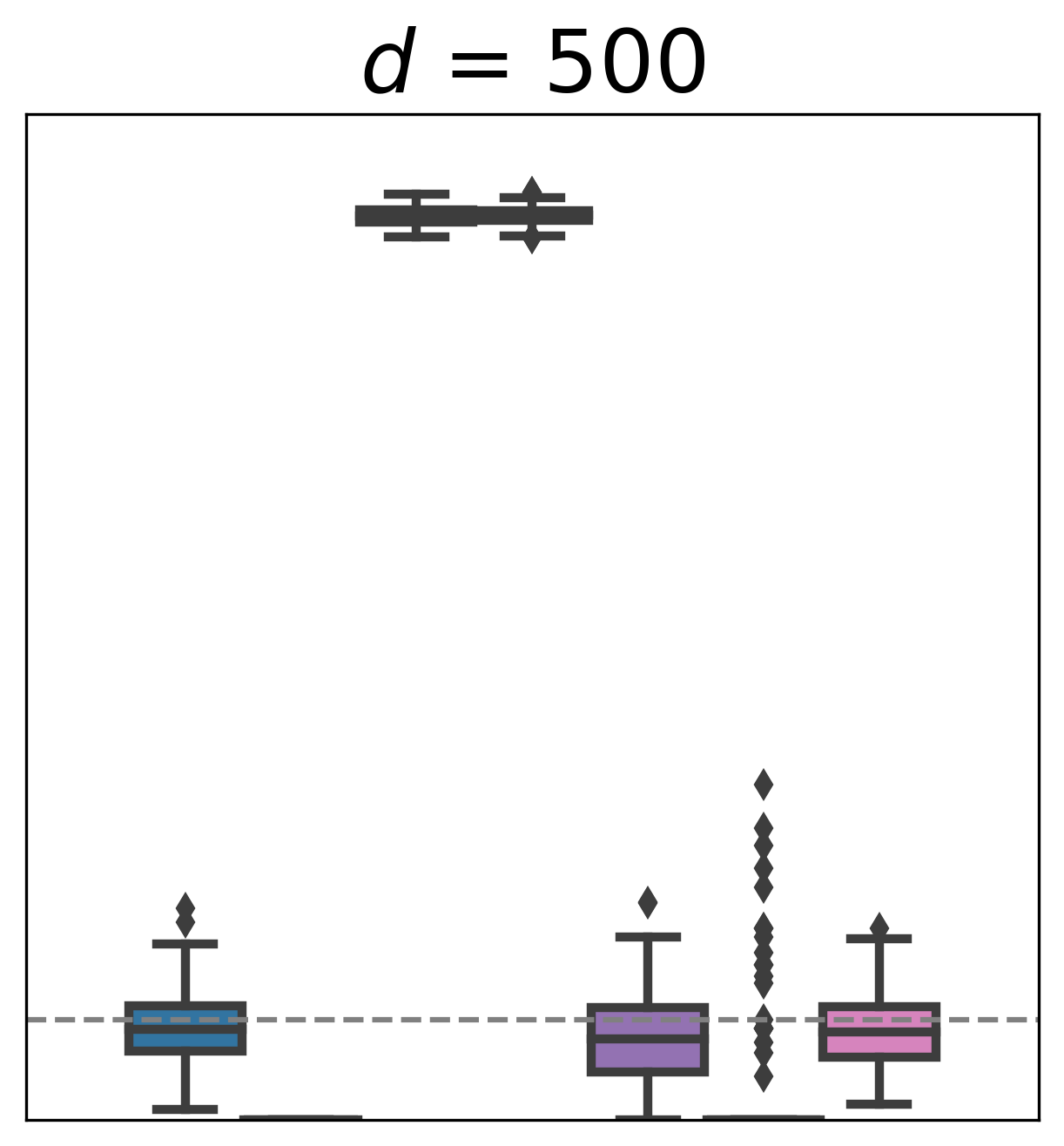} 
		\end{minipage}%
		\begin{minipage}[c]{0.15\linewidth}
		\includegraphics[width=\linewidth]{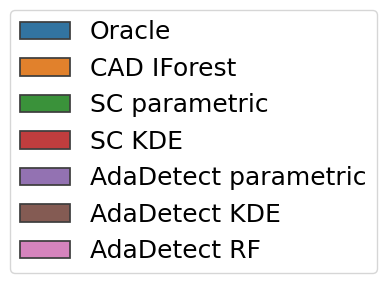} 
		\end{minipage}
		\raggedright
		\begin{minipage}{0.85\linewidth}
		\raisebox{-0.5 mm}{ \includegraphics[width=0.215\linewidth]{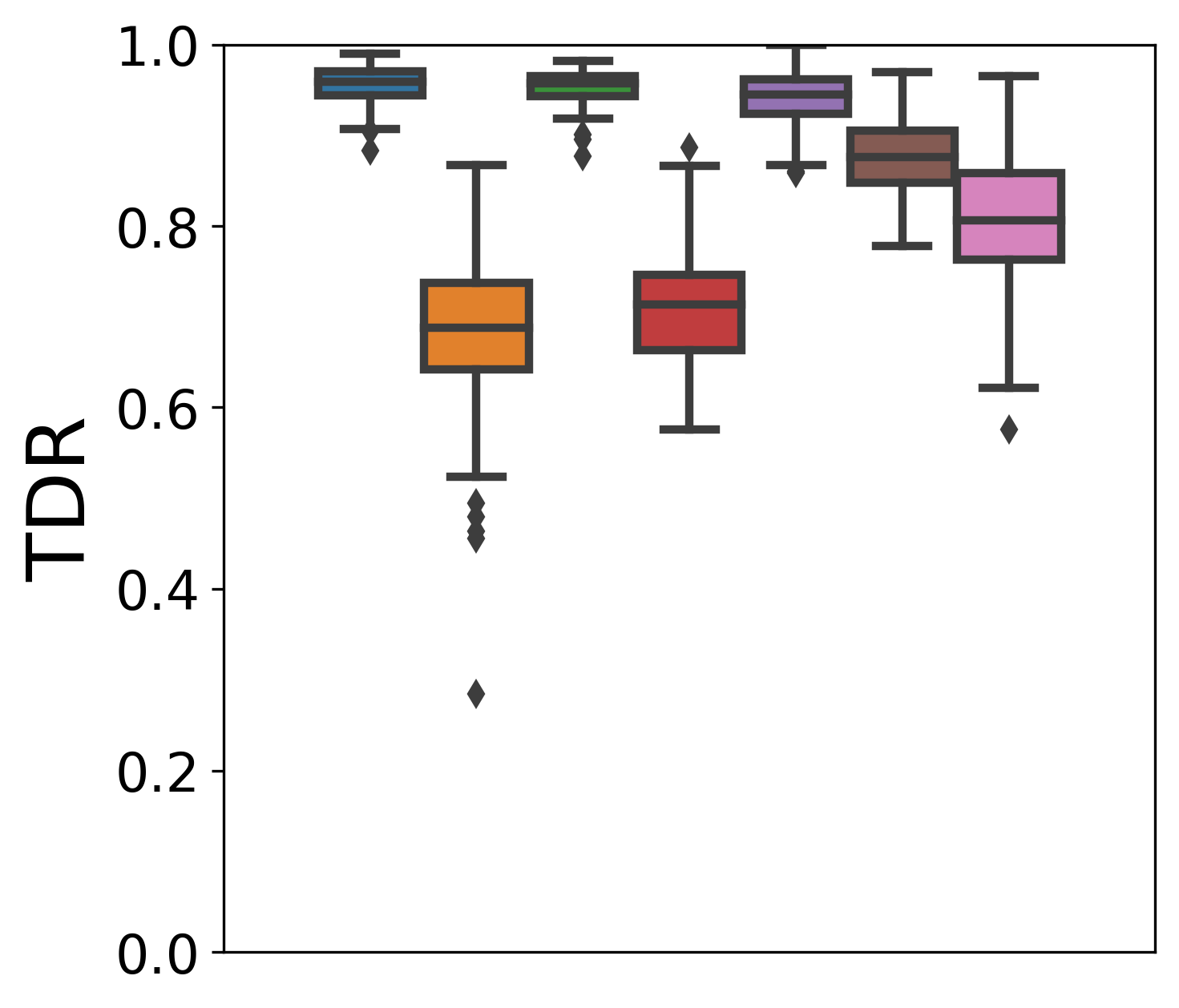} }
		\includegraphics[width=0.18\linewidth]{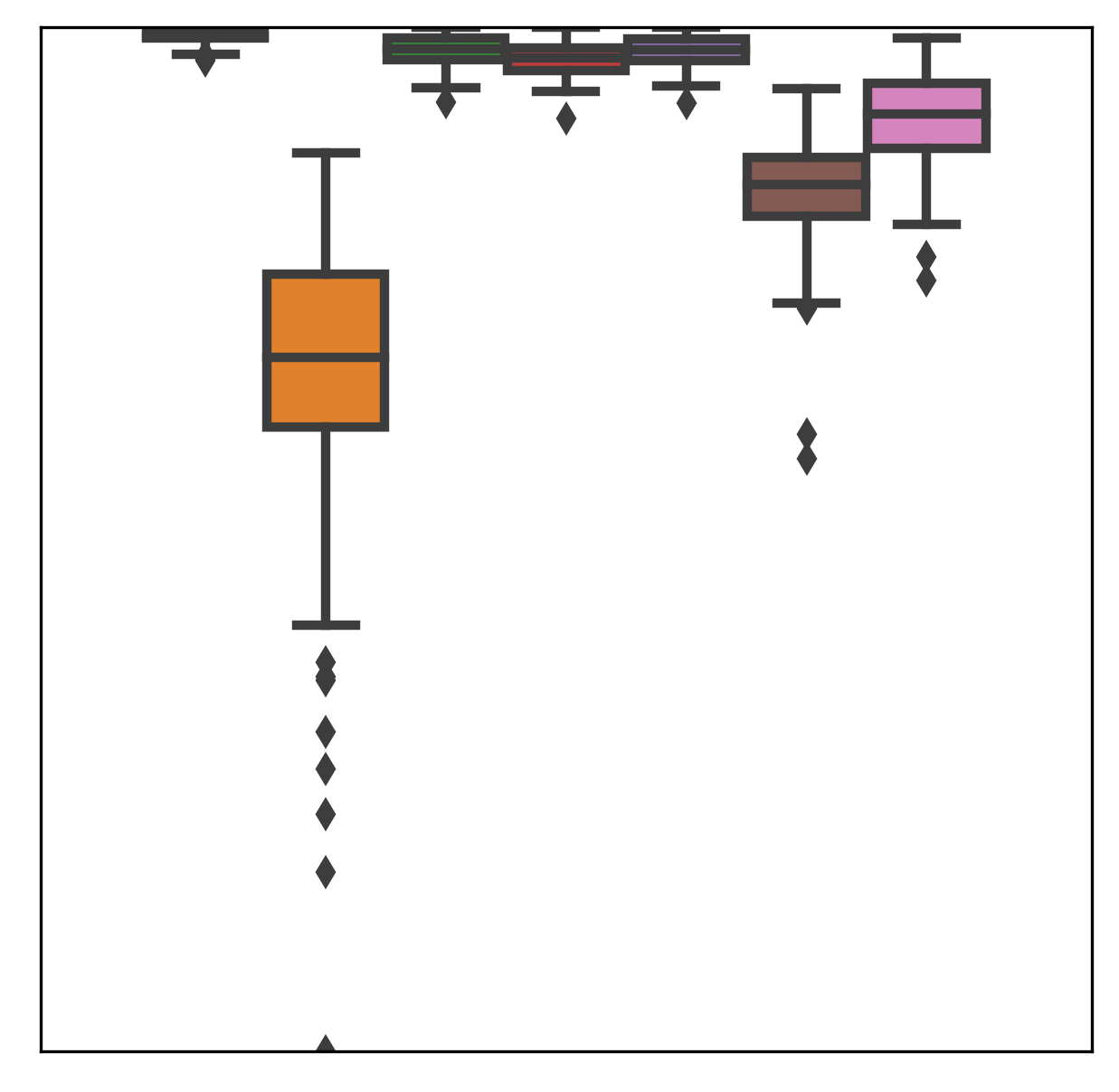} 
		\includegraphics[width=0.18\linewidth]{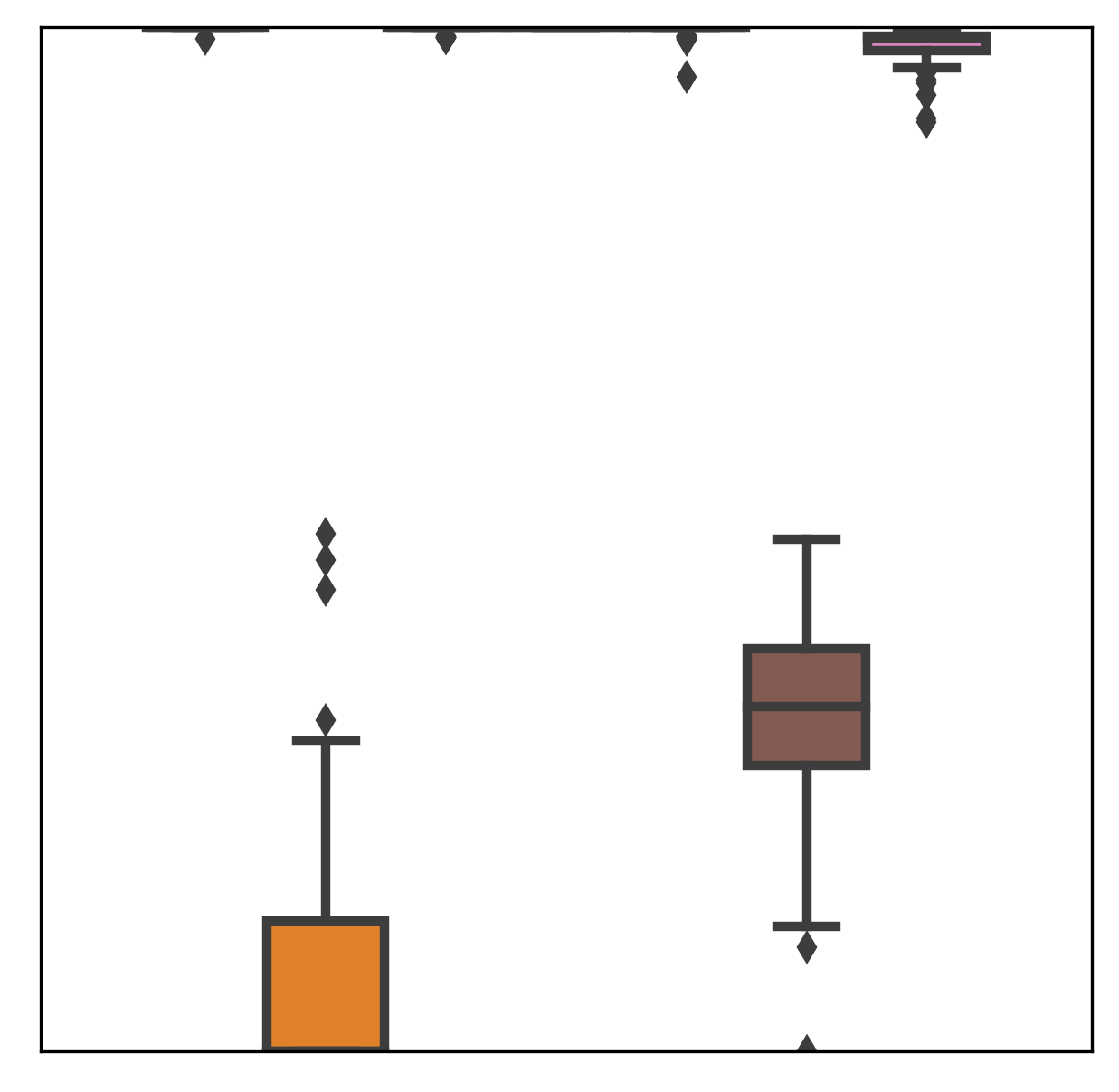} 
		\includegraphics[width=0.18\linewidth]{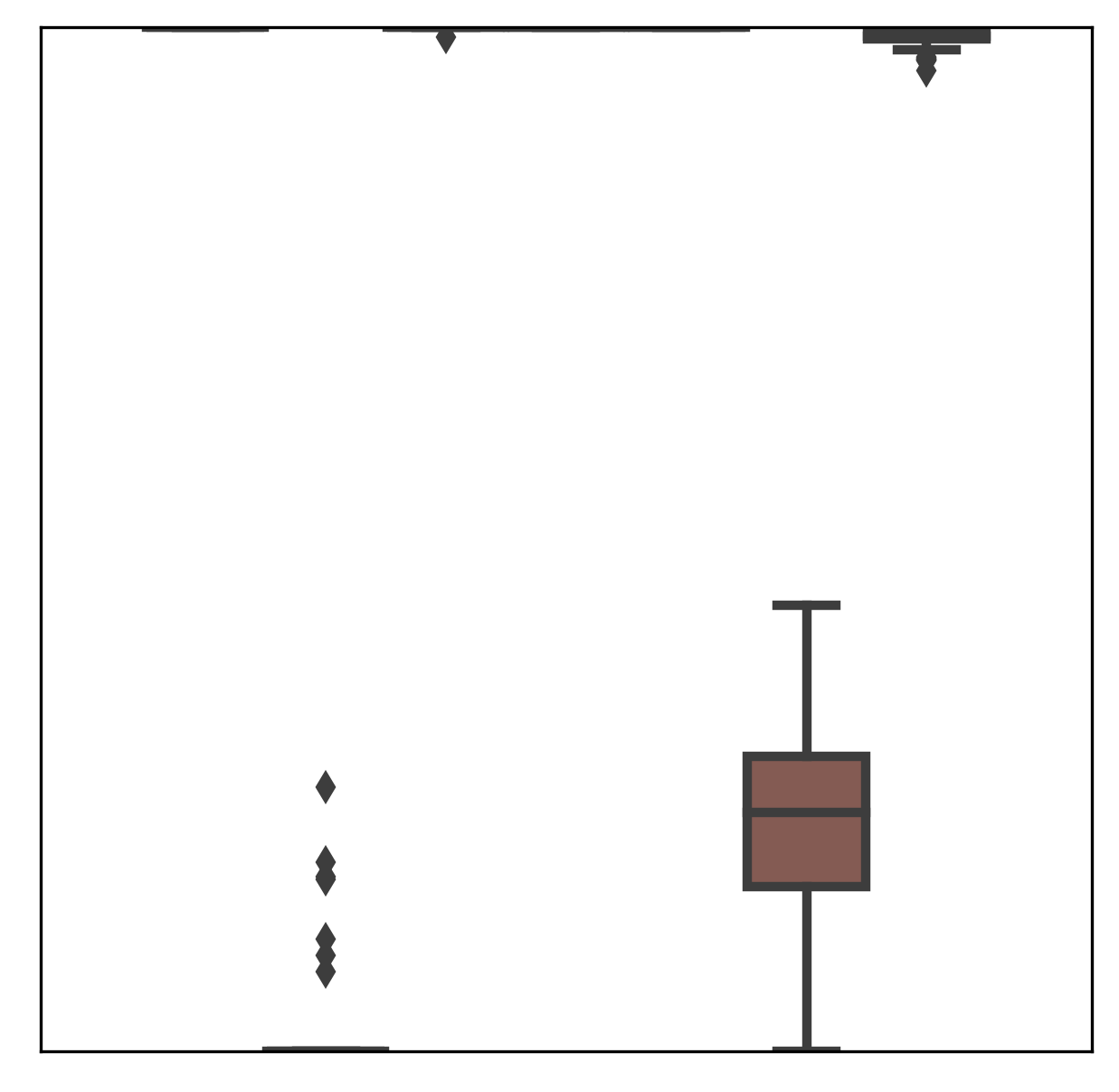} 
		\includegraphics[width=0.18\linewidth]{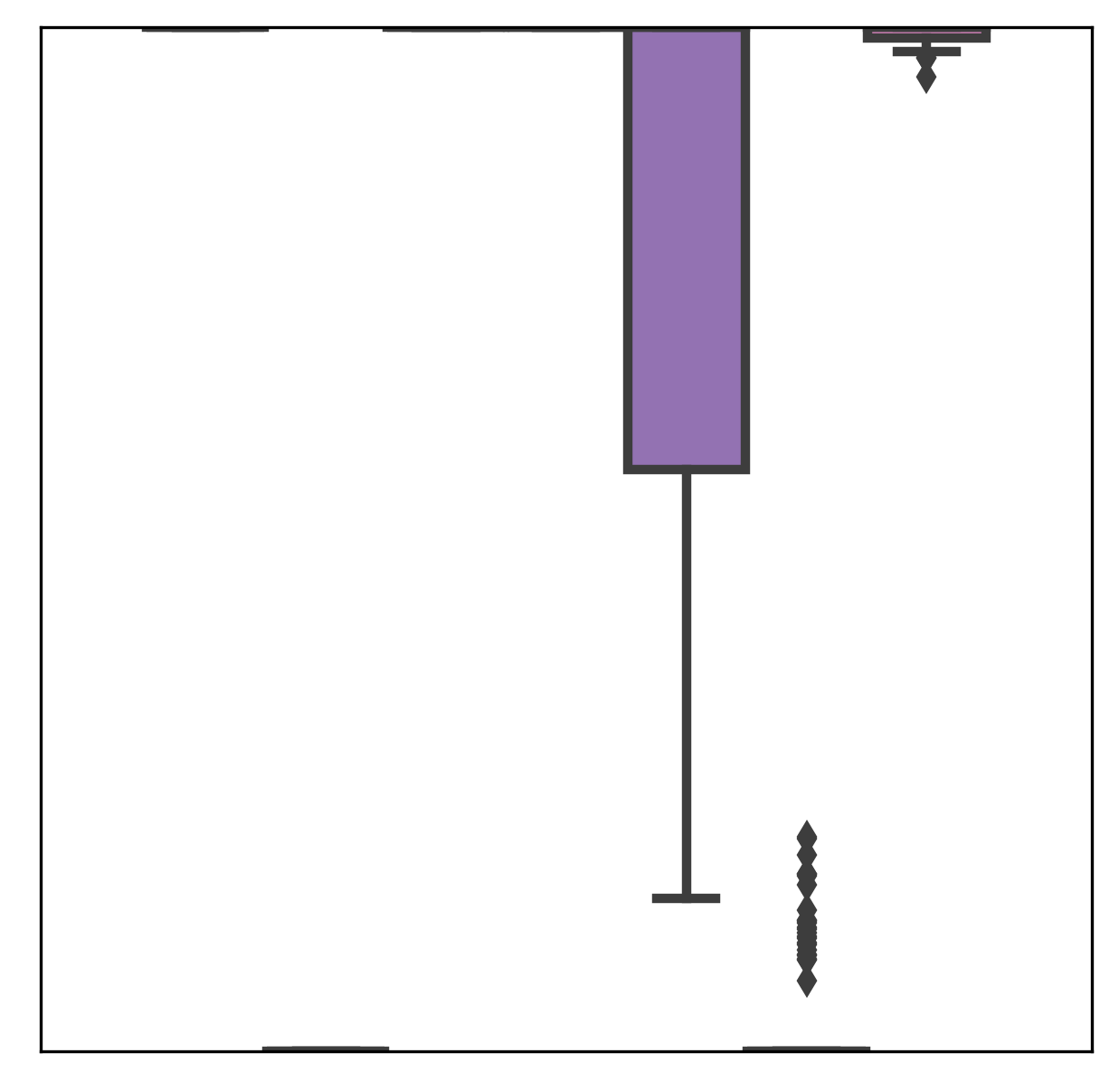} 
		\end{minipage}

	\caption{Gaussian setting. FDR (top) and TDR (bottom) as a function of $d$.  The dashed line indicates the target level. 
	\label{fig:gauss}} 
\end{figure} 

\paragraph{Non-gaussian setting}

We now consider a non-gaussian setting, where the first two coordinates of nulls and novelties are independent draws from $\text{Beta}(5, 5)$ and $\text{Beta}(1, 3)$, respectively, and the other coordinates are independent draws from $\text{Beta}(1, 1)$ for both nulls and novelties. Note that \texttt{AdaDetect parametric} is now based on a misspecified parametric model; see Section~\ref{sec:simumethod} for detail. Figure \ref{fig:beta} presents the results with the dimension $d$ varying. \texttt{AdaDetect parametric} is now clearly dominated by \texttt{AdaDetect RF}, especially in high dimensions. This shows that machine learning-based classification methods are more robust to model misspecification when combined with AdaDetect.

\begin{figure}
     	\begin{minipage}{0.85\linewidth}
		\raisebox{-0.5 mm}{ \includegraphics[width=0.215\linewidth]{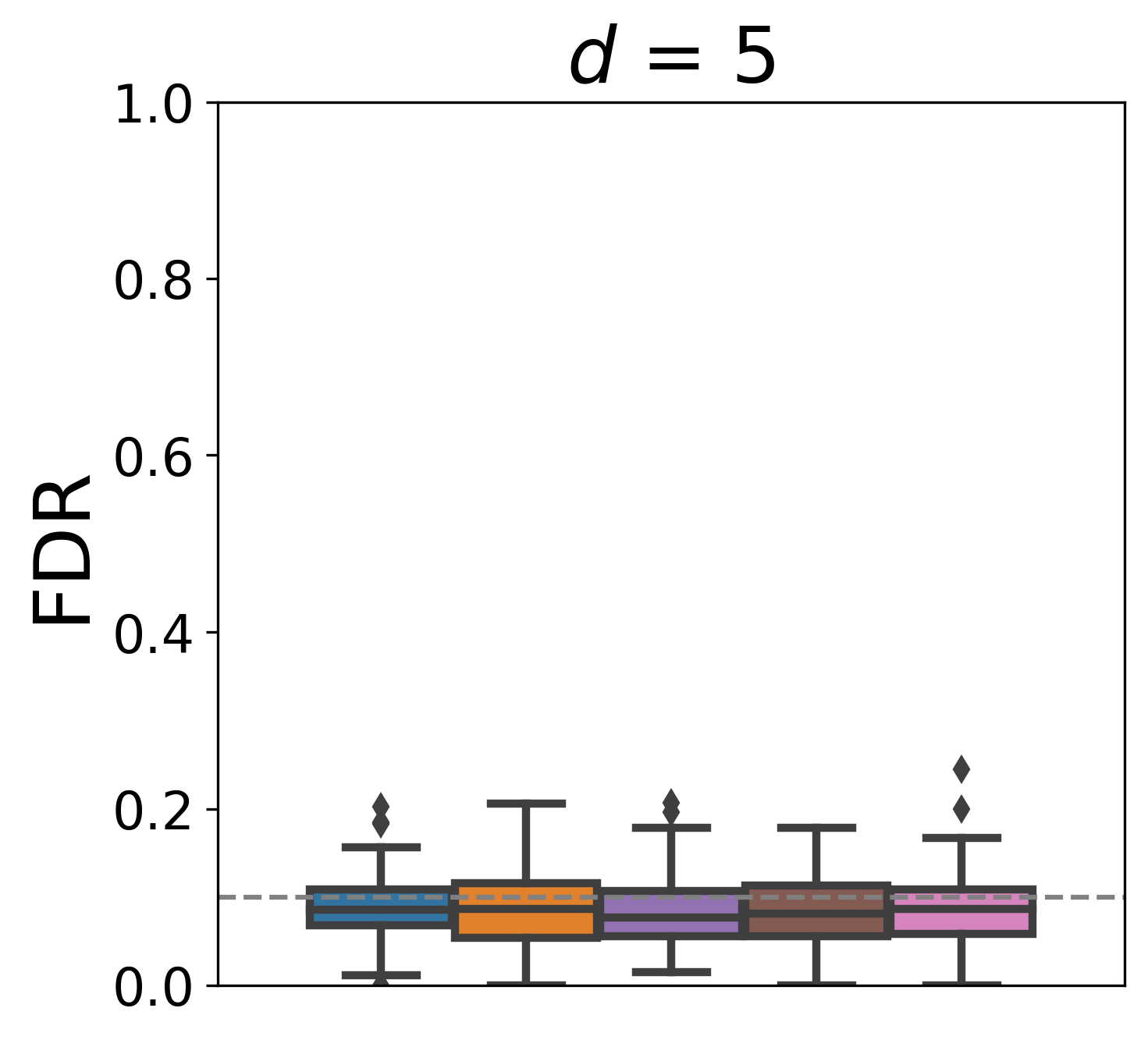} }
		\includegraphics[width=0.18\linewidth]{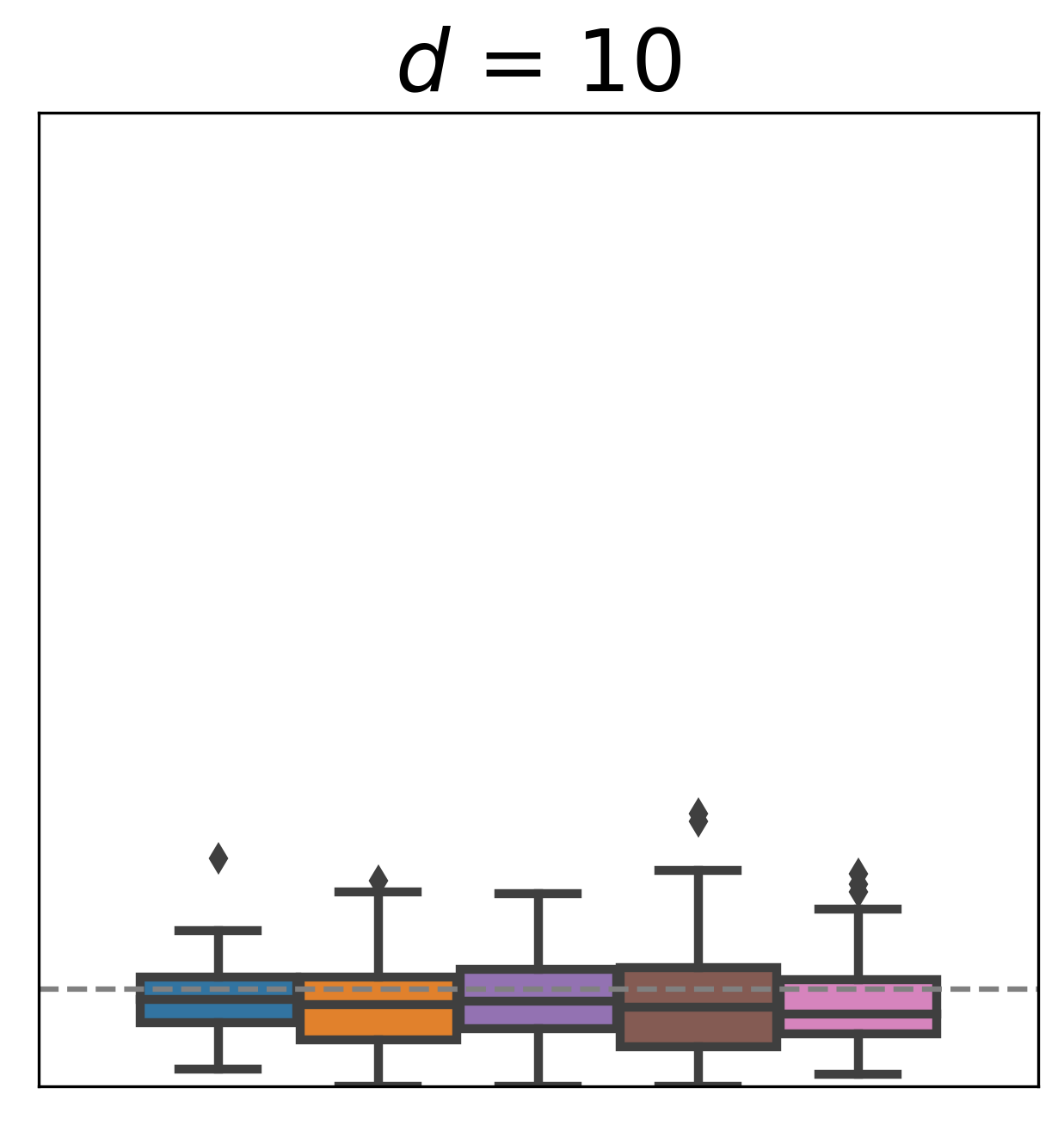} 
		\includegraphics[width=0.18\linewidth]{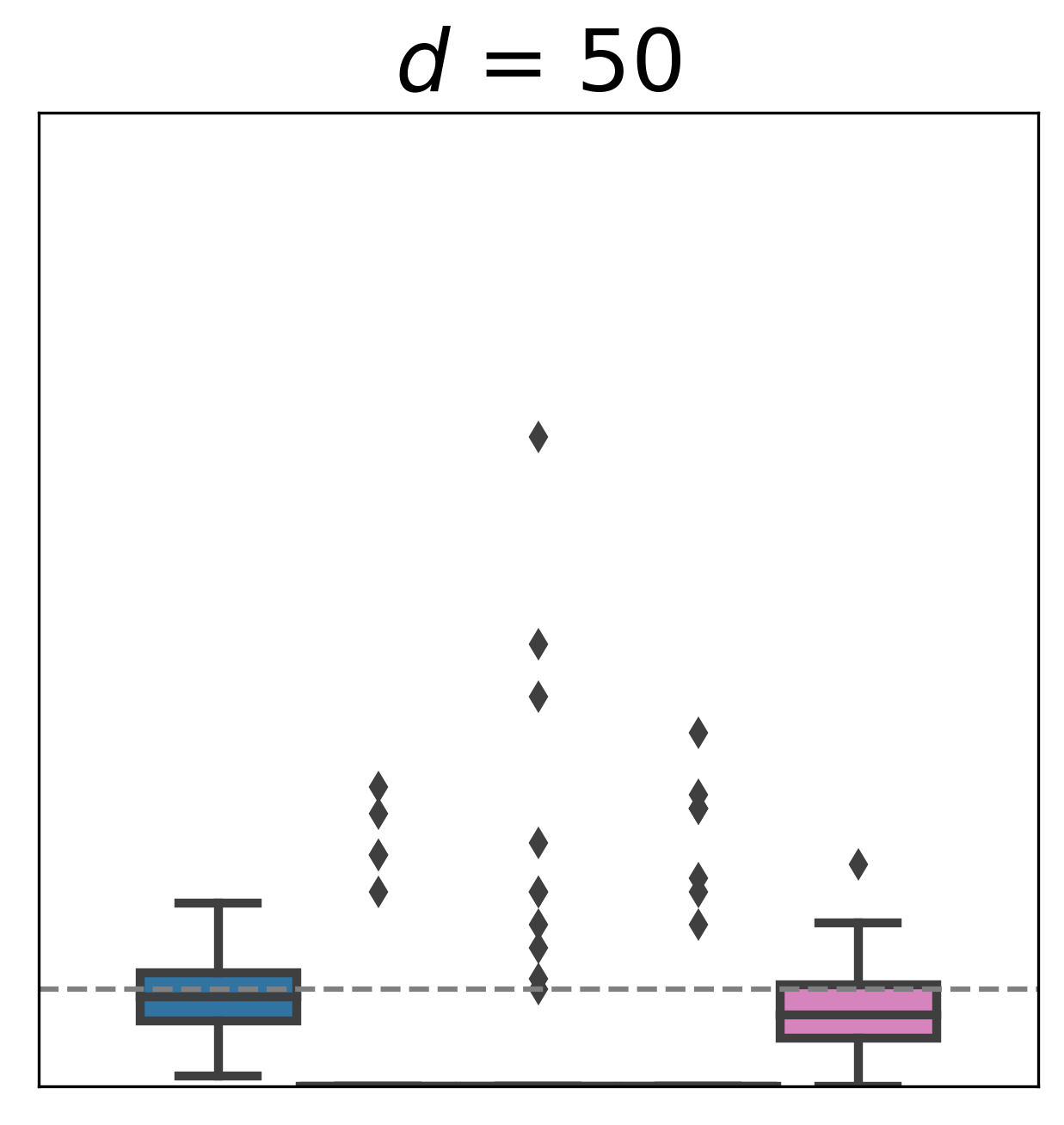} 
		\includegraphics[width=0.18\linewidth]{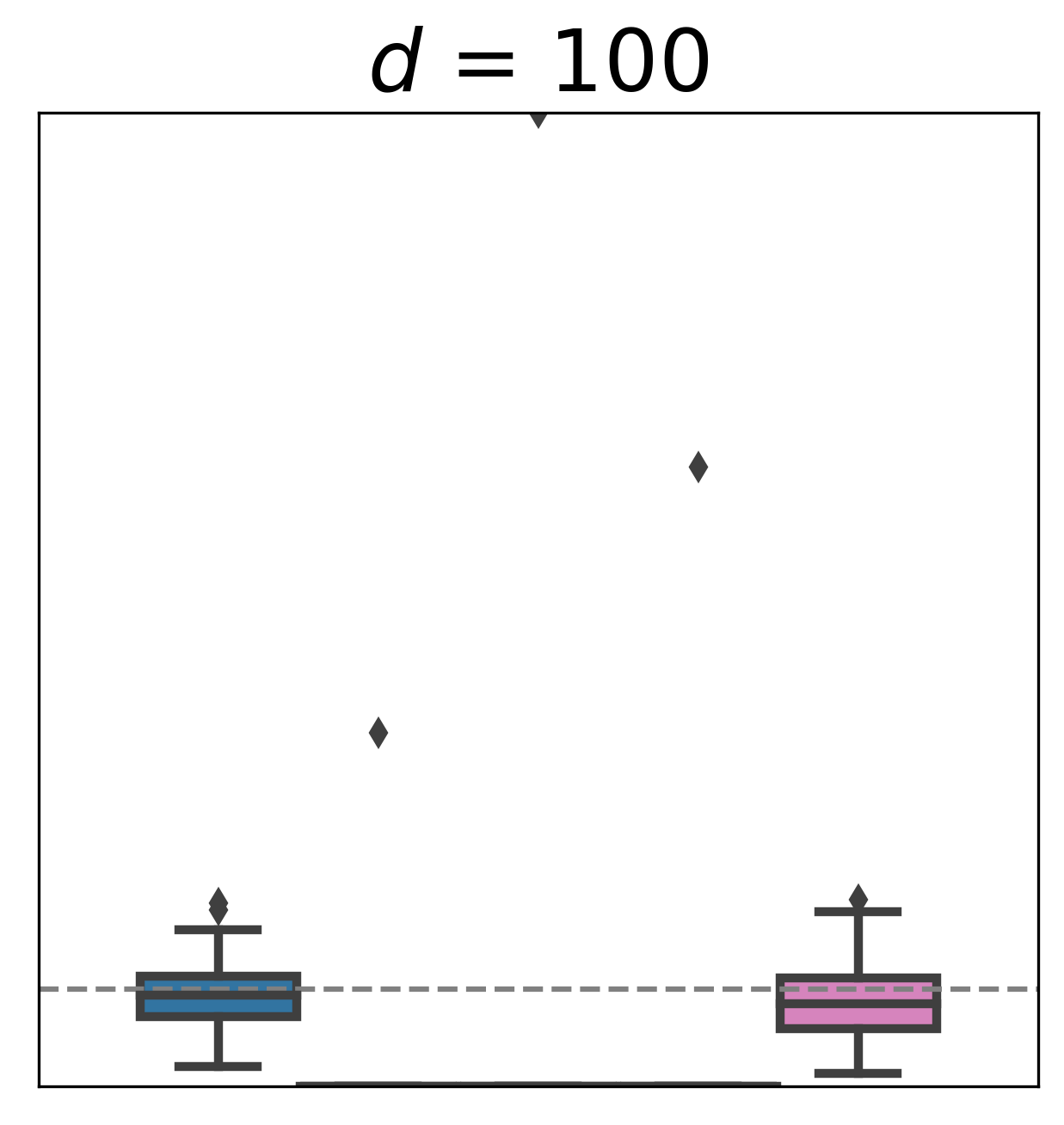} 
		\includegraphics[width=0.18\linewidth]{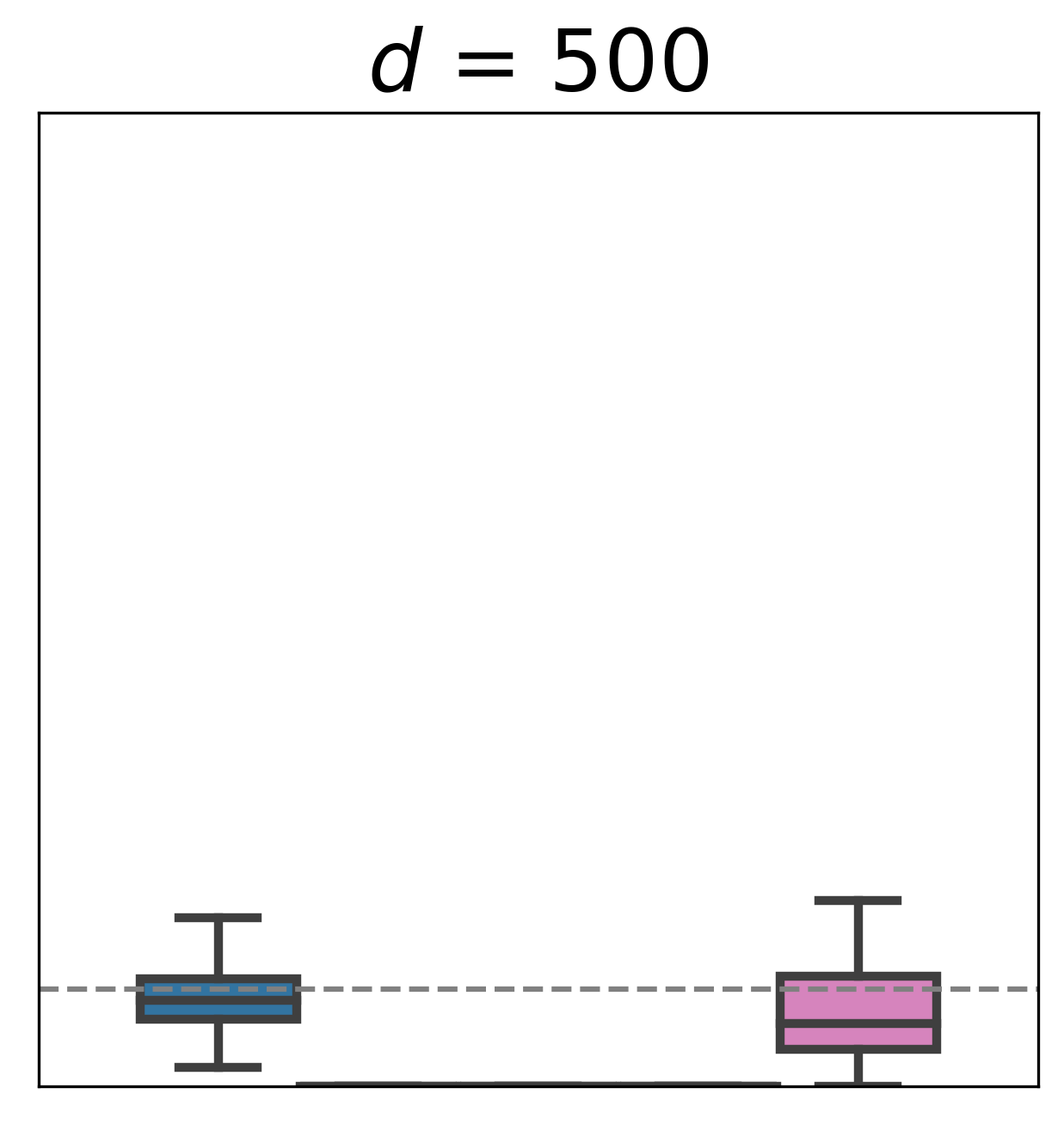} 
		\end{minipage}%
		\begin{minipage}[c]{0.15\linewidth}
		\includegraphics[width=\linewidth]{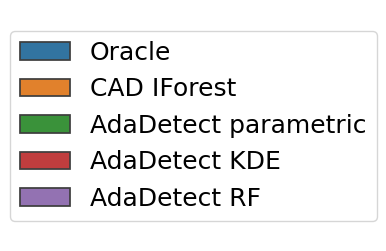} 
		\end{minipage}
		\raggedright
		\begin{minipage}{0.85\linewidth}
		\raisebox{-0.5 mm}{ \includegraphics[width=0.215\linewidth]{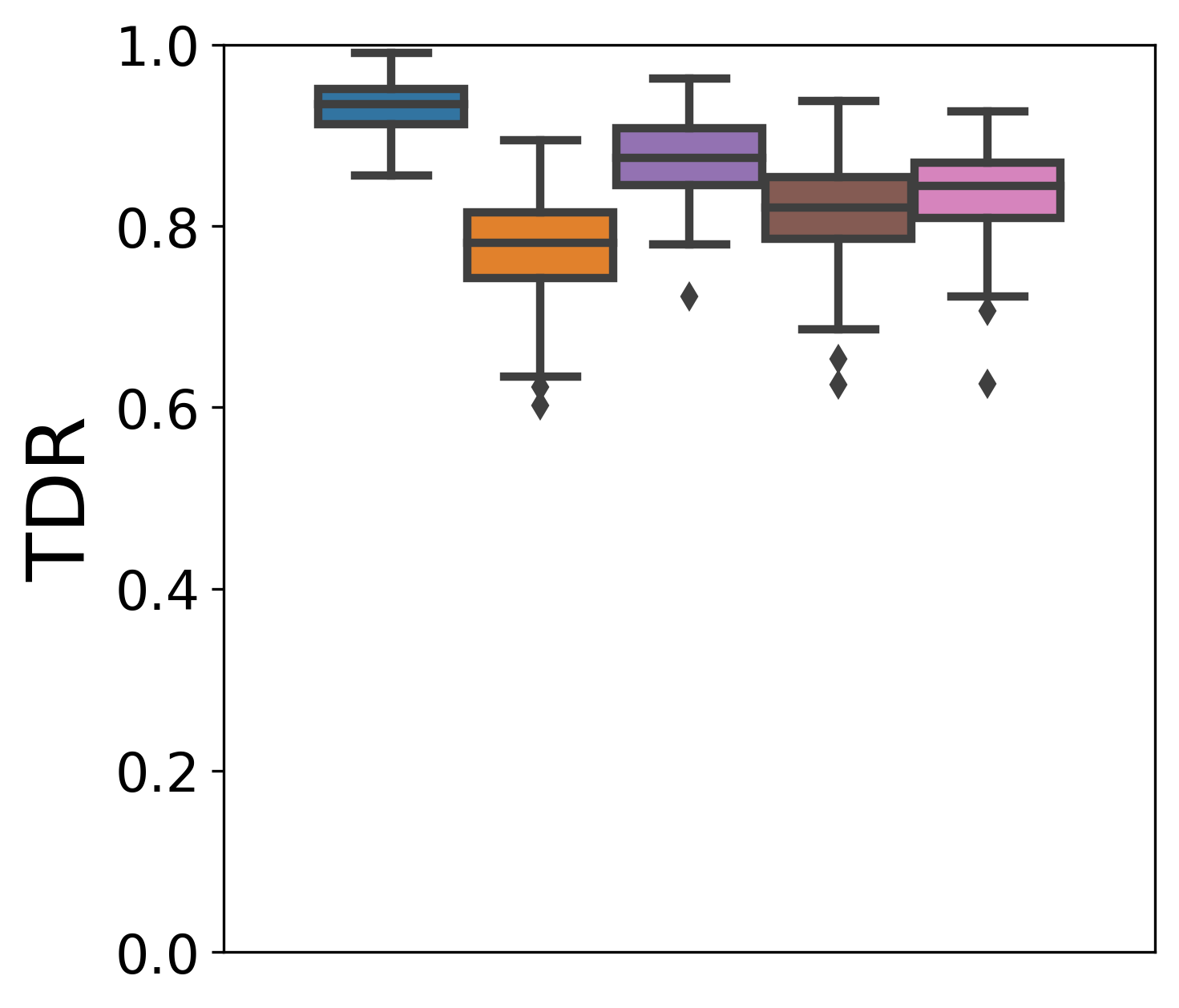} }
		\includegraphics[width=0.18\linewidth]{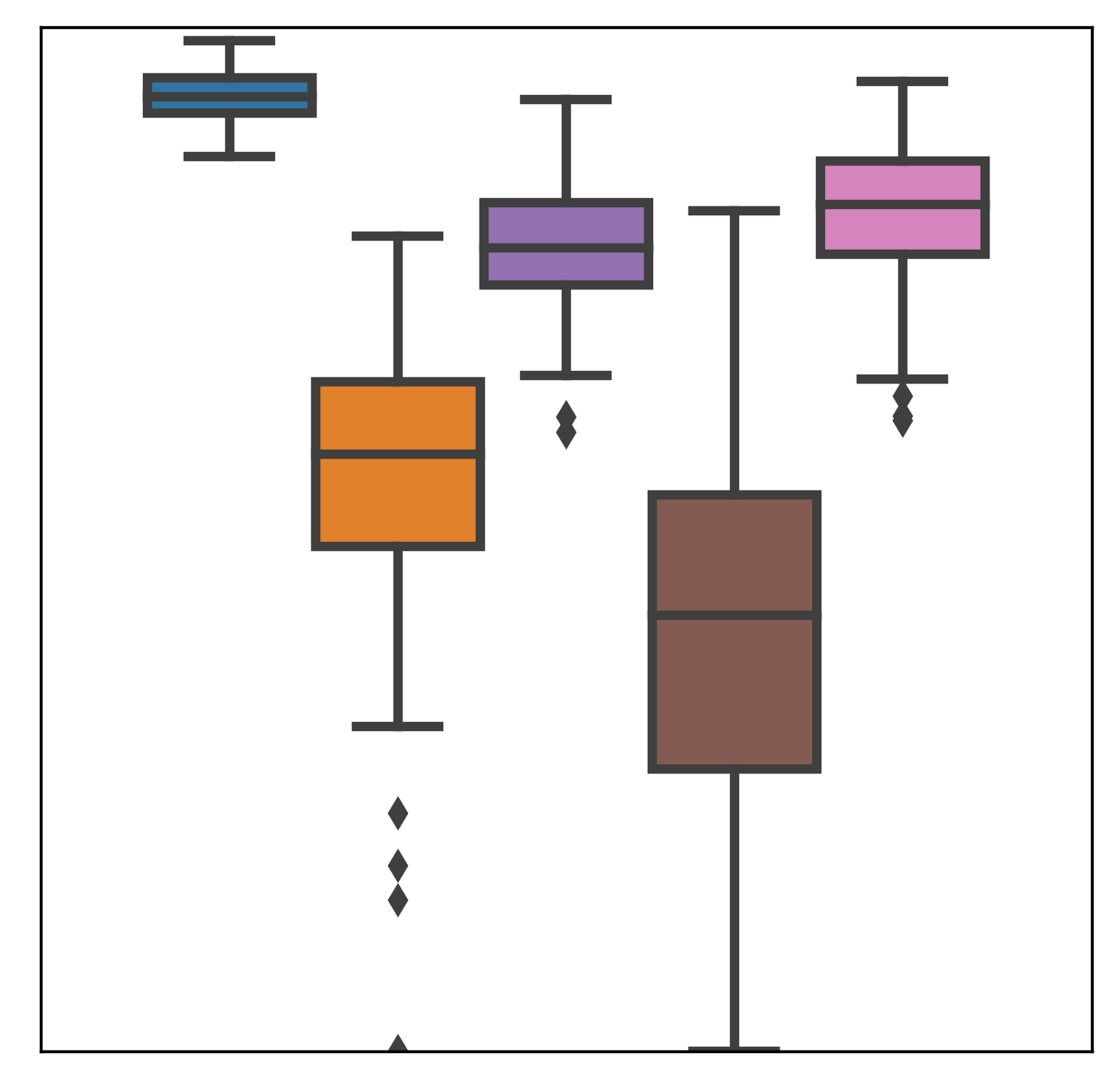} 
		\includegraphics[width=0.18\linewidth]{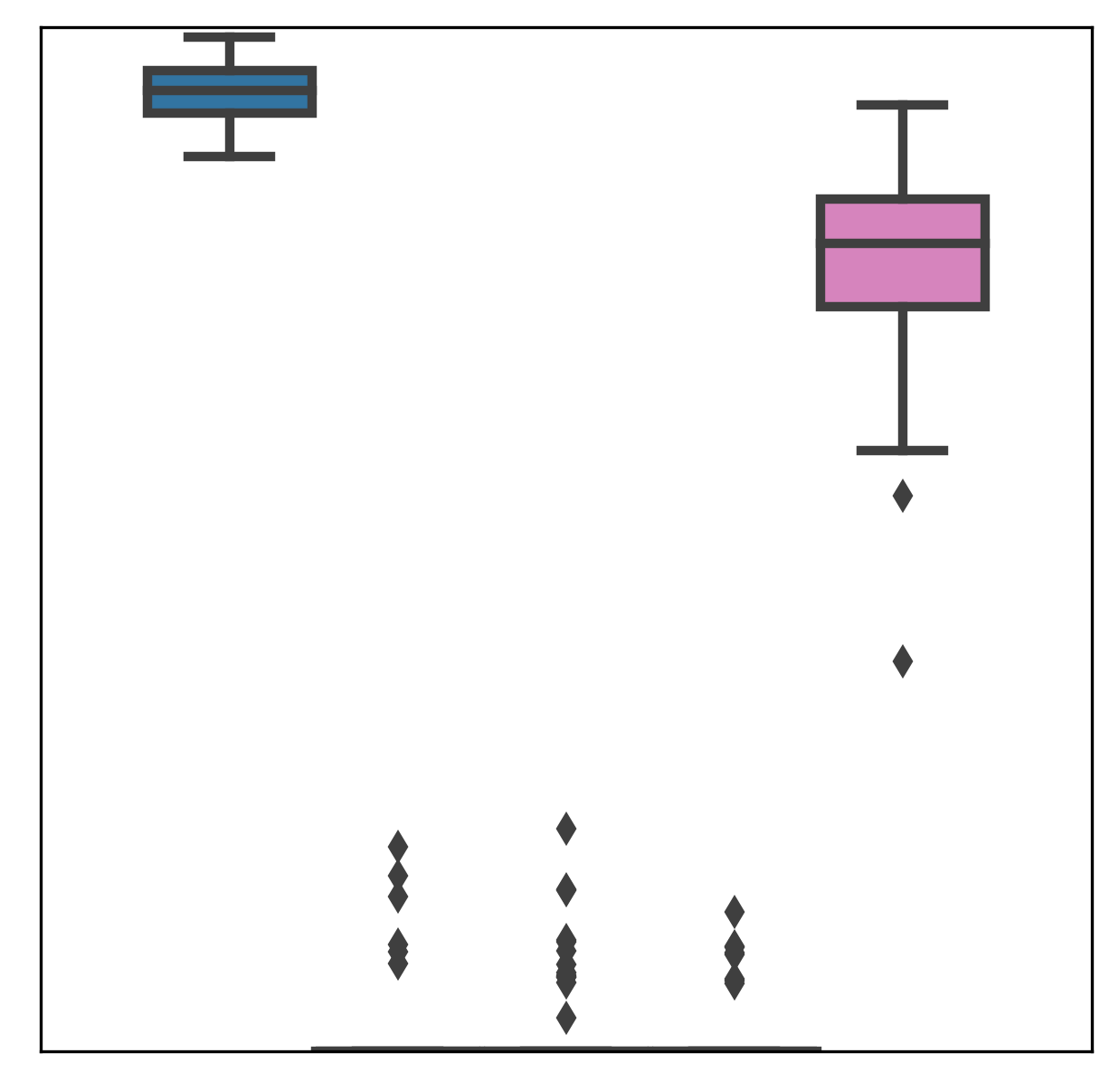} 
		\includegraphics[width=0.18\linewidth]{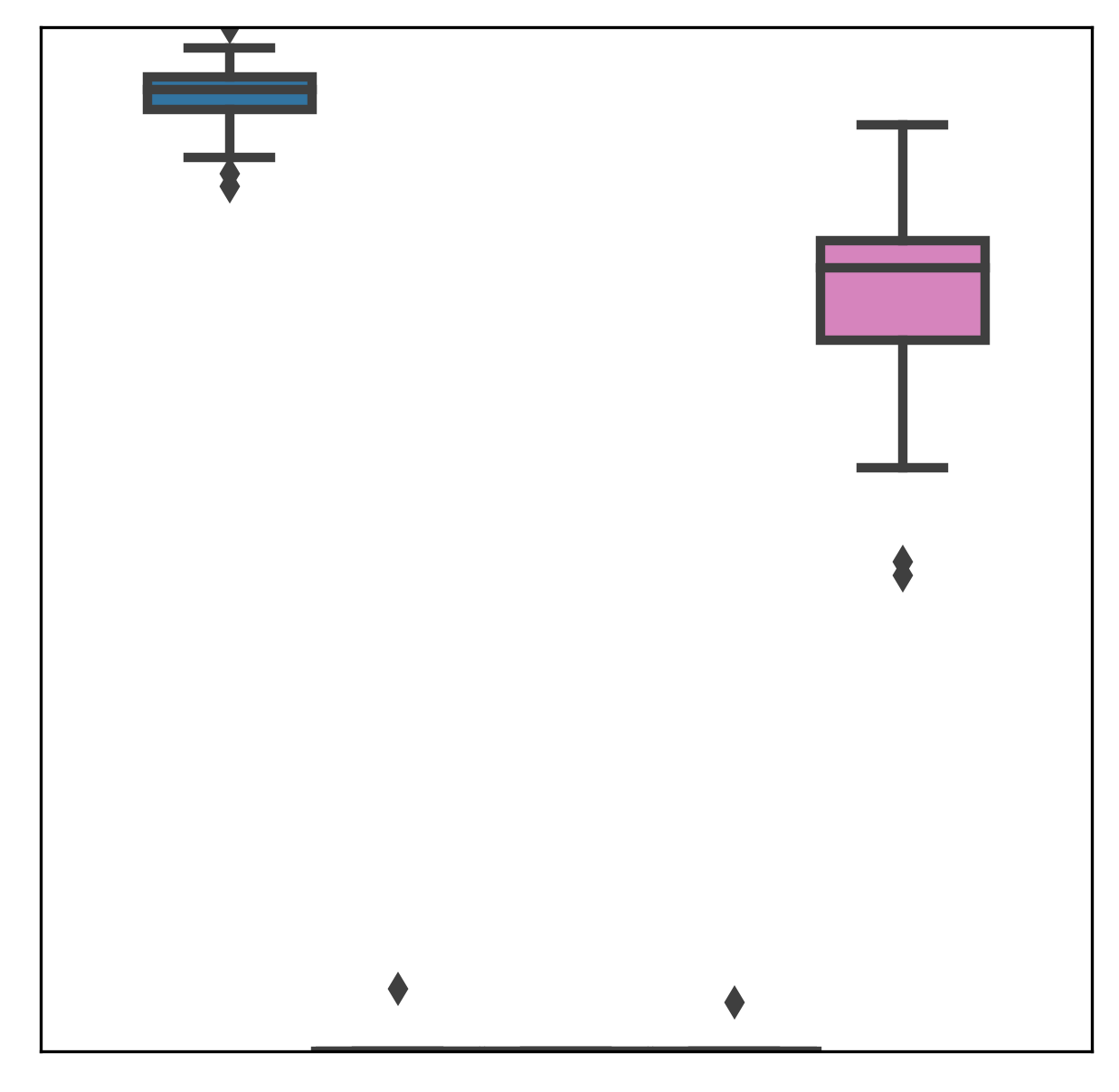} 
		\includegraphics[width=0.18\linewidth]{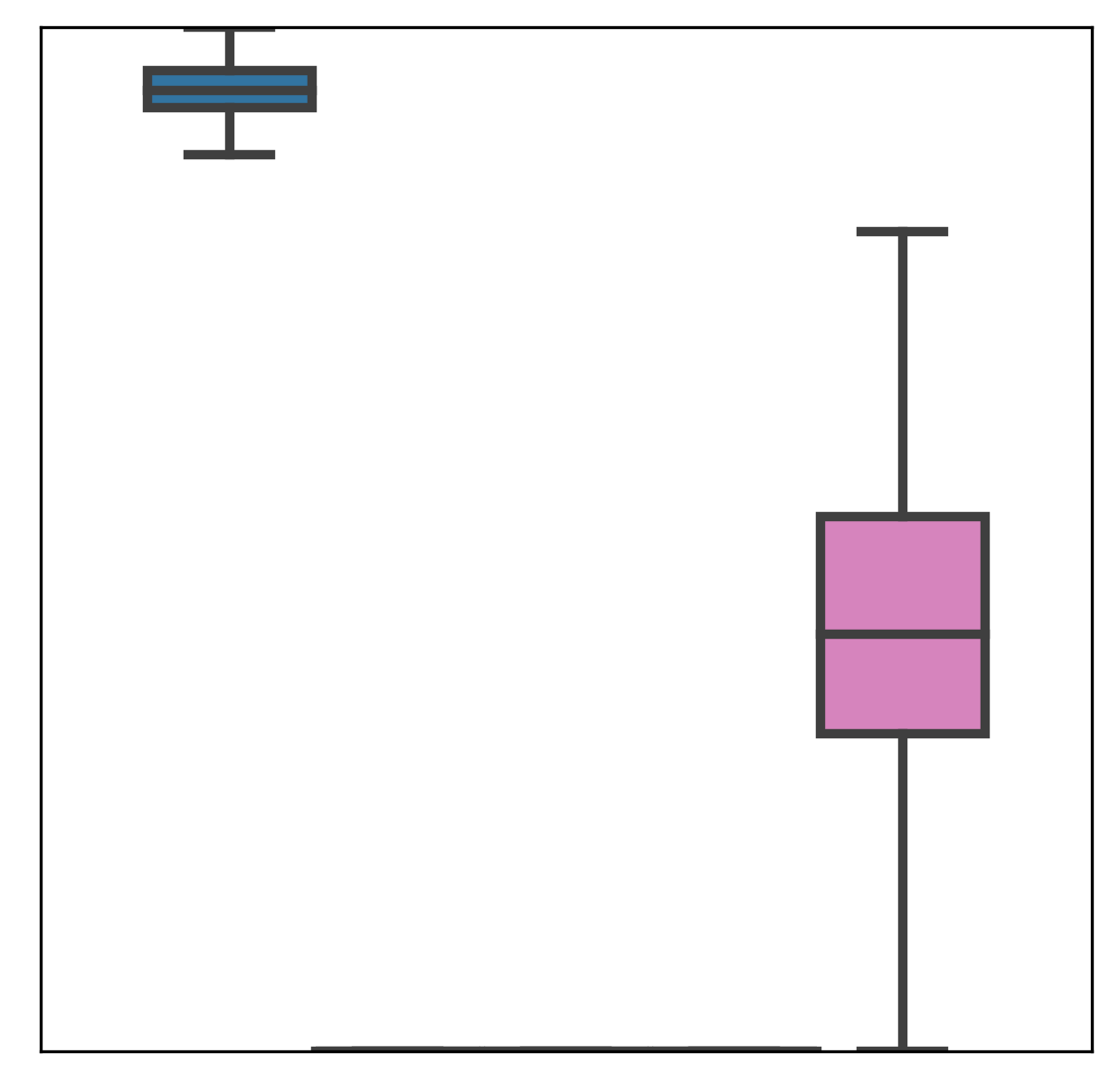} 
		\end{minipage}
		
	\caption{Non-gaussian (beta) setting. FDR (top) and TDR (bottom) as a function of $d$. The dashed line indicates the target level.} \label{fig:beta}
\end{figure}

\paragraph{Additional experiments with varying $k$ and $\ell$}

To examine the effect of $k$ and $\ell$, we study the performance of \texttt{AdaDetect RF} with \texttt{AdaDetect oracle} in the simple setting where $P_0 = \mathcal{N}(0, I_d)$ and $P_1=\mathcal{N}(\mu,I_d)$ with $d=4$ and $\mu = (\sqrt{2}, \dots, \sqrt{2})$.

Figure \ref{fig:lvar} presents the results for $k=m=1000$ and varying $\ell$. Interestingly, the TDR is not monotone in $\ell$. This is because a small $\ell$ makes p-value inaccurate while a large $\ell$ dilutes the signal in the mixed sample $(X_{k+1},\dots,X_{m+n})$ and degrades the quality of classification-based score function.

For $\ell=m=1000$ and varying values of $k$, a similar pattern is observed in Figure \ref{fig:kvar}. This suggests that an imbalanced classification problem may arise due to the presence of extreme values of $k$.


\begin{figure}[h!]
		\begin{minipage}{0.425\linewidth}
		\includegraphics[width=0.49\linewidth]{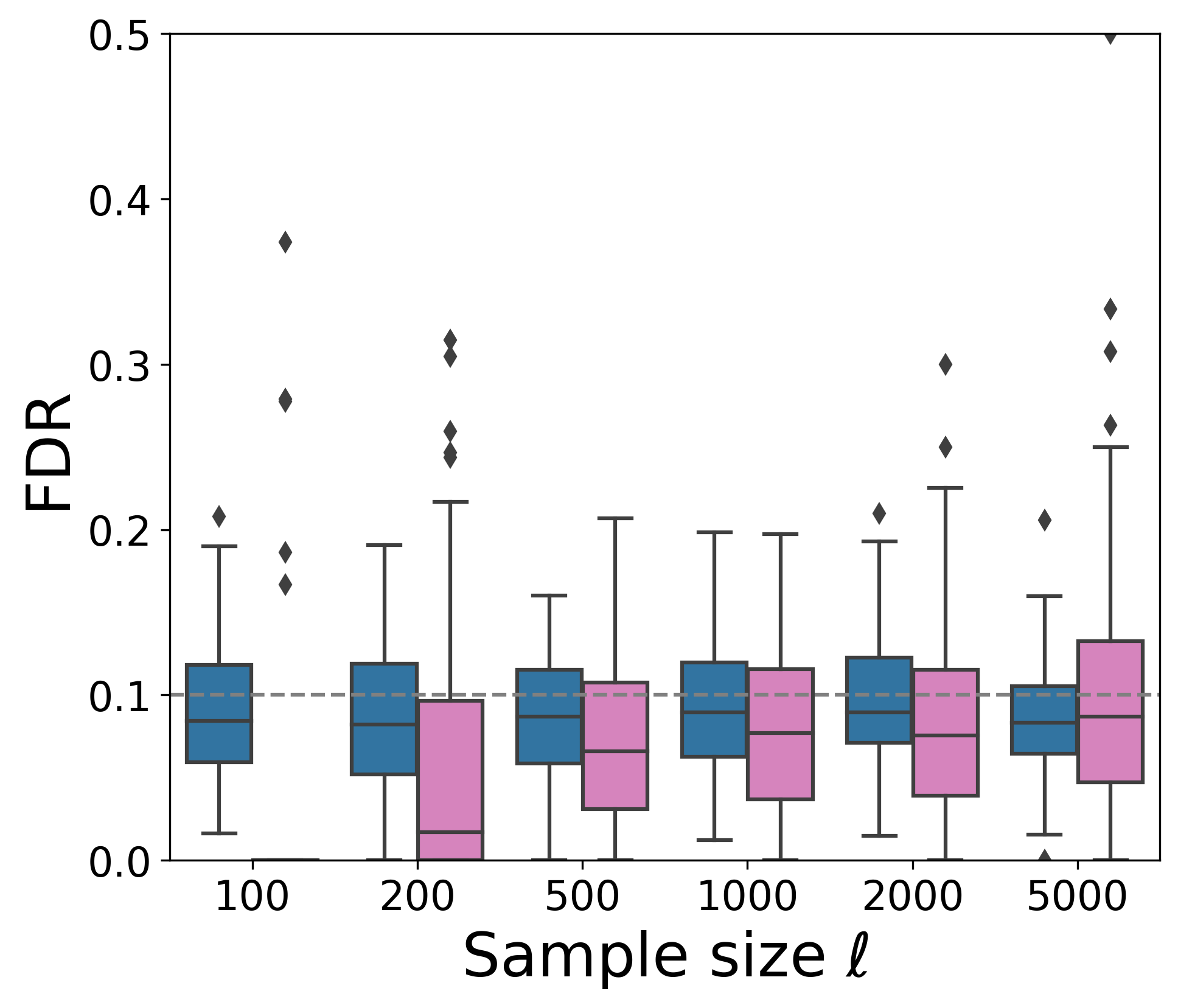} 
		\includegraphics[width=0.49\linewidth]{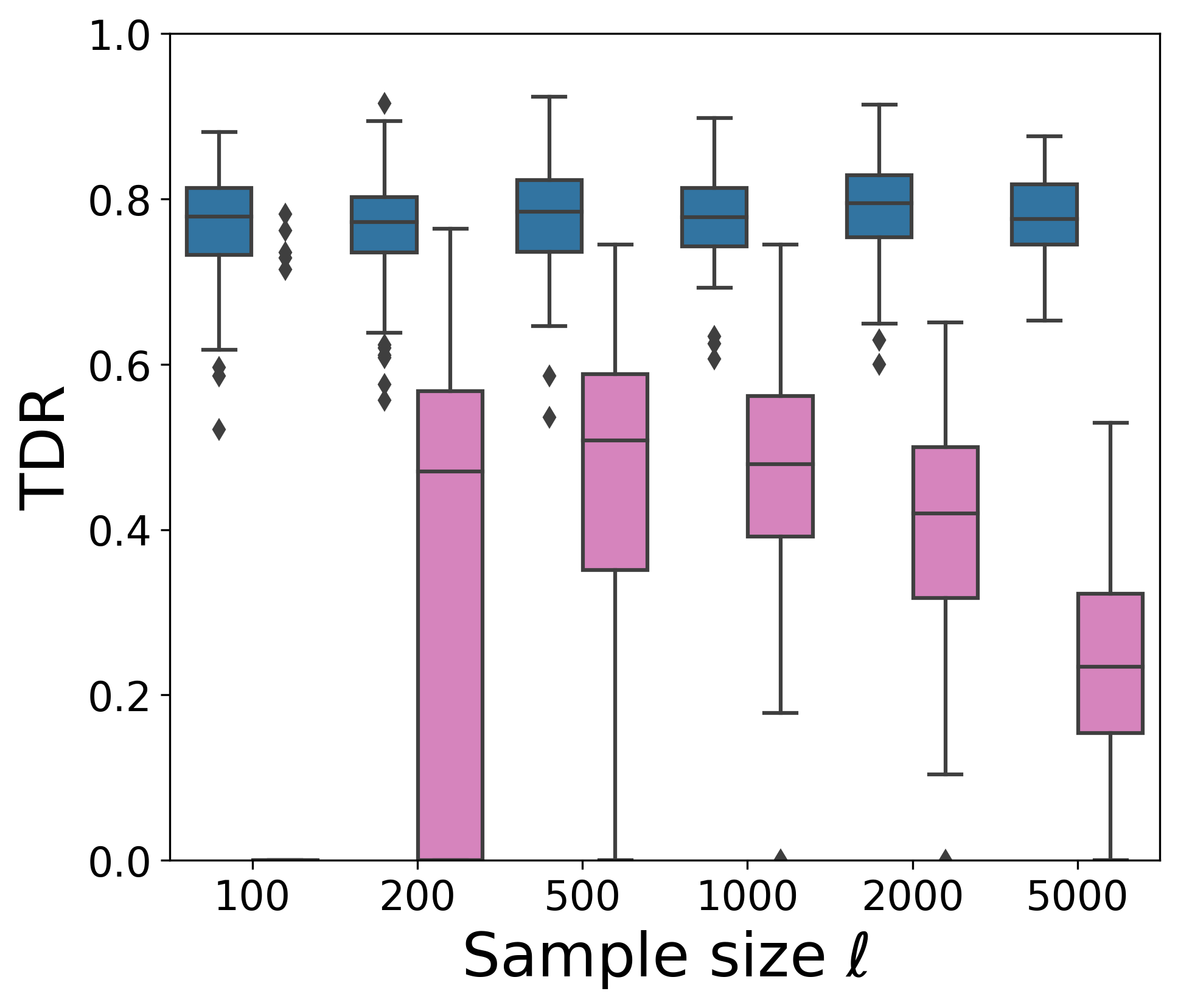} 
		\subcaption{}
		\label{fig:lvar}
		\end{minipage}%
		\begin{minipage}{0.425\linewidth}
		\includegraphics[width=0.49\linewidth]{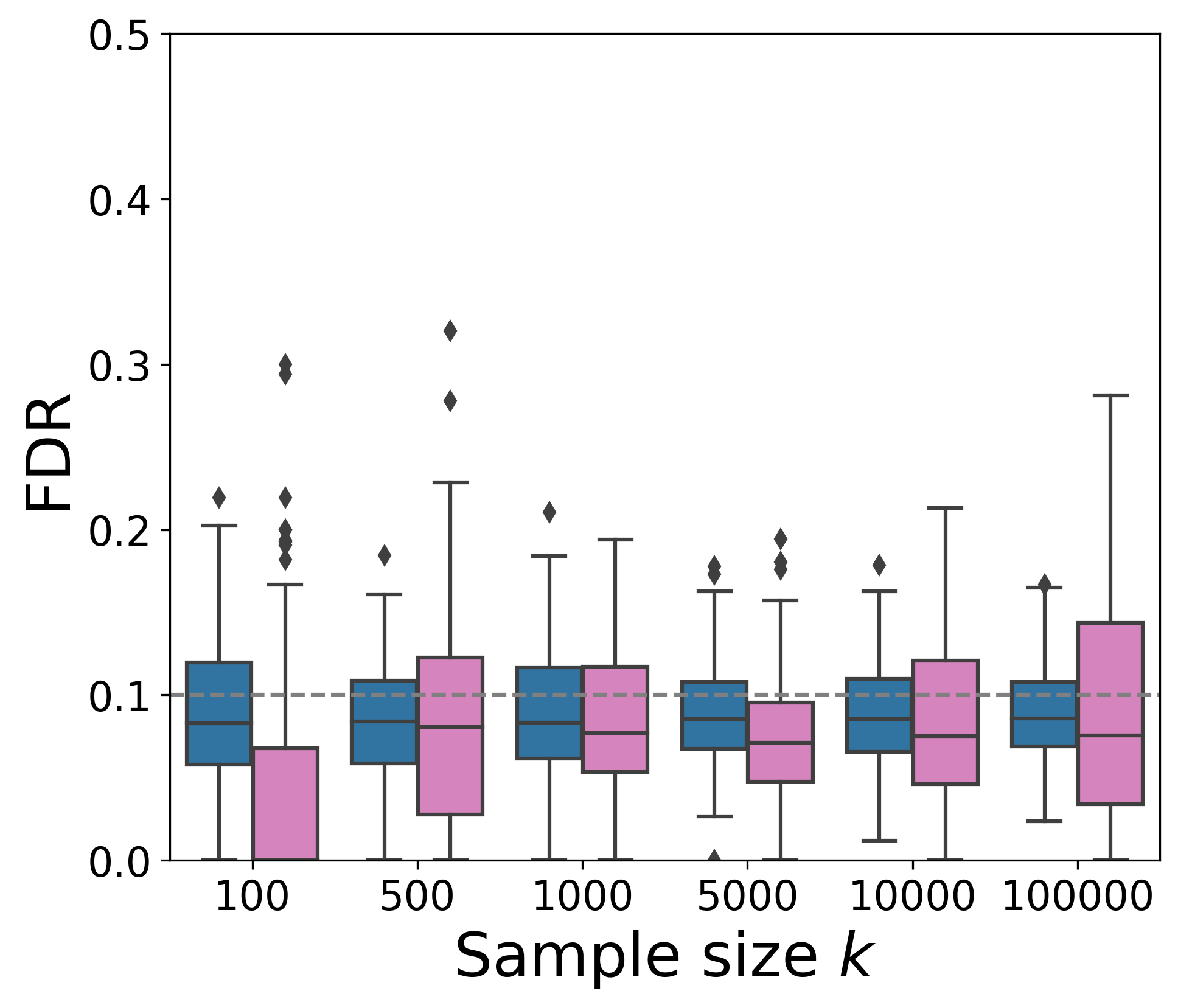} 
		\includegraphics[width=0.49\linewidth]{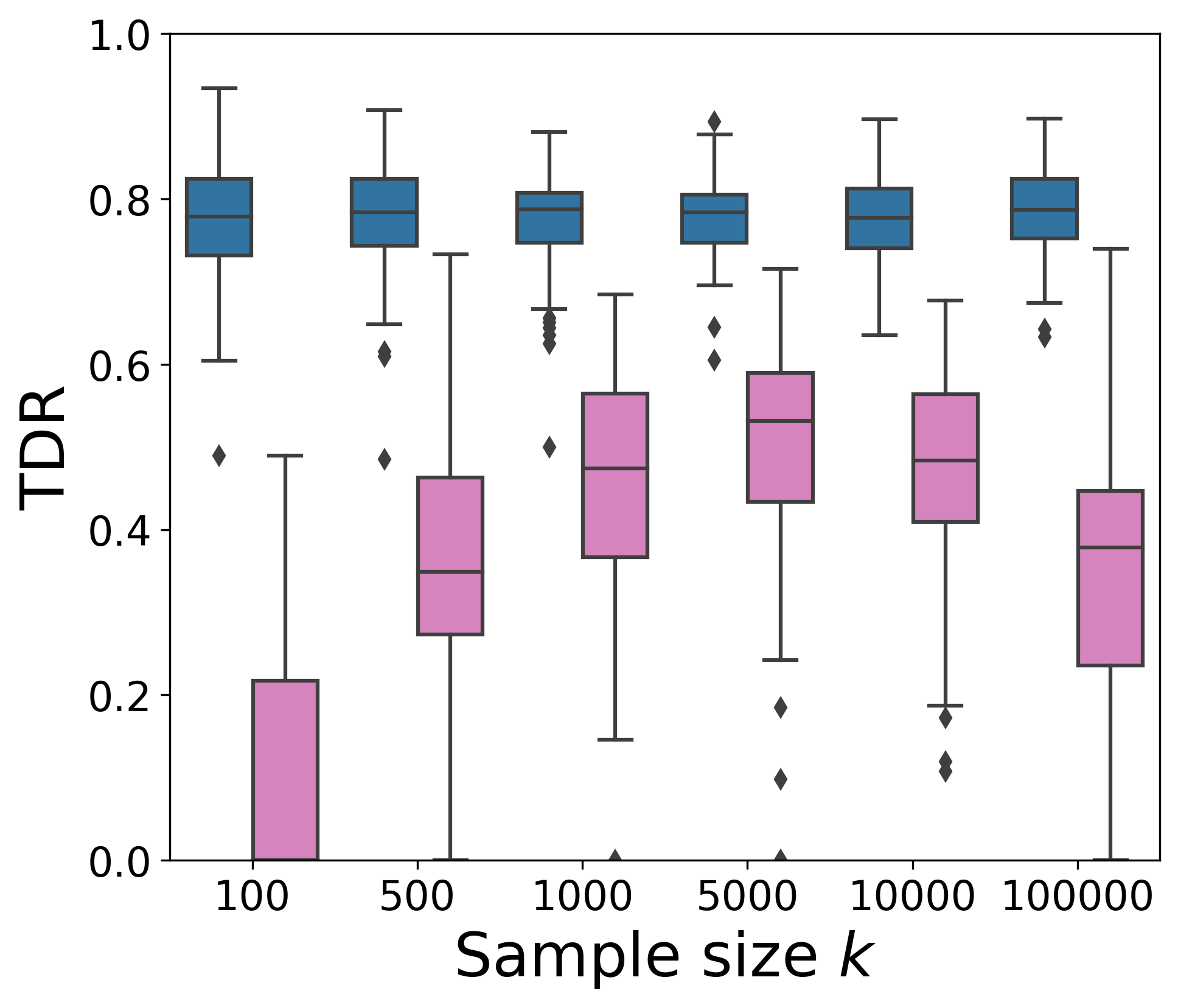} 
		\subcaption{}  
		\label{fig:kvar}
		\end{minipage}%
		\begin{minipage}[c]{0.15\linewidth}
		\includegraphics[width=\linewidth]{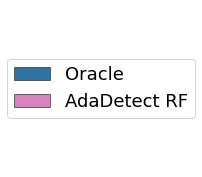} 
		\end{minipage}
	\caption{FDR and TDR as a function of $\ell$ (panel (a)) and $k$ (panel (b)). The dashed line indicates the nominal level. }\end{figure}

\newpage
\section{Additional experimental results for the astronomy application}\label{sec:appliplus}

In this section, we provide more results for more settings. Recall that, for each experiment, we sample $n$ nonvariable stars from as the NTS along with $m_1$ variable stars and $m_0 = m - m_1$ additional nonvariable stars as the test sample. For all experiments, we set $m = 100$.

First, we show how TDR varies with sparsity measured by $m_1$. Figure~\ref{app2} presents the TDR for $m_1 \in \{5, 15, 40, 90\}$ with different target FDR levels shown in the title of each panel. When $\alpha$ is low, the left two panels show that AdaDetect RF substantially outperforms the other methods. When the novelties are sparse, the panels in the middle column shows that AdaDetect RF still performs well when $\alpha = 0.2$ but underperforms when $\alpha = 0.5$, though $\alpha = 0.5$ is arguably less relevant in practice. When the novelties are dense, the right panels show that AdaDetect with adaptive scores outperform AdaDetect with non-adaptive scores when $\alpha = 0.05$ and they are nearly indistinguishable when $\alpha = 0.5$.

\begin{figure}[h!]
\centerline{
\includegraphics[scale=0.33]{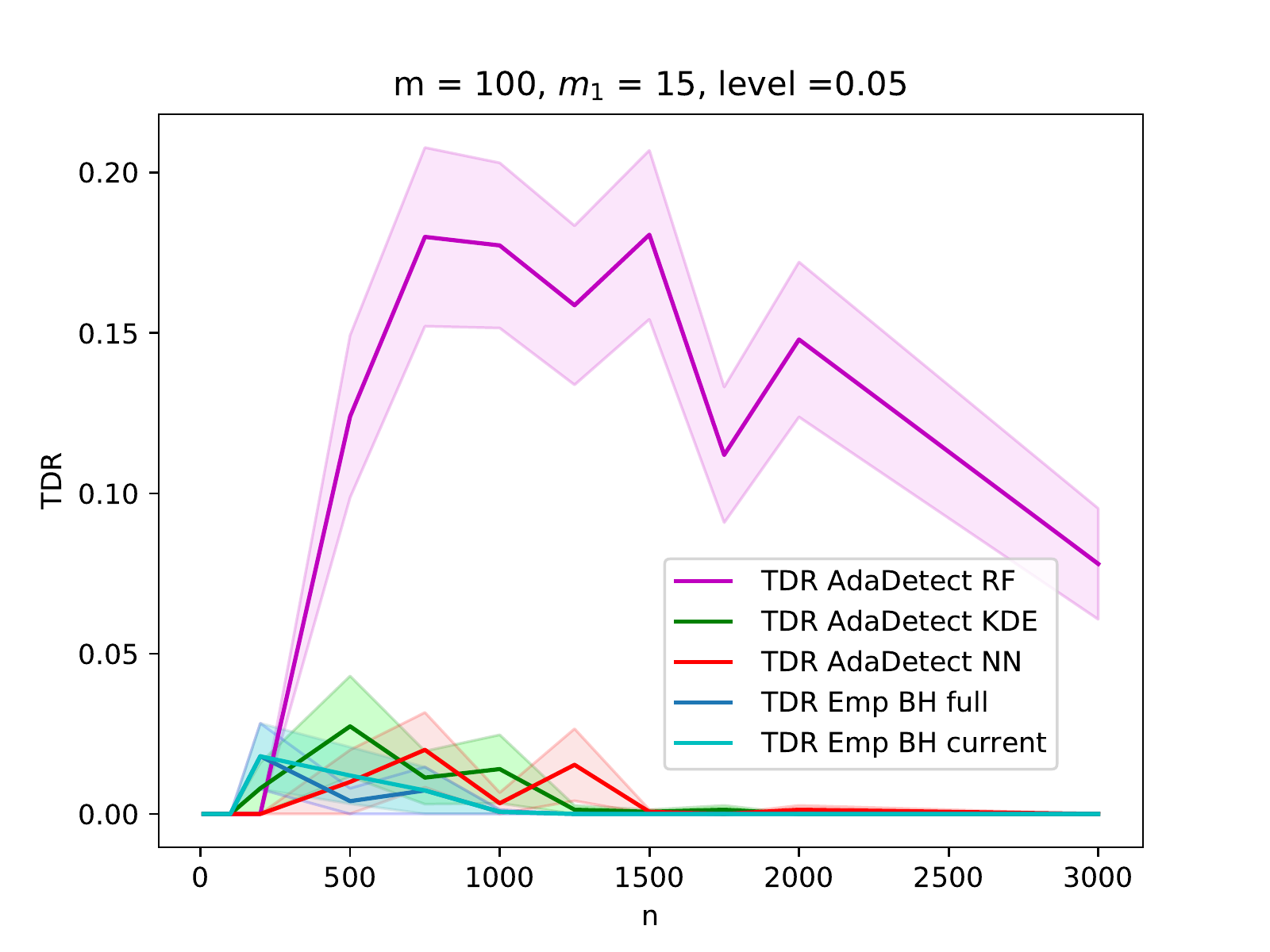}
\includegraphics[scale=0.33]{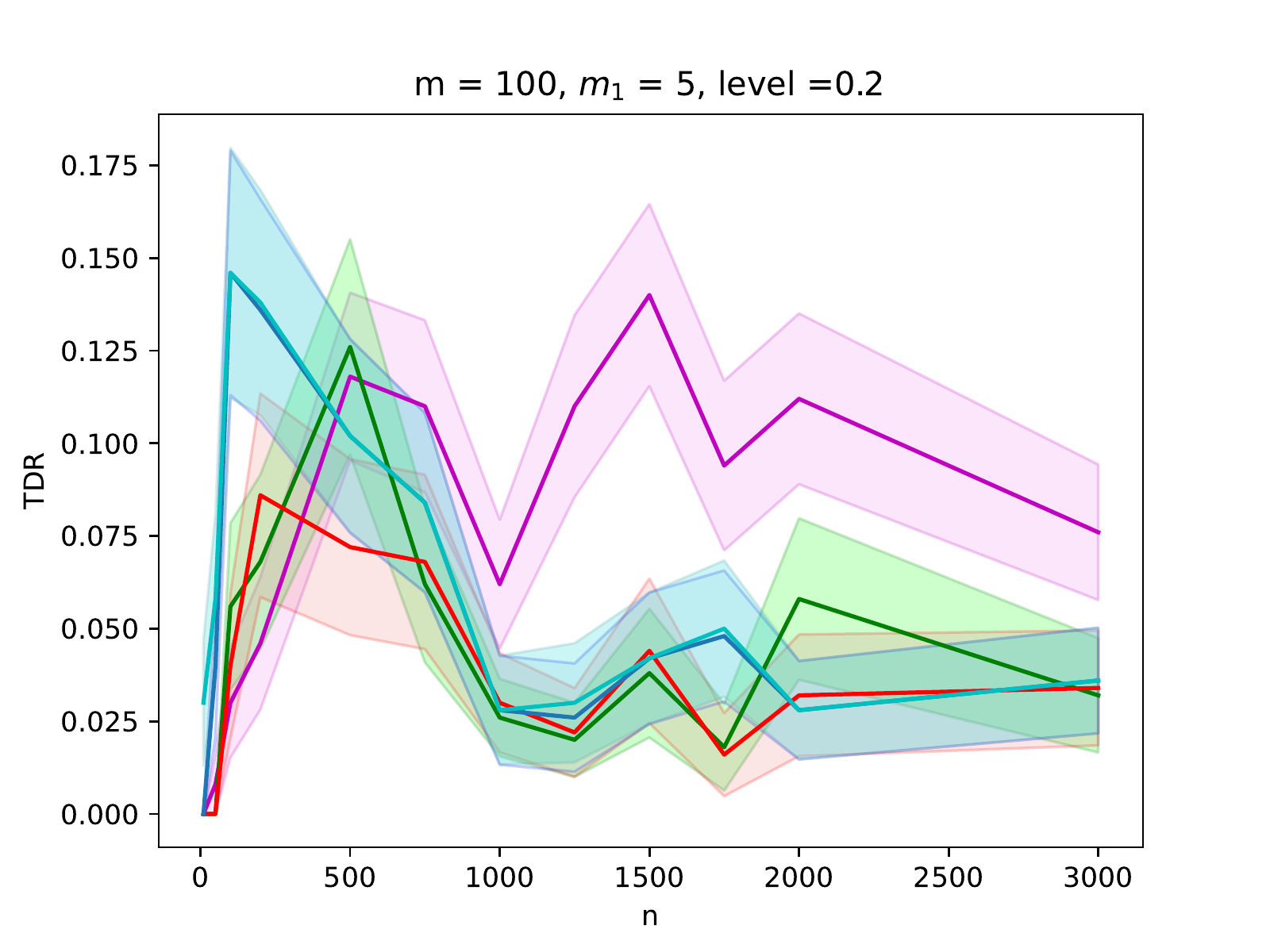}
\includegraphics[scale=0.33]{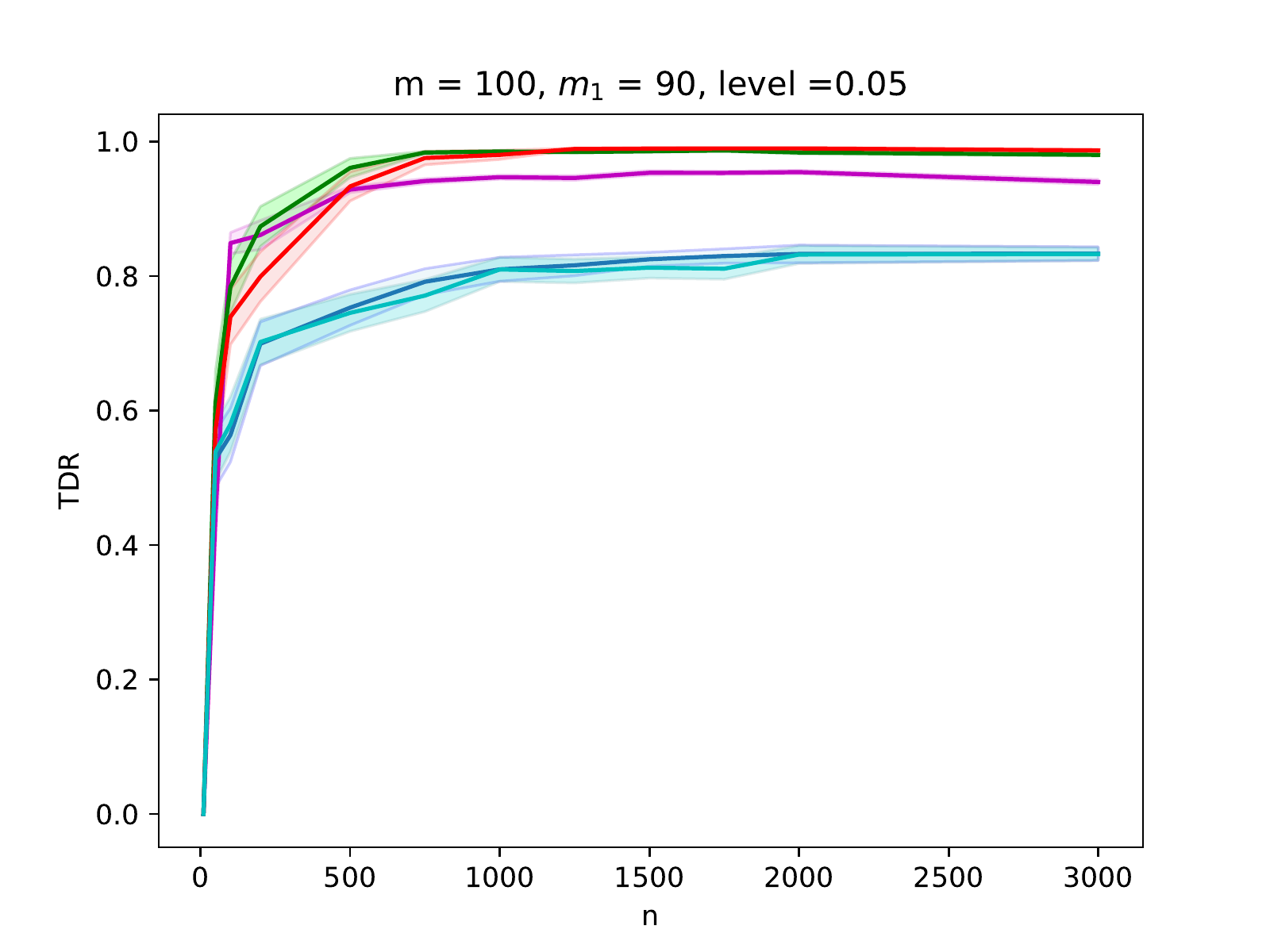}
}
\centerline{
\includegraphics[scale=0.33]{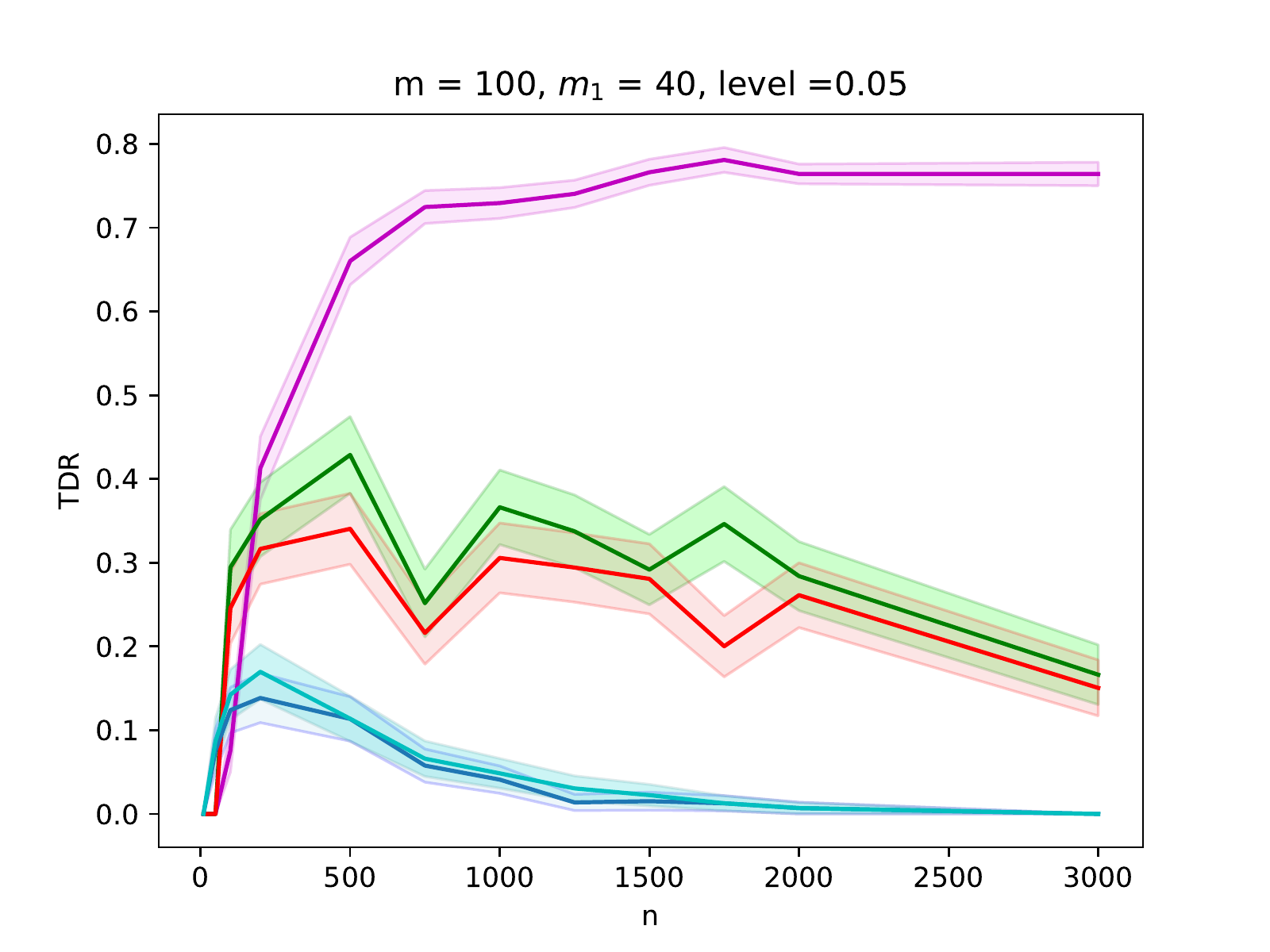}
\includegraphics[scale=0.33]{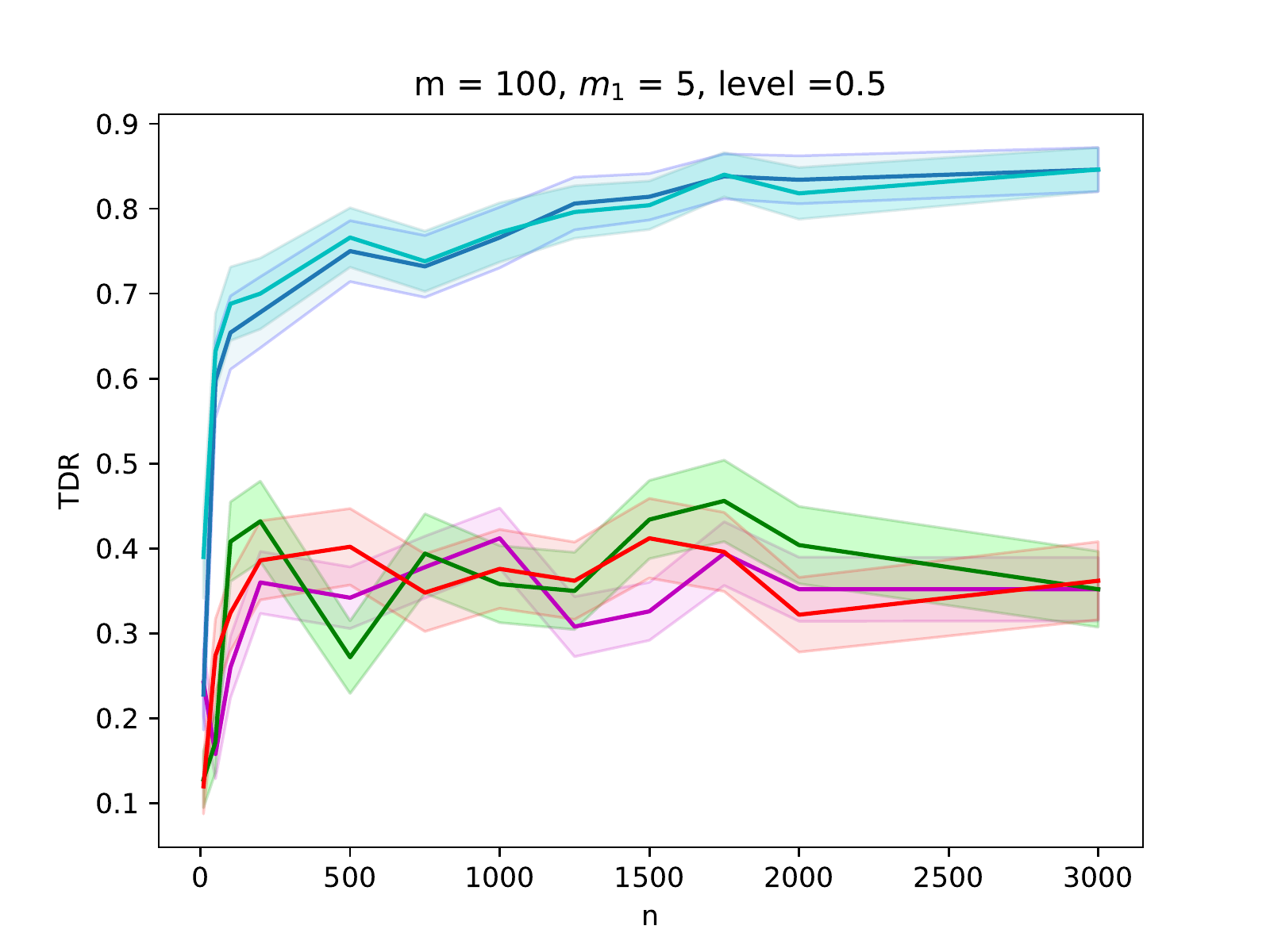}
\includegraphics[scale=0.33]{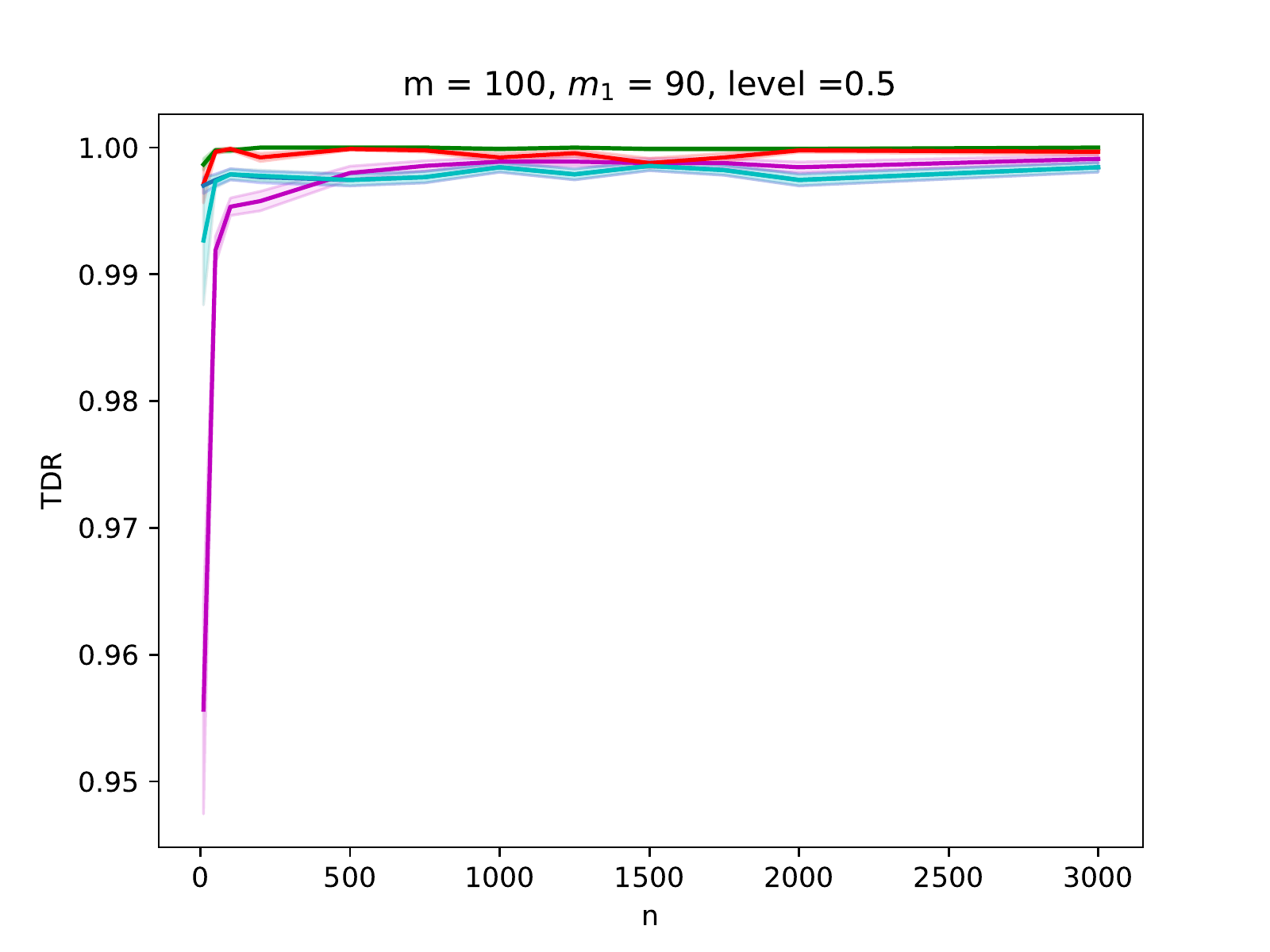}
}
\caption{ TDR for different sparsity  regimes: see $m_1$ and $\alpha$ in the titles. In all plots $m=100$.}
\label{app2}
\end{figure}

Next, we compare the FDR and TDR of AdaDetect KDE and Empirical BH with both $m_1$ and $n$ varying in Figure~\ref{app3}. For the purpose of visualization, we only show the point estimates of FDR and TDR without uncertainty measures. The left column shows that both methods control the FDR, though AdaDetect KDE is generally less conservative. The right column shows that AdaDetect KDE almost always has a higher power and it starts to reject at a higher sparsity level.
An interesting observation is that the power of AdaDetect KDE is decreasing in $n$ when the novelties are sparse (e.g., $m_1\approx 40$). This can be explained by the fact that the `contamination' of the NTS increases with increasing $n$, leading to a more noisy estimation of the test statistic.

\begin{figure}[h!]
\centerline{
\includegraphics[scale=0.5]{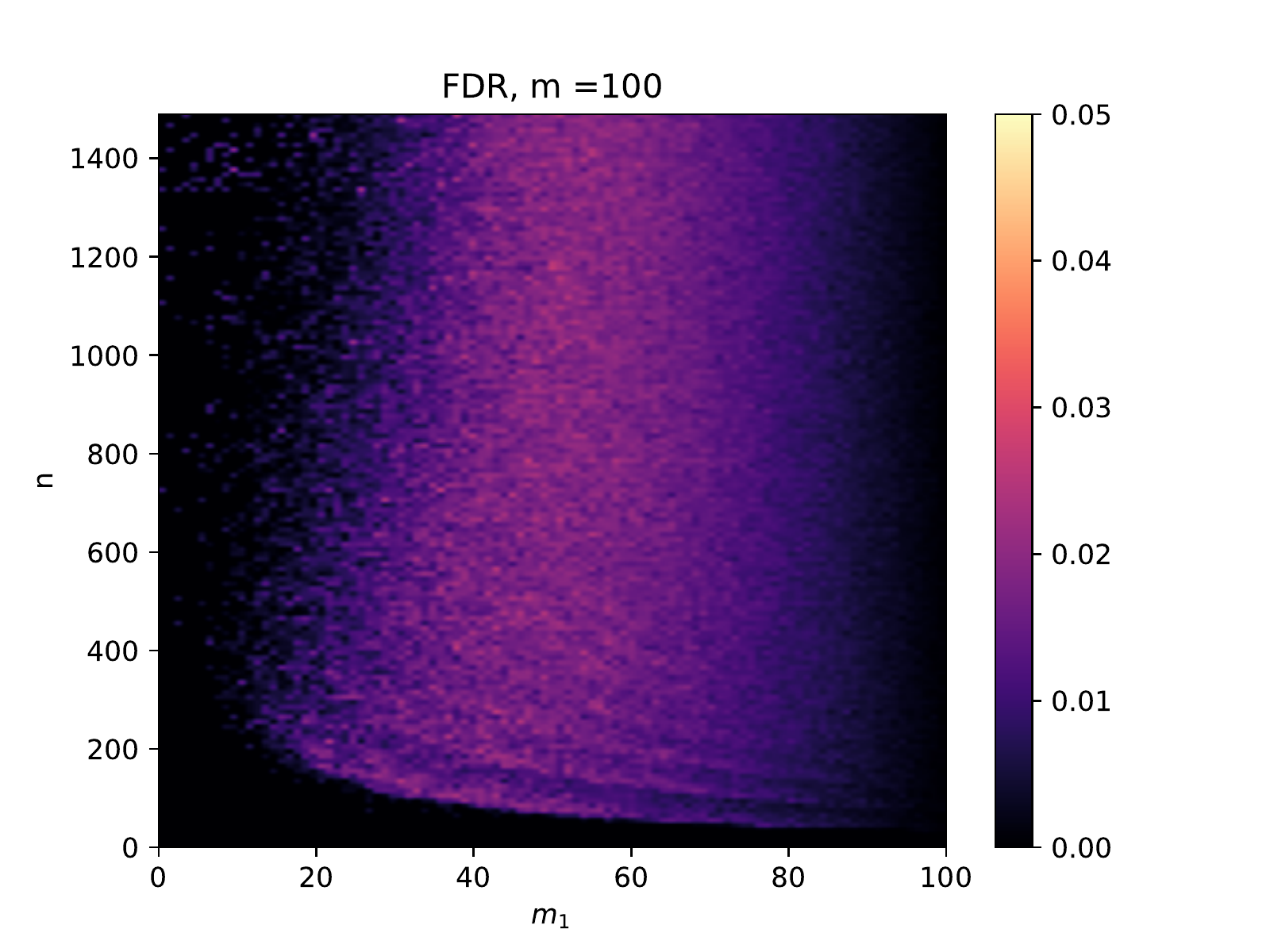}
\includegraphics[scale=0.5]{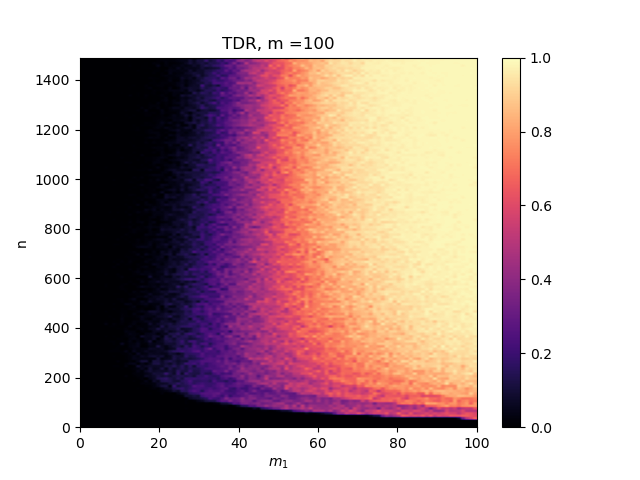}
}
\centerline{
\includegraphics[scale=0.5]{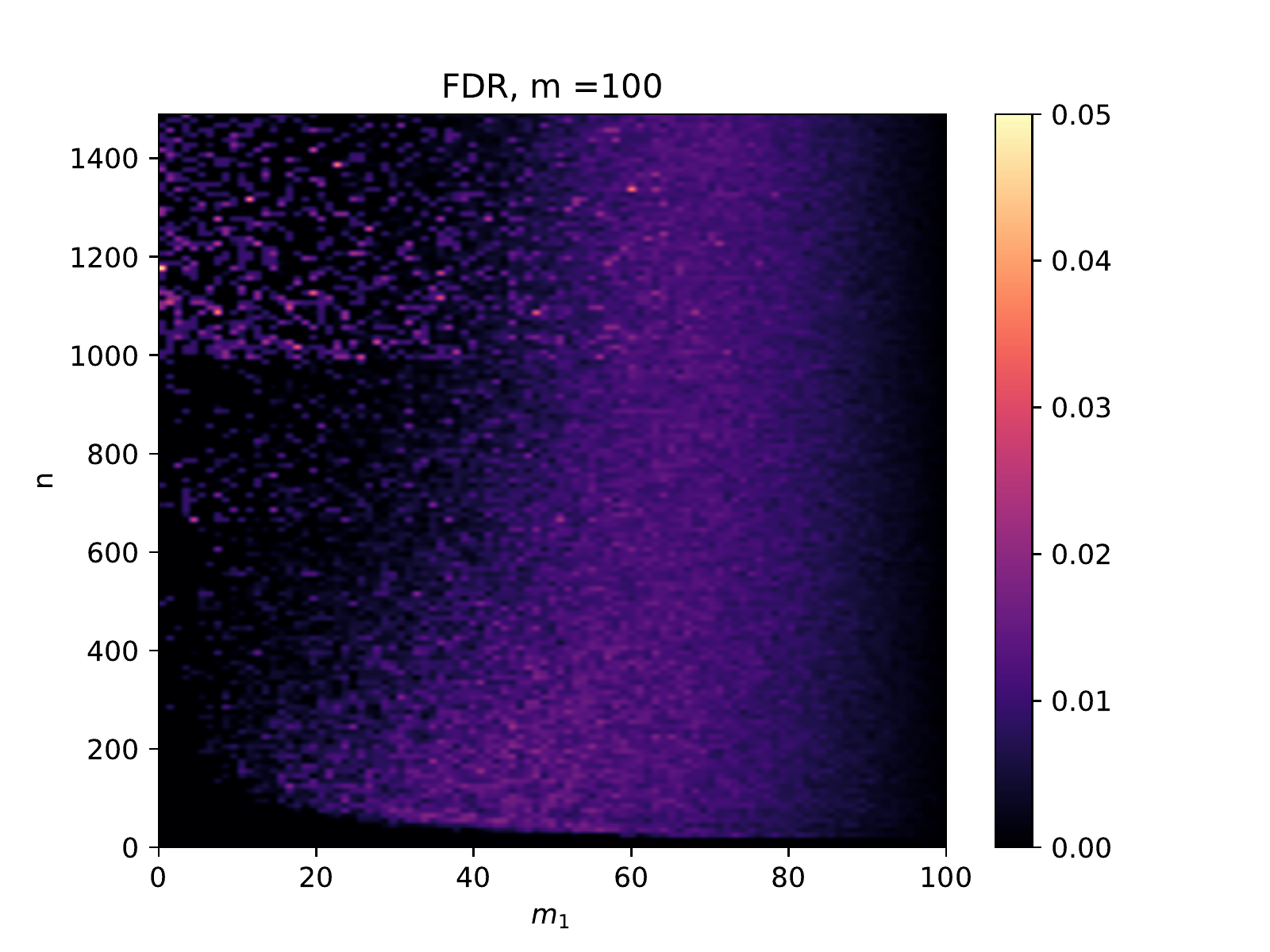}
\includegraphics[scale=0.5]{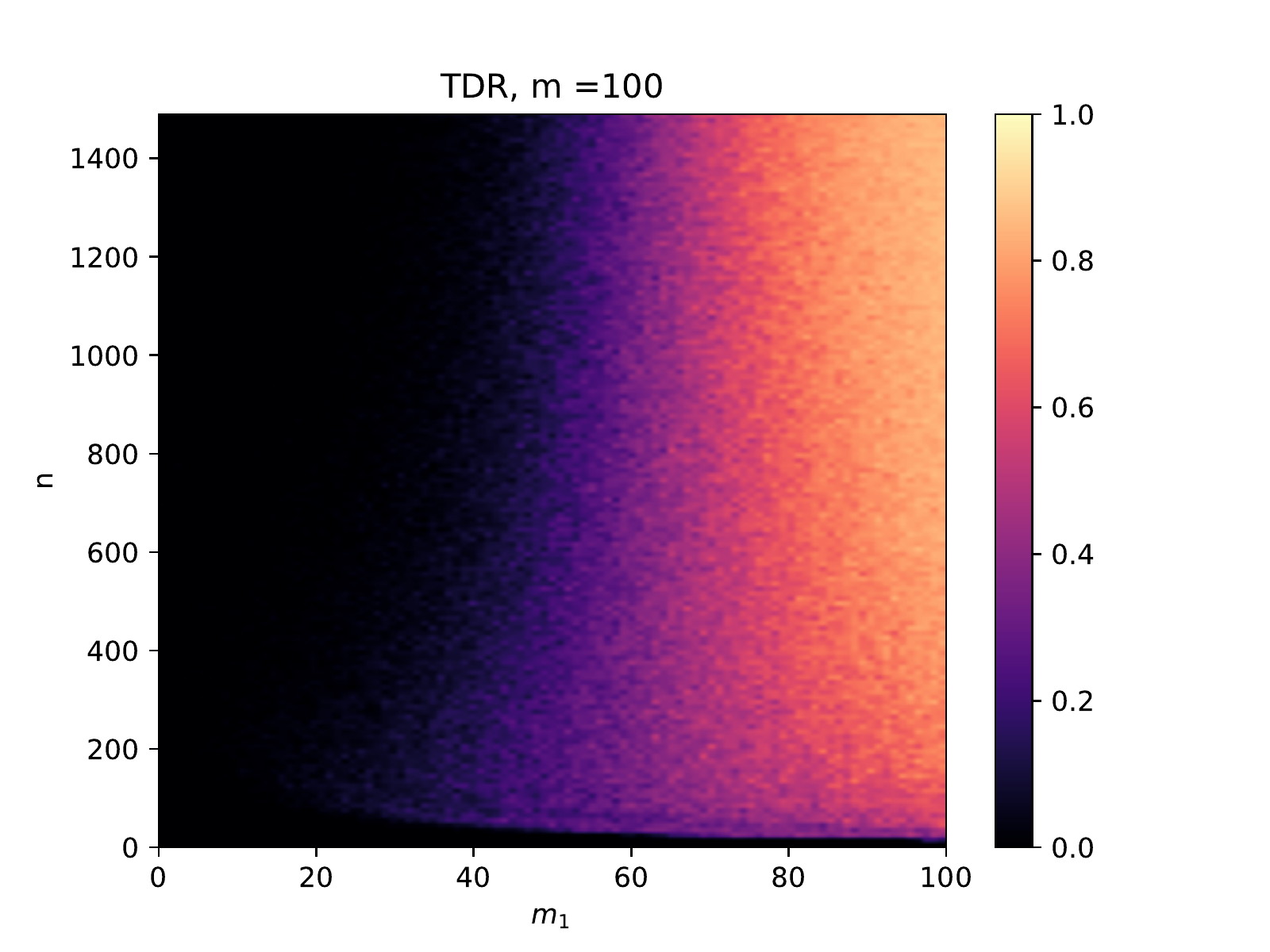}
}
\caption{ FDR (left column) and TDR (right column) for AdaDetect KDE (top row) and Empirical BH (bottom row). In all plots $m=100$ and $\alpha=0.05$.}
\label{app3}
\end{figure}

\end{document}